\newif\ifllncs  %
\newif\iffull   %
\newif\ifsubmit %

\fulltrue
\submittrue

\ifllncs
  \documentclass{llncs}

\else
  \documentclass[letterpaper,11pt]{article}
  
  \usepackage[in]{fullpage}
\fi

\usepackage{iftex}
\ifPDFTeX
  \usepackage[utf8]{inputenc}
  \usepackage[noTeX]{mmap}
  \usepackage[T1]{fontenc}
\fi
\ifLuaTeX
  \usepackage{luatex85}
  \usepackage[noTeX]{mmap}
\fi

\usepackage{amsmath, amsfonts, amsthm, amssymb, amsthm, color}
\usepackage[bookmarks]{hyperref}
\usepackage[nameinlink]{cleveref}

\newtheorem*{rep@theorem}{\rep@title}
\newcommand{\newreptheorem}[2]{%
\newenvironment{rep#1}[1]{%
 \def\rep@title{#2 \ref{##1}}%
 \begin{rep@theorem}}%
 {\end{rep@theorem}}}
\makeatother

\ifllncs\else
  \newtheorem{theorem}{Theorem}[section]
  \newtheorem{definition}[theorem]{Definition}
  \newtheorem{remark}[theorem]{Remark}
  \newtheorem{lemma}[theorem]{Lemma}
  \newtheorem{corollary}[theorem]{Corollary}
  
  \newtheorem{claim}[theorem]{Claim}
  \newtheorem{construction}{Construction}
  \newtheorem{conjecture}[theorem]{Conjecture}
\fi

\newtheorem*{theorem*}{Theorem}
\newtheorem*{construction*}{Construction}
\newtheorem*{lemma*}{Lemma}
\newtheorem*{remark*}{Remark}

\newreptheorem{theorem}{Theorem}
\newreptheorem{lemma}{Lemma}

\usepackage{appendix}
\usepackage{graphicx}
\usepackage{braket}
\usepackage{comment}
\usepackage{url}
\usepackage{algorithmic}
\usepackage[ruled,vlined]{algorithm2e}
\usepackage{mdframed}
\usepackage{authblk}

\usepackage{pgfplots, tikz}
\pgfplotsset{compat=1.14}

\usepackage[style=alphabetic,minalphanames=3,maxalphanames=4,maxnames=99]{biblatex}
\renewcommand*{\labelalphaothers}{\textsuperscript{+}}
\renewcommand*{\multicitedelim}{\addcomma\space}
\ifllncs
\fi

\usepackage{soul, xcolor, xparse}
\makeatletter
  \ExplSyntaxOn
    \cs_new:Npn \white_text:n #1
    {
      \fp_set:Nn \l_tmpa_fp {#1 * .01}
      \llap{\textcolor{white}{\the\SOUL@syllable}\hspace{\fp_to_decimal:N \l_tmpa_fp em}}
      \llap{\textcolor{white}{\the\SOUL@syllable}\hspace{-\fp_to_decimal:N \l_tmpa_fp em}}
    }
    \NewDocumentCommand{\whiten}{ m }
    {
      \int_step_function:nnnN {1}{1}{#1} \white_text:n
    }
  \ExplSyntaxOff
  
  \NewDocumentCommand{ \varul }{ D<>{5} O{0.2ex} O{0.1ex} +m } {%
    \begingroup
    \setul{#2}{#3}%
    \def\SOUL@uleverysyllable{%
      \setbox0=\hbox{\the\SOUL@syllable}%
      \ifdim\dp0>\z@
      \SOUL@ulunderline{\phantom{\the\SOUL@syllable}}%
      \whiten{#1}%
      \llap{%
        \the\SOUL@syllable
        \SOUL@setkern\SOUL@charkern
      }%
      \else
      \SOUL@ulunderline{%
        \the\SOUL@syllable
        \SOUL@setkern\SOUL@charkern
      }%
      \fi}%
    \ul{#4}%
    \endgroup
  }
\makeatother

\mdfdefinestyle{figstyle}{ 
  linecolor=black!7, 
  backgroundcolor=black!7, 
  innertopmargin=10pt, 
  innerleftmargin=25pt, 
  innerrightmargin=25pt, 
  innerbottommargin=10pt 
}

\newenvironment{gamespec}{
  \begin{mdframed}[style=figstyle]}{
  \end{mdframed}}

\newcommand{\E}{\mathop{\mathbb{E}}}
\newcommand{\Z}{\mathbb{Z}}
\newcommand{\R}{\mathbb{R}}

\newcommand{\As}{\mathcal{A}}

\newcommand{\cA}{\mathcal{A}}
\newcommand{\adv}{\mathcal{A}}
\newcommand{\advantage}{{\sf Adv}}

\newcommand{\Advz}{\cA_0}
\newcommand{\Bs}{\mathcal{B}}
\newcommand{\advB}{\mathcal{B}}

\newcommand{\advC}{\mathcal{C}}
\newcommand{\Cs}{\mathcal{C}}
\newcommand{\Ds}{\mathcal{D}}

\newcommand{\Ts}{\mathcal{T}}
\newcommand{\Rs}{\mathcal{R}}
\newcommand{\Xs}{\mathcal{X}}
\newcommand{\Ys}{\mathcal{Y}}
\newcommand{\Zs}{\mathcal{Z}}

\newcommand{\Sim}{{\sf Sim}}
\newcommand{\cM}{\mathcal{M}}

\newcommand{\Hs}{\mathcal{H}}

\newcommand{\LF}{{\sf LF}}

\newcommand{\ccobf}{{\sf CC.Obf}}
\newcommand{\obf}{{\sf Obf}}
\newcommand{\iO}{{\sf iO}}

\newcommand{\owf}{{\sf OWF}}

\newcommand{\prf}{{\sf PRF}}

\newcommand{\enc}{{\sf Enc}}
\newcommand{\dec}{{\sf Dec}}
\newcommand{\samp}{{\sf Samp}}

\newcommand{\sign}{{\sf Sign}}

\newcommand{\aux}{{\sf aux}}
\newcommand{\AUX}{{\sf AUX}}

\newcommand{\shO}{{\sf shO}}

\newcommand{\kgen}{\mathsf{KeyGen}}

\newcommand{\xv}{{\mathbf{x}}}
\newcommand{\fv}{{\mathbf{f}}}
\newcommand{\hv}{{\mathbf{h}}}
\newcommand{\bv}{{\mathbf{b}}}
\newcommand{\Rm}{{\mathbf{R}}}

\newcommand{\can}{{\sf Can}}

\newcommand{\sk}{{\sf sk}}

\newcommand{\pk}{{\sf pk}}

\newcommand{\tk}{\mathsf{tk}}

\newcommand{\TS}{\mathsf{TS}}

\newcommand{\eval}{\mathsf{Eval}}

\newcommand{\we}{\mathsf{WE}}

\newcommand{\CC}{\mathsf{CC}}

\newcommand{\obfCC}{\mathsf{\widetilde{CC}}}

\newcommand{\hatCC}{\mathsf{\widehat{CC}}}

\newcommand{\np}{\textsf{NP}}

\newcommand{\keygen}{\mathsf{QKeyGen}}

\newcommand{\setup}{\mathsf{Setup}}

\newcommand{\qsk}{\rho_\sk}
\newcommand{\ct}{\mathsf{ct}}

\newcommand{\AS}{\mathsf{CS}}

\newcommand{\iag}{\mathsf{IAG}}

\newcommand{\tokengen}{\mathsf{TokenGen}}

\newcommand{\qtoken}{\mathsf{\ket{tk}}}

\newcommand{\sig}{\mathsf{sig}}

\newcommand{\verify}{\mathsf{Verify}}

\newcommand{\revoke}{\mathsf{Revoke}}

\newcommand{\gentrigger}{\mathsf{GenTrigger}}

\newcommand{\puncture}{\mathsf{Puncture}}

\newcommand{\strongantipiracy}{\mathsf{StrongAntiPiracy}}
\def\P{\textsf{P}} %

\newcommand{\D}{\mathsf{D}}

\newcommand{\hatP}{\hat{\P}}

\newcommand{\param}{\lambda}

\newcommand{\rand}{{\sf rand}}

\newcommand{\cD}{{\mathcal{D}}}
\newcommand{\cN}{{\mathcal{N}}}
\newcommand{\cE}{{\mathcal{E}}}
\newcommand{\cR}{{\mathcal{R}}}
\newcommand{\shift}{{\sf Shift}}
\newcommand{\api}{{\sf API}}
\newcommand{\ati}{{\sf ATI}}

\newcommand{\ti}{{\sf TI}}
\newcommand{\projimp}{{\sf ProjImp}}
\newcommand{\cproj}{{\sf CProj}}
\newcommand{\ag}{{\sf AG}}
\newcommand{\cP}{{\mathcal{P}}}

\newcommand{\sdenc}{{\sf SDEnc}}

\newcommand{\general}{{\sf general}}

\newcommand{\lf}{{\sf lf}}
\newcommand{\lossy}{{\sf lossy}}
\newcommand{\inj}{{\sf inj}}
\newcommand{\compuncertain}{direct product hardness}
\newcommand{\CompUncertain}{Direct Product Hardness}

\newcommand{\revise}[1]{#1}

\newcommand{\anote}[1]{}
\newcommand{\jiahui}[1]{}
\newcommand{\qipeng}[1]{}

\newcommand{\F}{\mathbb{F}}
\newcommand{\N}{\mathbb{N}}

\newcommand{\poly}{{\sf poly}}
\newcommand{\negl}{{\sf negl}}
\newcommand{\subexp}{{\sf subexp}}

\DeclareMathOperator{\Tr}{Tr}

\newcommand{\unpredictableassumptions}{Assuming the existence of post-quantum $\iO$, one-way functions, compute-and-compare obfuscation for the class of unpredictable distributions (as in \Cref{def: cc obf}), and the \revise{strong monogamy-of-entanglement property} (\Cref{conj:strong_monogamy_it})}

\newcommand{\subexpassumptions}{Similarly, assuming the existence of post-quantum sub-exponentially secure $\iO$ and one-way functions, the quantum hardness of LWE and assuming the \revise{strong monogamy-of-entanglement property} (\Cref{conj:strong_monogamy_it})}

\begin{document}

\pagenumbering{gobble} %

\title{Hidden Cosets and Applications \\to Unclonable Cryptography}

\author[1]{Andrea Coladangelo}
\author[2]{Jiahui Liu}
\author[3]{Qipeng Liu}
\author[4]{Mark Zhandry}
\affil[1]{UC Berkeley,  Simons Institute for the Theory of Computing \& qBraid}
\affil[2]{University of Texas at Austin}
\affil[3]{Princeton University}
\affil[4]{Princeton University \& NTT Research}
\date{\today}

\maketitle

\begin{abstract}
    In 2012, Aaronson and Christiano introduced the idea of \emph{hidden subspace states} to build public-key quantum money [STOC '12]. Since then, this idea has been applied to realize several other cryptographic primitives which enjoy some form of unclonability.
    
    In this work, we study a generalization of hidden subspace states to hidden \emph{coset} states. This notion was considered independently by Vidick and Zhang [Eurocrypt '21], in the context of proofs of quantum knowledge from quantum money schemes. We explore unclonable properties of coset states and several applications:
    \begin{itemize}
        \item We show that, assuming indistinguishability obfuscation ($\iO$), hidden coset states possess a certain \emph{direct product hardness} property, which immediately implies a tokenized signature scheme in the plain model. Previously, a tokenized signature scheme was known only relative to an oracle, from a work of Ben-David and Sattath [QCrypt '17].
        
        \item Combining a tokenized signature scheme with extractable witness encryption, we give a construction of an unclonable decryption scheme in the plain model. The latter primitive was recently proposed by Georgiou and Zhandry [ePrint '20], who gave a construction relative to a classical oracle.
        
        \item We conjecture that coset states satisfy a certain natural (information-theoretic) monogamy-of-entanglement property. Assuming this conjecture is true, we remove the requirement for extractable witness encryption in our unclonable decryption construction, by relying instead on compute-and-compare obfuscation for the class of unpredictable distributions. This conjecture was later proved by Culf and Vidick in a follow-up work. 

        \item %
        Finally, we give a construction of a copy-protection scheme for pseudorandom functions (PRFs) in the plain model. Our scheme is secure either assuming $\iO$, $\owf$ and extractable witness encryption, or assuming $\iO, \owf$, compute-and-compare obfuscation for the class of unpredictable distributions, and the strong %
        monogamy property mentioned above. This is the first example of a copy-protection scheme with provable security in the plain model for a class of functions that is not evasive.
    \end{itemize}
\end{abstract}

\newpage

\pagenumbering{arabic}

\iffull
{
  \hypersetup{linkcolor=Violet}
  \setcounter{tocdepth}{2}
  \tableofcontents
}
\fi

\section{Introduction}
The no-cloning principle of quantum mechanics asserts that quantum information cannot be generically copied. This principle has profound consequences in quantum cryptography, as it puts a fundamental restriction on the possible strategies that a malicious party can implement. One of these consequences is that quantum information enables cryptographic tasks that are provably impossible to realize classically, the most famous example being information-theoretically secure key distribution~\cite{BenBra84}.

Beyond this, the no-cloning principle opens up an exciting avenue to realize cryptographic tasks which enjoy some form of \emph{unclonability}, e.g. quantum money~\cite{wiesner1983conjugate, aaronson2012quantum, farhi2012quantum, zhandry2019quantum, kane2018quantum}, quantum tokens for digital signatures~\cite{ben2016quantum}, copy-protection of programs \cite{aaronson2009quantum, aaronsonnew, coladangelo2020quantum}, and more recently unclonable encryption \cite{gottesman2002uncloneable, Broadbent2019UncloneableQE} and decryption \cite{georgiou-zhandry20}.

In this work, we revisit the \emph{hidden subspace} %
idea proposed by Aaronson and Christiano, which has been employed towards several of the applications above. We propose a generalization of this idea, which involves hidden \emph{cosets} (affine subspaces), and we show applications of this to signature tokens, unclonable decryption and copy-protection.

Given a subspace $A \subseteq \mathbb{F}_2^n$, the corresponding \emph{subspace state} is defined as a uniform superposition over all strings in the subspace $A$, i.e. $$\ket{A} := \frac{1}{\sqrt{|A|}}\sum_{x \in A} \ket{x}\,,$$
The first property that makes this state useful is that applying a Hadamard on all qubits creates a uniform superposition over all strings in $A^{\perp}$, the orthogonal complement of $A$, i.e. $H^{\otimes n} \ket{A} = \ket{A^{\perp}}$.

The second property, which is crucial for constructing unclonable primitives with some form of verification, is the following. Given one copy of $\ket{A}$, where $A \subseteq{F}_2^n$ is uniformly random of dimension $n/2$, it is impossible to produce two copies of $\ket{A}$ except with negligible probability. As shown by~\cite{aaronson2012quantum}, unclonability holds even when given quantum access to oracles for membership in $A$ and $A^{\perp}$, as long as the number of queries is polynomially bounded. On the other hand, such membership oracles allow for verifying the state $\ket{A}$, leading to publicly-verifiable quantum money, where the verification procedure is the following:
\begin{itemize}
    \item Given an alleged quantum money state $\ket{\psi}$, query the oracle for membership in $A$ on input $\ket{\psi}$. Measure the outcome register, and verify that the outcome is $1$.
    \item If so, apply $H^{\otimes n}$ to the query register, and query the oracle for membership in $A^{\perp}$. Measure the outcome register, and accept the money state if the outcome is $1$.
\end{itemize}

It is not difficult to see that the unique state that passes this verification procedure is $\ket{A}$.

In order to obtain a quantum money scheme in the plain model (without oracles), Aaronson and Christiano suggest instantiating the oracles with some form of program  obfuscation. This vision is realized subsequently in \cite{zhandry2019quantum}, where access to the oracles for subspace membership is replaced by a suitable obfuscation of the membership programs, which can be built from indistinguishability obfuscation (\textsf{iO}). More precisely, Zhandry shows that, letting $P_A$ and $P_{A^{\perp}}$ be programs that check membership in $A$ and $A^{\perp}$ respectively, any computationally bounded adversary who receives a uniformly random subspace state $\ket{A}$ together with $\iO(P_A)$ and $\iO(P_{A^{\perp}})$ cannot produce two copies of $\ket{A}$ except with negligible probability.

The subspace state idea was later employed to obtain \emph{quantum tokens} for digital signatures \cite{ben2016quantum}. What these are is best explained by the (award-winning) infographic in \cite{ben2016quantum} (see the ancillary arXiv files there). Concisely, a quantum signature token allows Alice to provide Bob with the ability to sign \emph{one and only one} message in her name, where such signature can be publicly verified using Alice's public key. The construction of quantum tokens for digital signatures from \cite{ben2016quantum} is the following.
\begin{itemize}
    \item Alice samples a uniformly random subspace $A \subseteq \mathbb{F}_2^n$, which constitutes her secret key. A signature token is the state $\ket{A}$.
    \item Anyone in possession of a token $\ket{A}$ can sign message $0$ by outputting a string $v \in A$ (this can be obtained by measuring $\ket{A}$ in the computational basis), and can sign message $1$ by outputting a string $w \in A^{\perp}$ (this can be done by measuring $\ket{A}$ in the Hadamard basis).
    \item Signatures can be publicly verified assuming quantum access to an oracle for subspace membership in $A$ and in $A^{\perp}$ (such access can be thought of as Alice's public key).
\end{itemize}
In order to guarantee security of the scheme, i.e. that Bob cannot produce a valid signature for more than one message, Ben-David and Sattath prove the following strengthening of the original property proven by Aaronson and Christiano. Namely, they show that any query-bounded adversary with quantum access to oracles for membership in $A$ and $A^{\perp}$ cannot produce, except with negligible probability, a pair $(v,w)$ where $v \in A\setminus \{0\}$ and $w \in A^{\perp}\setminus\{0\}$. We refer to this property as a \emph{direct product hardness} property.

The natural step to obtain a signature token scheme in the plain model is to instantiate the subspace membership oracles using $\iO$, analogously to the quantum money application. However, unlike for the case of quantum money, here one runs into a technical barrier, which we expand upon in Section \ref{sec: tech ovw direct product hardness}. 
Thus, a signature token scheme is not known in the plain model, and this has remained an open question since \cite{ben2016quantum}.

In general, a similar difficulty in obtaining schemes that are secure in the plain model as opposed to an oracle model seems prevalent in works about other unclonable primitives. For example, in the case of copy-protection of programs, we know that copy-protection of a large class of evasive programs, namely compute-and-compare programs, is possible with provable non-trivial security against fully malicious adversaries in the quantum random oracle model (QROM) \cite{coladangelo2020quantum}. Other results achieving provable security in the plain model are secure only against a restricted class of adversaries \cite{ananth2020secure, kitagawa-takashi2020ssl, broadbent2021secure}. To make the contrast between plain model and oracle model even more stark, all unlearnable programs can be copy-protected assuming access to (highly structured) oracles \cite{aaronsonnew}, but we know, on the other hand, that a copy-protection scheme for all unlearnable programs in the plain model does not exist (assuming Learning With Errors is hard for quantum computers) \cite{ananth2020secure}.

Likewise, for the recently proposed task of unclonable decryption, the only currently known scheme is secure only in a model with access to subspace membership oracles \cite{georgiou-zhandry20}.

\subsection{Our Results}
\label{sec: results}

We study a generalization of subspace states, which we refer to as \emph{coset} states. This notion has also been studied independently in a work of Vidick and Zhang \cite{vidick2021classical}, in the context of proofs of quantum knowledge from quantum money schemes. 

For $A\subseteq \mathbb{F}_2^n$, and $s,s' \in \mathbb{F}_2^n$, the corresponding coset state is:

$$ \ket{A_{s,s'}} := \sum_{x \in A} (-1)^{\langle x , s' \rangle} \ket{x+s} \,,$$

where here $\langle x , s' \rangle$ denotes the inner product of $x$ and $s'$. In the computational basis, the quantum state is a superposition over all elements in the coset $A+s$, while, in the Hadamard basis, it is a superposition over all elements in  $A^\perp+s'$. Let $P_{A+s}$ and $P_{A^\perp + s'}$ be programs that check membership in the cosets $A + s$ and $A^\perp + s'$ respectively. To check if a state $\ket{\psi}$ is a coset state with respect to $A, s,s'$, one can compute $P_{A+s}$ in the computational basis, and check that the outcome is $1$; then, apply $H^{\otimes n}$ followed by
$P_{A^\perp + s'}$, and check that the outcome is $1$.

\paragraph{Computational Direct Product Hardness.} Our first technical result is establishing a \emph{computational direct product hardness} property in the plain model, assuming post-quantum $\mathsf{iO}$ and one-way functions.

\begin{theorem}[Informal]
Any quantum polynomial-time adversary who receives $\ket{A_{s,s'}}$ and programs $\iO(P_{A+s})$ and $\iO(P_{A^\perp+s'})$ for uniformly random $A \subseteq{\mathbb{F}_2^n}$, $s,s' \in \mathbb{F}_2^n$, cannot produce a pair $(v,w) \in (A+s) \times (A^{\perp} + s')$, except with negligible probability in $n$.

\end{theorem}

As we mentioned earlier, this is in contrast to regular subspace states, for which a similar direct product hardness
is currently not known in the plain model, but only in a model with access to subspace membership oracles.

We then apply this property to obtain the following primitives.

\paragraph{Signature Tokens.} Our direct product hardness immediately implies a \emph{signature token} scheme in the plain model (from  post-quantum $\iO$ and one-way functions), thus resolving the main question left open in \cite{ben2016quantum}.

\begin{theorem}[Informal]
Assuming post-quantum $\iO$ and one-way functions, there exists a signature token scheme. 
\end{theorem}
In this signature token scheme, the public verification key is the pair $(\iO(P_{A+s}), \iO(P_{A^\perp+s'}))$, and a signature token is the coset state $\ket {A_{s, s'}}$. Producing signatures for both messages $0$ and $1$ is equivalent to finding elements in both $A+s$ and $A^{\perp} + s'$, which violates our computational direct product hardness property.

\paragraph{Unclonable Decryption.} Unclonable decryption, also known as \emph{single-decryptor encryption}, was introduced in \cite{georgiou-zhandry20}. 
Informally, a single-decryptor encryption scheme is a (public-key) encryption scheme in which the secret key is a \textit{quantum state}. The scheme satisfies a standard notion of security (in our case, CPA security), as well as the following additional security guarantee: no efficient quantum algorithm with one decryption key is able to produce two working decryption keys. 
We build a single-decryptor encryption scheme using a signature tokens scheme and extractable witness encryption in a black-box way. By leveraging our previous result about the existence of a signature token scheme in the plain model, we are able to prove security without the need for the structured oracles used in the original construction of \cite{georgiou-zhandry20}. 

\begin{theorem}[Informal]
Assuming post-quantum $\iO$, one-way functions, and extractable witness encryption, there exists a public-key single-decryptor encryption scheme.
\end{theorem}

\paragraph{Copy-protection of PRFs.} The notion of a copy-protection scheme was introduced by Aaronson in \cite{aaronson2009quantum} and recently explored further in \cite{ananth2020secure, coladangelo2020quantum, aaronsonnew, broadbent2021secure}.

In a copy-protection scheme, the vendor of a classical program wishes to provide a user the ability to run the program on any input, while ensuring that the functionality cannot be ``pirated'': informally, the adversary, given one copy of the program, cannot produce two programs that enable evaluating the program correctly.

Copy-protection is trivially impossible classically, since classical information can always be copied. This impossibility can be in principle circumvented if the classical program is encoded in a quantum state, due to the no-cloning principle. However, positive results have so far been limited. A copy-protection scheme \cite{coladangelo2020quantum} is known for a class of evasive programs, known as compute-and-compare programs, with provable non-trivial security against fully malicious adversaries in the Quantum Random Oracle Model (QROM). Other schemes in the plain model are only secure against restricted classes of adversaries (which behave honestly in certain parts of the protocol) \cite{ananth2020secure, kitagawa-takashi2020ssl,   broadbent2021secure}. Copy-protection schemes for more general functionalities are known \cite{aaronsonnew}, but these are only secure assuming very structured oracles (which depend on the functionality that is being copy-protected).

In this work, we present a copy-protection scheme for a family of pseudorandom functions (PRFs). In such a scheme, for any classical key $K$ for the PRF, anyone in possession of a \textit{quantum} key $\rho_K$ is able to evaluate $PRF(K, x)$ on any input $x$. 

The copy-protection property that our scheme satisfies is that given a quantum key $\rho_K$, no efficient algorithm can produce two (possibly entangled) keys such that these two keys allow for simultaneous correct evaluation on uniformly random inputs, with noticeable probability.

Similarly to the unclonable decryption scheme, our copy-protection scheme is secure assuming post-quantum $\iO$, one-way functions, and extractable witness encryption. 

\begin{theorem}[Informal]
Assuming post-quantum $\iO$, one-way functions, and extractable witness encryption, there exists a copy-protection scheme for a family of PRFs.
\end{theorem}

We remark that our scheme requires a particular kind of PRFs, namely puncturing and extracting with small enough error (we refer to Section \ref{sec: PRF prelim} for precise definitions). However, PRFs satisfying these properties can be built from just one-way functions.
The existence of extractable witness encryption is considered to be a very strong assumption. In particular, it was shown to be impossible in general (under a special-purpose obfuscation conjecture) \cite{garg2017implausibility}. However, we emphasize that no provably secure copy-protection schemes with standard malicious security in the plain model are known at all. %
Given the central role of PRFs in the construction of many other cryptographic primitives, we expect that our copy-protection scheme, and the techniques developed along the way, will play an important role as a building block to realize \textit{unclonable} versions of other primitives.

\vspace{2mm}
To avoid the use of extractable witness encryption, we put forth a (information-theoretic) conjecture about a \emph{monogamy of entanglement} property of coset states, which we discuss below
\footnote{This conjecture is proved true in the follow-up work by Culf and Vidick \cite{culfvidick2021cosetsproof} after the first version of this paper.}.

Assuming this conjecture is true, we show that both unclonable decryption and copy-protection of PRFs can be constructed \emph{without} extractable witness encryption, by relying instead on compute-and-compare obfuscation \cite{wichs2017obfuscating, goyal2017lockable} (more details on the latter can be found in Section \ref{sec:cc}).
\begin{theorem}[Informal]
Assuming post-quantum $\iO$, one-way functions, and obfuscation of compute-and-compare programs against unpredictable distributions, there exist: (i) a public-key single-decryptor encryption scheme, and (ii) a copy-protection scheme for a family of PRFs.
\end{theorem}

As potential evidence in support of the monogamy-of-entanglement conjecture, we prove a weaker version of the monogamy of entanglement property, which we believe will still be of independent interest (more details on this are below).

\begin{remark}While $\iO$ was recently constructed based on widely-believed  computational assumptions~\cite{JLS20-io-wellfounded}, the latter construction is not quantum resistant, and the situation is less clear quantumly. However, several works have proposed candidate post-quantum obfuscation schemes ~\cite{BGMZ18, WeeWichs20,BDGM20}, and based on these works $\iO$ seems plausible in the post-quantum setting as well.
\end{remark}

\begin{remark}
Compute-and-compare obfuscation against unpredictable distributions is known to exist assuming LWE (or $\iO$) and assuming the existence of Extremely Lossy Functions (ELFs) \cite{zhandry2019magic} \cite{wichs2017obfuscating, goyal2017lockable}. Unfortunately, the only known constructions of ELFs rely on hardness assumptions that are broken by quantum computers (exponential hardness of decisional Diffie-Hellman). To remedy this, we give a construction of computate-and-compare obfuscation against \emph{sub-exponentially} unpredictable distributions, from plain LWE (see Theorem \ref{thm: cc sub-exp from lwe}, and its proof in \Cref{sec:CC_quantum_aux}). The latter weaker obfuscation is sufficient to prove security of our single-decryptor encryption scheme, and copy-protection scheme for PRFs, if one additionally assumes \emph{sub-exponentially} secure $\iO$ and one-way functions.
\end{remark}

\paragraph{Monogamy-of-Entanglement.}
As previously mentioned, we conjecture that coset states additionally satisfy a certain (information-theoretic) \emph{monogamy of entanglement} property, similar to the one satisfied by BB84 states, which is studied extensively in \cite{tomamichel2013monogamy}. Unlike the monogamy property of BB84 states, the monogamy property we put forth is well-suited for applications with public verification, in a sense made more precise below.

This monogamy property states that Alice, Bob and Charlie cannot cooperatively win the following game with a challenger, except with negligible probability. The challenger first prepares a uniformly random coset state $\ket {A_{s, s'}}$ and gives the state to Alice. Alice outputs two (possibly entangled) quantum states and sends them to Bob and Charlie respectively. Finally, Bob and Charlie both get the description of the subspace $A$. The game is won if Bob outputs a vector in $A+s$ and Charlie outputs a vector in  $A^\perp+s'$.

Notice that if Alice were told $A$ before she had to send the quantum states to Bob and Charlie, then she could recover $s$ and $s'$ (efficiently) given $\ket{A_{s,s'}}$. Crucially, $A$ is  only revealed to Bob and Charlie \emph{after} Alice has sent them the quantum states (analogously to the usual monogamy-of-entanglement game based on BB84 states, where $\theta$ is only revealed to Bob and Charlie after they receive their states from Alice.).

We note that the hardness of this game is an \emph{information-theoretic} conjecture. As such, there is hope that it can be proven unconditionally. 

Under this conjecture, we show that the problem remains hard (computationally) even if Alice additionally receives the programs $\iO(P_{A+s})$ and $\iO(P_{A^{\perp}+s'})$. Based on this result, we then obtain unclonable decryption and copy-protection of PRFs from post-quantum $\iO$ and one-way functions, and compute-and-compare obfuscation against unpredictable distributions. We thus remove the need for extractable witness encryption (more details on this are provided in the technical overview, Section \ref{sec: tech ovw direct product hardness}).

As evidence in support of our conjecture, we prove a weaker information-theoretic monogamy property, namely that Alice, Bob and Charlie cannot win at a monogamy game that is identical to the one described above, except that at the last step, Bob and Charlie are each required to return a pair in $(A+s) \times ( A^{\perp} + s')$, instead of a single element each.  Since coset states have more algebraic structure than BB84 states, a more refined analysis is required to prove this (weaker) property compared to that of \cite{tomamichel2013monogamy}. We again extend this monogamy result to the case where Alice receives programs $\iO(P_{A+s})$ and $\iO(P_{A^{\perp}+s'})$.

We emphasize that our monogamy result for coset states differs from the similar monogamy result for BB84 states in one crucial way: the result still holds when Alice receives programs that allow her to verify the correctness of her state (namely $\iO(P_{A+s})$ and $\iO(P_{A^{\perp}+s'})$). This is not the case for the BB84 monogamy result. In fact,
Lutomirski \cite{lutomirski2010online} showed that an adversary who is given $\ket {x^\theta}$ and a public verification oracle that outputs $1$ if the input state is correct and $0$ otherwise, can efficiently copy the state $\ket {x^\theta}$. At the core of this difference is the fact that coset states are highly entangled, whereas strings of BB84 states have no entanglement at all.

For this reason, we believe that the monogamy property of coset states may be of independent interest, and may find application in contexts where public verification of states is important. %

\paragraph{Proof for the Strong Monogamy-of-Entanglement Conjecture.}
After the first version of this paper, Vidick and Culf posted a follow-up paper \cite{culfvidick2021cosetsproof} that proved the strong monogamy-of-entanglement conjecture stated above (formalized in \Cref{sec: monogamy conjectured}). We thank Vidick and Culf for following up on our work. 

The readers can therefore consider the ``strong monogamy-of-entanglement conjecture'' removed from the assumptions in all formal statements in this paper.

\paragraph{Acknowledgements}
A.C. and Q.L. were Quantum Postdoctoral Fellows at the Simons Institute for the Theory of Computing supported by NSF QLCI Grant No. 2016245. A.C. and Q.L. were also supported by DARPA under agreement No. HR00112020023. J. L. and M. Z. were supported by the NSF. J. L. was also supported by Scott Aaronson's Simons Investigator award. The authors are grateful for the support of the Simons Institute, where this collaboration was initiated. Any opinions, findings and conclusions or recommendations expressed in this material are those of the author(s) and do not necessarily reflect the views of the United States Government or DARPA.

\section{Technical Overview}
\label{sec: tech overview}

\subsection{Computational Direct Product Hardness for Coset States}
\label{sec: tech ovw direct product hardness}
Our first technical contribution is to establish a \textit{computational} direct product hardness property for coset states. In this section, we aim to give some intuition for the barrier to proving such a property for regular subspace states, and why resorting to coset states helps. 

We establish the following: a computationally bounded adversary who receives $\ket{A_{s,s'}}$ and programs $\iO(P_{A+s})$ and $\iO(P_{A^\perp+s'})$ for uniformly random $A, s,s'$, cannot produce a pair $(v,w)$, where $v \in A +s $ and $w \in A^{\perp} + s'$, except with negligible probability.

The first version of this direct product hardness property involved regular subspace states, and was \textit{information-theoretic}. It was proven by Ben-David and Sattath \cite{ben2016quantum}, and it established the following: given a uniformly random subspace state $\ket{A}$, where $A \subseteq \mathbb{F}_2^n$ has dimension $n/2$, no adversary can produce a pair of vectors $v,w$ such that $v \in A$ and $w \in A^\perp$ respectively, even with access to oracles for membership in $A$ and in $A^\perp$. 

The first successful instantiation of the membership oracles in the plain model is due to Zhandry, in the context of public-key quantum money \cite{zhandry2019quantum}. Zhandry showed that replacing the membership oracles with indistinguishability obfuscations of the membership programs $P_A$ and $P_{A^{\perp}}$ is sufficient to prevent an adversary from copying the subspace state, and thus is sufficient for public-key quantum money. In what follows, we provide some intuition as to how one proves this ``computational no-cloning'' property, and why the same proof idea does not extend naturally to the direct product hardness property for regular subspace states.

In \cite{zhandry2019quantum}, Zhandry shows that $\iO$ realizes what he refers to as a \textit{subspace-hiding obfuscator}. A subspace hiding obfuscator $\shO$ has the property that any computationally bounded adversary who chooses a subspace $A$ cannot distinguish between $\shO(P_A)$ and $\shO(P_B)$ for a uniformly random superspace $B$ of $A$ (of not too large dimension). In turn, a subspace hiding obfuscator can then be used to show that an adversary who receives $\ket{A}$, $\shO(P_A)$ and $\shO(P_{A^{\perp}})$, for a uniformly random $A$, cannot produce two copies of $\ket{A}$. This is done in the following way. For the rest of the section, we assume that $A \subseteq \mathbb{F}_2^n$ has dimension $n/2$.
\begin{itemize}
    \item Replace $\shO(P_A)$ with $\shO(P_B)$ for a uniformly random superspace $B$ of $A$, where $\text{dim}(B) = \frac34 n$. Replace $\shO(P_{A^{\perp}})$ with $\shO(P_C)$ for a uniformly random superspace $C$ of $A^{\perp}$, where $\text{dim}(C) = \frac{3}{4}n$.
    \item Argue that the task of copying a subspace state $\ket{A}$, for a uniformly random subspace $C^{\perp} \subseteq A \subseteq B$ (even knowing $B$ and $C$ directly) is just as hard as the task of copying a uniformly random subspace state of dimension $\ket{A'} \subseteq \mathbb{F}_2^{n/2}$ where $\text{dim}(A') = \frac{n}{4}$. The intuition for this is that knowing $C^{\perp}$ fixes $\frac{n}{4}$ dimensions out of the $\frac{n}{2}$ original dimensions of $A$. Then, you can think of the first copying task as equivalent to the second up to a change of basis. 
    Such reduction completely removes the adversary's knowledge about the membership programs. 
    \item The latter task is of course hard (it would even be hard with access to membership oracles for $A'$ and $A'^{\perp}$).
\end{itemize}

One can try to apply the same idea to prove a \emph{computational direct product hardness property} for subspace states, where the task is no longer to copy $\ket{A}$, but rather we wish to show that a bounded adversary receiving $\ket{A}$ and programs $\iO(P_{A})$ and $\iO(P_{A^\perp})$, for uniformly random $A$, cannot produce a pair $(v,w)$, where $v \in A$ and $w \in A^{\perp}$. Applying the same replacements as above using $\shO$ allows us to reduce this task to the task of finding a pair of vectors in $A \times A^{\perp}$ given $\ket{A}$,$B,C$, such that $C^{\perp} \subseteq A \subseteq B$. Unfortunately, unlike in the case of copying, this task is easy, because any pair of vectors in $C^{\perp} \times B^{\perp}$ also belongs to $A \times A^{\perp}$. This is the technical hurdle that ones runs into when trying to apply the proof idea from \cite{zhandry2019quantum} to obtain a computational direct hardness property for subspace states.

Our first result is that we overcome this hurdle by using coset states. In the case of cosets, the natural analog of the argument above results in a replacement of the program that checks membership in $A+s$ with a program that checks membership in $B+s$. Similarly, we replace $A^{\perp}+s'$ with $C+s'$. The crucial observation is that, since $B+s = B+s+t$ for any $t \in B$, the programs $P_{B+s}$ and $P_{B+s+t}$ are functionally equivalent. So, an adversary who receives $\iO(P_{B+s})$ cannot distinguish this from $\iO(P_{B+s+t})$ for any $t$. We can thus argue that $t$ functions as a randomizing mask that prevents the adversary from guessing $s$ and finding a vector in $A+s$.

\paragraph{\textbf{Signature Tokens.}}
The computational direct product hardness immediately gives a signature token scheme in the plain model:
\begin{itemize}
    \item Alice samples a key $(A,s,s')$ uniformly at random. This constitutes her secret key. The verification key is $(\iO(P_{A+s}), \iO(P_{A^{\perp}+s'}))$. A signature token is $\ket{A_{s,s'}}$.
    \item Anyone in possession of a token  can sign message $0$ by outputting a string $v \in A+s$ (this can be obtained by measuring the token in the computational basis), and can sign message $1$ by outputting a string $w \in A^{\perp}+s'$ (this can be done by measuring the token in the Hadamard basis).
    \item Signatures can be publicly verified using Alice's public key.
\end{itemize}
If an algorithm produces both signatures for messages $0$ and $1$, it finds vectors $v \in A + s$ and $w \in A^\perp + s'$, which violates computational direct product hardness.

\subsection{Unclonable Decryption} Our second result is an \emph{unclonable decryption} scheme (also known as a \emph{single-decryptor encryption} scheme \cite{georgiou-zhandry20} - we will use the two terms interchangeably in the rest of the paper) from black-box use of a signature token scheme and extractable witness encryption. This construction removes the need for structured oracles, as used in the construction of \cite{georgiou-zhandry20}.

Additionally, we show that, assuming the conjectured monogamy property described in Section \ref{sec: results}, we obtain an unclonable decryption scheme from just $\iO$ and post-quantum one-way functions, where $\iO$ is used to construct obfuscators for both subspace-membership programs and compute-and-compare programs~\cite{goyal2017lockable,wichs2017obfuscating}.

In this overview, we focus on the construction from the monogamy property, as we think it is conceptually more interesting. 

Recall that a single-decryptor encryption scheme is a public-key encryption scheme in which the secret key is a quantum state. On top of the usual encryption security notions, one can define ``single-decryptor'' security: this requires that it is not possible for an adversary who is given the secret key to produce two (possibly entangled) decryption keys, which both enable simultaneous successful decryption of ciphertexts. A simplified version of our single-decryptor encryption scheme is the following. Let $n \in \mathbb{N}$.
\begin{itemize}
    \item The key generation procedure samples uniformly at random $A \subseteq \mathbb{F}_2^n$, with $\text{dim}(A) = \frac{n}{2}$ and $s, s' \in \mathbb{F}_2^n$ uniformly at random. The public key is the pair $(\iO(P_{A+s}), \iO(P_{A^\perp + s'}))$. The (quantum) secret key is the coset state $\ket {A_{s, s'}}$. 
    
    \item To encrypt a message $m$, sample uniformly $r \leftarrow \{0,1\}$, and set $R = \iO(P_{A+ s})$ if $r = 0$ and $R = \iO(P_{A^\perp + s'})$ if $r = 1$. Then, let $C$ be the following program: \vspace{2mm}
    \\
        \quad $C$: on input $v$, output the message $m$ if $R(v) = 1$ and otherwise output $\bot$.
        \vspace{2mm}
        \\
    The ciphertext is then $(r, \iO(C))$.
    \item To decrypt a ciphertext $(r, \iO(C))$ with the quantum key $\ket {A_{s, s'}}$, one simply runs the program $\iO(C)$ coherently on input $\ket {A_{s, s'}}$ if $r = 0$, and on $H^{\otimes n} \ket {A_{s, s'}}$ if $r=1$.
\end{itemize}

In the full scheme, we actually amplify security by sampling $r \leftarrow \{0,1\}^{\lambda}$, and having $\lambda$ coset states, but we choose to keep the presentation in this section as simple as possible.

The high level idea for single-decryptor security is the following. 
Assume for the moment that $\iO$ were an ideal obfuscator (we will argue after this that $\iO$ is good enough). Consider a pirate who receives a secret key, produces two copies of it, and gives one to Bob and the other to Charlie. Suppose both Bob and Charlie can decrypt ciphertexts $(r, \iO(C))$ correctly with probability close to $1$, over the randomness in the choice of $r$ (which is crucially chosen only after Bob and Charlie have received their copies). %
Then, there must be some efficient quantum algorithm, which uses Bob's (resp. Charlie's) auxiliary quantum information (whatever state he has received from the pirate), and is able to output a vector in $A + s$. This is because in the case of $r = 0$, the program $C$ outputs the plaintext message $m$ exclusively on inputs $v \in A + s$. Similarly, there must be an algorithm that outputs a vector in $A^{\perp} + s'$ starting from Bob's (resp. Charlie's) auxiliary quantum information. Notice that this doesn't imply that Bob can \textit{simultaneously} output a pair in $(A+s) \times (A^{\perp} +s')$, because explicitly recovering a vector in one coset might destroy the auxiliary quantum information preventing recovery of a vector in the other (and this very fact is of course crucial to the direct product hardness). Hence, in order to argue that it is not possible for both Bob and Charlie to be decrypting with probability close to $1$, we have to use the fact that Bob and Charlie have separate auxiliary quantum information, and that each of them can recover vectors in $A+s$ or $A^{\perp}+s'$, which means that this can be done simultaneously, now violating the direct product hardness property.

The crux of the security proof is establishing that $\iO$ is a good enough obfuscator to enable this argument to go through. 

To this end, we first notice that there is an alternative way of computing membership in $A+s$, which is functionally equivalent to the program $C$ defined above.

Let $\can_A(s)$ be a function that computes the lexicographically smallest vector in $A + s$ (think of this as a representative of the coset). It is not hard to see that a vector $t$ is in $A + s$ if and only if $\can_A(t) = \can_A(s)$. Also $\can_A$ is efficiently computable given $A$. Therefore, a functionally equivalent program to $C$, in the case that $r=0$, is:

\vspace{0.3em}
\quad $\widetilde{C}$: on input $v$, output $m$ if $\can_A(v) = \can_A(s)$, otherwise output $\bot$.

\vspace{0.3em}

By the security of $\iO$, an adversary can't distinguish $\iO(C)$ from $\iO(\widetilde{C})$. 

The key insight is that now the program $\widetilde{C}$ is a \textit{compute-and-compare} program \cite{goyal2017lockable, wichs2017obfuscating}. The latter is a program described by three parameters: an efficiently computable function $f$, a target $y$ and an output $z$. The program outputs $z$ on input $x$ if $f(x) = y$, and otherwise outputs $\bot$. 
In our case, $f =\can_A$, $y = \can_A(s)$, and $z = m$. 
Goyal et al.~\cite{goyal2017lockable} and Wichs et al.~\cite{wichs2017obfuscating} show that, assuming LWE or assuming $\iO$ and certain PRGs, a compute-and-compare program can be obfuscated provided $y$ is (computationally) unpredictable given the function $f$ and the auxiliary information. More precisely, the obfuscation guarantee is that the obfuscated compute-and-compare program is indistinguishable from the obfuscation of a (simulated) program that outputs zero on every input (notice, as a sanity check, that if $y$ is unpredictable given $f$, then the compute-and-compare program must output zero almost everywhere as well). We will provide more discussion on compute-and-compare obfuscation for unpredictable distributions in the presence of quantum auxiliary input in \Cref{sec:cc} and \Cref{sec:CC_quantum_aux}. 
\begin{itemize}
    \item By the security of $\iO$, we can replace the ciphertext $(0, \iO(C))$, with the ciphertext $(0, \iO(\CC.\obf(\widetilde{C})))$ where $\CC.\obf$ is an obfuscator for compute-and-compare programs (this is because $C$ has the same functionality as $\CC.\obf(\widetilde{C})$).
    \item By the security of \textsf{CC.Obf}, we can replace the latter with $(0, \iO(\CC.\obf(Z)))$, where $Z$ is the zero program. It is clearly impossible to decrypt from the latter, since no information about the message is present.
\end{itemize}
Thus, assuming $\iO$ cannot be broken, a Bob that is able to decrypt implies an adversary breaking the compute-and-compare obfuscation. This implies that there must be an efficient algorithm that can predict $y = \can_A(s)$ with non-negligible probability given the function $\can_A$ and the auxiliary information received by Bob. Similarly for Charlie. 

Therefore, if Bob and Charlie, with their own quantum auxiliary information, can both independently decrypt respectively $(0, \iO(C))$ and $(1, \iO(C'))$ with high probability (where here $C$ and $C'$ only differ in that the former releases the encrypted message on input a vector in $A+s$, and $C'$ on input a vector in $A^{\perp} + s'$), then there exist efficient quantum algorithms for Bob and Charlie that take as input the descriptions of $\can_A(\cdot)$ and $\can_{A^{\perp}}(\cdot)$ respectively (or of the subspace $A$), and their respective auxiliary information, and recover $\can_A(s)$ and $\can_{A^\perp}(s')$ respectively with non-negligible probability. Since $\can_A(s) \in A + s$ and $\can_{A^\perp}(s') \in A^\perp + s'$, this violates the strong monogamy property of coset states described in Section \ref{sec: results}.

Recall that this states that Alice, Bob and Charlie cannot cooperatively win the following game with a challenger, except with negligible probability. The challenger first prepares a uniformly random coset state $\ket {A_{s, s'}}$ and gives the state to Alice. Alice outputs two (possibly entangled) quantum states and sends them to Bob and Charlie respectively. Finally, Bob and Charlie both get the description of the subspace $A$. The game is won if Bob outputs a vector in $A+s$ and Charlie outputs a vector in  $A^\perp+s'$. Crucially, in this monogamy property, Bob and Charlie 
will both receive the description of the subspace $A$ in the final stage, yet it is still not possible for both of them to be simultaneously successful.

What allows to deduce the existence of efficient extracting algorithms is the fact that the obfuscation of compute-and-compare programs from \cite{goyal2017lockable, wichs2017obfuscating} holds provided $y$ is computationally unpredictable given $f$ (and the auxiliary information). Thus, an algorithm that breaks the obfuscation property implies an efficient algorithm that outputs $y$ (with noticeable probability) given $f$ (and the auxiliary information).

In our other construction from signature tokens and extractable witness encryption, one can directly reduce unclonable decryption security to  direct product hardness. We do not discuss the details of this construction in this section, instead we refer the reader to Section \ref{sec: unclonable dec witness enc}.

\subsection{Copy-Protecting PRFs}
Our last contribution is the construction of copy-protected PRFs assuming post-quantum $\iO$, one-way functions and the monogamy property we discussed in the previous section. Alternatively just as for unclonable decryption, we can do away with the monogamy property by assuming extractable witness encryption. 

A copy-protectable PRF is a regular PRF $F: \{0,1\}^k \times \{0,1\}^m \rightarrow \{0,1\}^{m'}$, except that it is augmented with a \textit{quantum key} generation procedure, which we refer to as $\textsf{QKeyGen}$. This takes as input the classical PRF key $K$ and outputs a quantum state $\rho_K$. The state $\rho_K$ allows to efficiently compute $F(K, x)$ on any input $x$ (where correctness holds with overwhelming probability). Beyond the standard PRF security, the copy-protected PRF satisfies the following additional security guarantee:
any computationally bounded adversary that receives $\rho_K$ cannot process $\rho_K$ into two states, such that each state enables efficient evaluation of $F(K, \cdot)$ on uniformly random inputs. 

A simplified version of our construction has the following structure. For the rest of the section, we take all subspaces to be of $\mathbb{F}_2^n$ with dimension $n/2$.
\begin{itemize}
\item The quantum key generation procedure $\textsf{QKeyGen}$ takes as input a classical PRF key $K$ and outputs a quantum key. The latter consists of a number of uniformly sampled coset states $| (A_i)_{s_i, s_i'} \rangle$, for $i \in [\lambda]$, together with a (classical) \textit{obfuscation} of the classical program $P$ that operates as follows.
$P$ takes an input of the form $(x, v_1, \ldots, v_{\lambda})$; checks that each vector $v_i$ belongs to the correct coset ($A_i + s_i$ if $x_i = 0$, and $A_i^{\perp} + s_i'$ if $x_i = 1$); if so, outputs the value $F(K, x)$, otherwise outputs $\perp$.
\item A party in possession of the quantum key can evaluate the PRF on input $x$ as follows: for each $i$ such that $x_i = 1$, apply $H^{\otimes n}$ to $|(A_i)_{s_i, s_i'} \rangle$. Measure each resulting coset state in the standard basis to obtain vectors $v_1,\ldots, v_{\lambda}$. Run the obfuscated program on input $(x, v_1, \ldots, v_{\lambda})$.
\end{itemize}

Notice that the program has the classical PRF key $K$ hardcoded, as well as the values $A_i, s_i, s_i'$, so giving the program in the clear to the adversary would be completely insecure: once the adversary knows the key $K$, he can trivially copy the functionality $F(K, \cdot)$; and even if the key $K$ is hidden by the obfuscation, but the $A_i, s_i, s_i'$ are known, a copy of the (classical) obfuscated program $P$, together with the $A_i, s_i, s_i'$ is sufficient to evaluate $F(K, \cdot)$ on any input.

So, the hope is that an appropriate obfuscation will be sufficient to hide all of these parameters. If this is the case, then the intuition for why the scheme is secure is that in order for two parties to simultaneously evaluate correctly on uniformly random inputs, each party should be able to produce a vector in $A_i+s$ or in $A_i^{\perp} + s_i'$. If the two parties accomplish this separately, then this implies that it is possible to simultaneously extract a vector in $A_i+s_i$ and one in $A_i^{\perp} + s_i'$, which should not be possible. \footnote{Again, we point out that we could not draw this conclusion if only a single party were able to do the following two things, each with non-negligible probability: produce a vector in $A+s_i$ and produce a vector in $A^{\perp}+s_i'$. This is because in a quantum world, being able to perform two tasks with good probability, does not imply being able to perform both tasks simultaneously. So it is crucial that both parties are able to separately recover the vectors.}

We will use $\iO$ to obfuscate the program $P$. In the next part of this overview, we will discuss how we are able to deal with the fact that the PRF key $K$ and the cosets are hardcoded in the program $P$. First of all, we describe a bit more precisely the copy-protection security that we wish to achieve. The latter is captured by the following security game between a challenger and an adversary $(A, B, C)$: 
\begin{itemize}
    \item The challenger samples a uniformly random PRF key $K$ and runs $\textsf{QKeyGen}$ to generate $\rho_K$. Sends $\rho_K$ to $A$.
    \item $A$ sends quantum registers to two spatially separated parties $B$ and $C$.
    \item The challenger samples uniformly random inputs $x,x'$ to $F(K,\cdot)$. Sends $x$ to $B$ and $x'$ to $C$.
    \item $B$ and $C$ return $y$ and $y'$ respectively to the challenger.
\end{itemize}
$(A, B, C)$ wins if $y = F(K, x)$ and $y' = F(K, x')$.

Since the obfuscation we are using is not VBB, but only $\iO$, there are two potential issues with security. $B$ and $C$ could be returning correct answers not because they are able to produce vectors in the appropriate cosets, but because:
\begin{itemize}
    \item[(i)] $\iO(P)$ leaks information about the PRF key $K$.
    \item[(ii)] $\iO(P)$ leaks information about the cosets.
\end{itemize}
We handle issue (i) via a delicate ``puncturing'' argument \cite{sahai2014use}. At a high level, a puncturable PRF $F$ is a PRF augmented with a procedure that takes a key $K$ and an input value $x$, and produces a ``punctured'' key $K\setminus \{x\}$, which enables evaluation of $F(K, \cdot)$ at any point other than $x$. The security guarantee is that a computationally bounded adversary possessing the punctured key $K\setminus \{x\}$ cannot distinguish between $F(K, x)$ and a uniformly random value (more generally, one can puncture the key at any polynomially sized set of points). Puncturable PRFs can be obtained from OWFs using the \cite{GGM86} construction \cite{boneh2013constrained}. %

By puncturing $K$ precisely at the challenge inputs $x$ and $x'$, one is able to hardcode a punctured PRF key $K\setminus \{x,x'\}$ in the program $P$, instead of $K$, and setting the output of program $P$ at $x$ to uniformly random $z$ and $z'$, instead of to $F(K,x)$ and $F(K,x')$ respectively. The full argument is technical, and relies on the ``hidden trigger'' technique introduced in \cite{sahai2014use}, \revise{which allows the ``puncturing'' technique to work even when the program $P$ is generated before $x$ and $x'$ are sampled}.

Once we have replaced the outputs of the program $P$ on the challenge inputs $x,x'$ with uniformly random outputs $z, z'$, we can handle issue (ii) in a similar way to the case of unclonable decryption in the previous section. 

By the security of $\iO$, we can replace the behaviour of program $P$ at $x$ by a suitable functionally equivalent compute-and-compare program that checks membership in the appropriate cosets. We then replace this by an obfuscation of the same compute-and-compare program, and finally by an obfuscation of the zero program. We can then perform a similar reduction as in the previous section from an adversary breaking copy-protection security (and thus the security of the compute-and-compare obfuscation) to an adversary breaking the monogamy of entanglement game described in the previous section. 

As in the previous section, we can replace the reliance on the conjectured monogamy property by extractable witness encryption. In fact, formally, we directly reduce the security of our copy-protected PRFs to the security of our unclonable decryption scheme.

\section{Preliminaries}
In this paper, we use $\lambda$  to denote security parameters. We denote a function belonging to the class of polynomial functions by $\poly(\cdot)$. 
We say a function $f(\cdot) : \mathbb{N} \to \mathbb{R}^+$ is negligible if for all constant $c > 0$, $f(n) < \frac{1}{n^c} $ for all large enough $n$. 
We use $\negl(\cdot)$ to denote a negligible function.
We say a function $f(\cdot) : \mathbb{N} \to \mathbb{R}^+$ is sub-exponential if there exists a constant $0 < c < 1$, such that $f(n) = {2^{n^c}} $ for all large enough $n$. 
We use $\subexp(\cdot)$ to denote a sub-exponential function.

When we refer to a probabilistic algorithm $\mathcal{A}$, sometimes we need to specify the randomness $r$ used by $\mathcal{A}$ when running on some input $x$. We write this as $\mathcal{A}(x ; r)$.

For a finite set $S$, we use $x \gets S$ to denote uniform sampling of $x$ from the set $S$. 
We denote $[n] = \{1, 2, \cdots, n\}$. 
A binary string $x \in \{0,1\}^{\ell}$ is represented as $x_1 x_2 \cdots x_\ell$.
For two strings $x, y$, $x || y$  is the concatenation of $x$ and $y$. 

We refer to a probabilistic polynomial-time algorithm as PPT, and we refer to a quantum polynomial-time algorithm as QPT.

We will assume familiarity with basic quantum information and computation concepts. We refer the reader to \Cref{appendix:quantum_info} and \cite{nielsen2002quantum} for a reference.

\subsection{Pseudorandom Functions}

For the rest of this paper, we will assume that all of the classical cryptographic primitives used are post-quantum (i.e. secure against quantum adversaries), and we sometimes omit mentioning this for convenience, except in formal definitions and theorems.

\begin{definition}[PRF] \label{def:wPRF}
A pseudorandom function (PRF) is a function $F:\{0,1\}^k \times \{0,1\}^n \to \{0,1\}^m$, where $\{0, 1\}^k$ is the
key space, and  $\{0, 1\}^n$ and $\{0, 1\}^m$ are the domain and range. $k, n$ and $m$ are implicity functions of a security parameter $\lambda$. The following should hold: 
\begin{itemize}
    \item For every $K \in \{0,1\}^k$, $F(K, \cdot)$ is efficiently computable;
    \item PRF security:  no efficient quantum adversary $\As$ making quantum queries can distinguish between a truly random function and the function $F(K, \cdot)$; that is for every such $\As$, there exists a negligible function $\negl$, 
    \begin{align*}
        \left| \Pr_{K \gets \{0,1\}^k}\left[ \As^{F(K, \cdot)}() = 1 \right]  -  \Pr_{O: \{0,1\}^n \to \{0,1\}^m}\left[ \As^{O}() = 1 \right] \right| \leq \negl(\lambda)
    \end{align*}
\end{itemize}
\end{definition}

\subsection{Indistinguishability Obfuscation}

\begin{definition}[Indistinguishability Obfuscator (iO)~\cite{barak2001possibility,garg2016candidate,sahai2014use}]
A uniform PPT machine $\iO$ is an indistinguishability obfuscator for a circuit class $\{\Cs_\lambda\}_{\lambda \in \mathbb N}$ if the following conditions are satisfied:
\begin{itemize}
    \item For all $\lambda$, all $C \in \Cs_\lambda$, all inputs $x$, we have 
    \begin{align*}
        \Pr\left[\widehat{C}(x) = C(x) \,|\, \widehat{C} \gets \iO(1^\lambda, C) \right] = 1
    \end{align*}
    
    \item (Post-quantum security): For all (not necessarily uniform) QPT adversaries $(\samp, D)$, the following holds: if $\Pr[\forall x, C_0(x) = C_1(x) \,:\, (C_0, C_1, \sigma) \gets \samp(1^\lambda)] > 1 - \alpha(\lambda)$ for some negligible function $\alpha$, then there exists a negligible function $\beta$ such that:
    \begin{align*}
        & \Bigg|\Pr\left[D(\sigma, \iO(1^\lambda, C_0)) =1\,:\, (C_0, C_1, \sigma) \gets \samp(1^\lambda)\right] \\ 
        - & \Pr\left[D(\sigma, \iO(1^\lambda, C_1)) = 1\,:\, (C_0, C_1, \sigma) \gets \samp(1^\lambda)\right] \Bigg| \leq \beta(\lambda)
    \end{align*}
\end{itemize}
\end{definition}

Whenever we assume the existence of $\iO$ in the rest of the paper, we refer to $\iO$ for the class of polynomial-size circuits, i.e. when $\mathcal{C}_{\lambda}$ is the collection of all circuits of size at most $\lambda$.

We will also make use of the stronger notion of \emph{sub-exponentially secure} $\iO$. By the latter, we mean that the distinguishing advantage above is $1/\subexp$ for some sub-exponential function $\subexp$, instead of negligible (while the adversary is still $QPT$).

Similarly, we will also make use of sub-exponentially secure one-way functions. For the latter, the advantage is again $1/\subexp$ (and the adversary is $QPT$).

\subsection{Compute-and-Compare Obfuscation} \label{sec:cc}
\begin{definition}[Compute-and-Compare Program]
    Given a function $f:\{0,1\}^{\ell_{\sf in}} \to \{0,1\}^{\ell_{\sf out}}$ along with a target value $y \in \{0,1\}^{\ell_{\sf out}}$ and a message $z \in \{0,1\}^{\ell_{\sf msg}}$, we define the compute-and-compare program: 
    \begin{align*}
        \CC[f, y, z](x) = \begin{cases}
                            z & \text{ if } f(x) = y \\
                            \bot & \text{ otherwise }
                        \end{cases}
    \end{align*}
\end{definition}

We define the following class of \emph{unpredictable distributions} over pairs of the form $(\CC[f, y, z], \aux)$, where $\aux$ is auxiliary quantum information. These distributions are such that $y$ is computationally unpredictable given $f$ and $\aux$.

\begin{definition}[Unpredictable Distributions]
\label{def:cc_unpredictable_dist}
    We say that a family of distributions $D = \{D_\lambda\}$ where $D_{\lambda}$ is a distribution over pairs of the form $(\CC[f, y, z], \aux)$ where $\aux$ is a quantum state, belongs to the class of \emph{unpredictable distributions} if the following holds. There exists a negligible function $\negl$ such that, for all QPT algorithms $\As$, 
    \begin{align*}
        \Pr_{ (\CC[f, y, z], \aux) \gets D_\lambda } \left[ A(1^\lambda, f, \aux) = y \right] \leq \negl(\lambda). 
    \end{align*}
\end{definition}

We further define the class of \emph{sub-exponentially unpredictable distributions}, where we require the guessing probability to be inverse sub-exponential in the security parameter. 
\begin{definition}[Sub-Exponentially Unpredictable Distributions]
\label{def:cc_subexp_unpredictable_dist}
    We say that a family of distributions $D = \{D_\lambda\}$ where $D_{\lambda}$ is a distribution over pairs of the form $(\CC[f, y, z], \aux)$ where $\aux$ is a quantum state, belongs to the class of \emph{sub-exponentially unpredictable distributions} if the following holds. There exists a sub-exponential function $\subexp$ such that, for all QPT algorithms $\As$, 
    \begin{align*}
        \Pr_{ (\CC[f, y, z], \aux) \gets D_\lambda } \left[ A(1^\lambda, f, \aux) = y \right] \leq 1/\subexp(\lambda). 
    \end{align*}
\end{definition}

We assume that a program $P$ has an associated set of parameters $P.{\sf param}$ (e.g input size, output size, circuit size, etc.), which we are not required to hide.
\begin{definition}[Compute-and-Compare Obfuscation]
\label{def: cc obf}
    A PPT algorithm $\ccobf$ is an obfuscator for the class of unpredictable distributions (or sub-exponentially unpredictable distributions) if for any family of distributions $D = \{ D_{\lambda}\}$ belonging to the class, the following holds:
    \begin{itemize}
        \item Functionality Preserving: there exists a negligible function $\negl$ such that for all $\lambda$, every program $P$ in the support of $D_\lambda$, 
        \begin{align*}
            \Pr[\forall x,\,  \widetilde{P}(x) = P(x),\, \widetilde{P} \gets \ccobf(1^\lambda, P) ] \geq 1 - \negl(\lambda)
        \end{align*}
        \item Distributional Indistinguishability: there exists an efficient simulator $\Sim$ such that:
        \begin{align*}
            (\ccobf(1^\lambda, P), \aux) \approx_c (\Sim(1^\lambda, P.{\sf param}), \aux)
        \end{align*}
        where $(P, \aux) \gets D_\lambda$.
    \end{itemize}
\end{definition}

Combining the results of \cite{wichs2017obfuscating, goyal2017lockable} with those of \cite{zhandry2019magic}, we have the following two theorems. For the proofs and discussions, we refer the readers to \Cref{sec:CC_quantum_aux}. Note that although \Cref{thm:CC__from_ELF_iO} is a strictly stronger statement, currently we do not know of any post-quantum construction for ELFs. 
\begin{reptheorem}{thm:CC_subexp_from_LWE_iO}
\label{thm: cc sub-exp from lwe}
Assuming the existence of post-quantum $\iO$ and the
quantum hardness of LWE, there exist obfuscators for sub-exponentially unpredictable distributions, as in \Cref{def: cc obf}.
\end{reptheorem}
\begin{reptheorem}{thm:CC__from_ELF_iO}
    Assuming the existence of post-quantum $\iO$ and post-quantum
extremely lossy functions (ELFs), there exist obfuscators as in \Cref{def: cc obf}. for any unpredictable distributions.
\end{reptheorem}

\subsection{Subspace Hiding Obfuscation}
\label{sec: shO}

Subspace-hiding obfuscation was introduced by Zhandry~\cite{zhandry2019quantum} as a key component in constructing public-key quantum money. This notion requires that the obfuscation of a circuit that computes membership in a subspace $A$ is indistinguishable from the obfuscation of a circuit that computes membership in a uniformly random superspace of $A$ (of dimension sufficiently far from the full dimension). The formal definition is as follows.
\begin{definition}[\cite{zhandry2019quantum}]
A subspace hiding obfuscator (shO) for a field $\F$ and dimensions $d_0, d_1$ is a PPT algorithm $\shO$ such that:
\begin{itemize}
    \item \textbf{Input.} $\shO$ takes as input the description of a linear subspace $S \subseteq \F^n$ of dimension $d \in \{d_0, d_1\}$.
    
    For concreteness, we will assume $S$ is given as a matrix whose rows form a basis for $S$.
    \item \textbf{Output.} $\shO$ outputs a circuit $\hat{S}$ that computes membership in $S$. Precisely, let $S(x)$ be the function that decides membership in $S$. Then there exists a negligible function $\negl$,
    \begin{align*}
        \Pr[\hat{S}(x) = S(x)~~\forall x : \hat{S} \leftarrow \shO(S)] \geq 1 - \negl(n)
    \end{align*}

    \item \textbf{Security.} For security, consider the following game between an adversary and a challenger.
    \begin{itemize}
        \item The adversary submits to the challenger a subspace $S_0$ of dimension $d_0$.
        \item The challenger samples a uniformly random subspace $S_1 \subseteq \F^n$ of dimension $d_1$ such that $S_0 \subseteq S_1$.
        
        It then runs $\hat{S} \leftarrow \shO(S_b)$, and gives $\hat{S}$ to the adversary.
        \item The adversary makes a guess $b'$ for $b$.
    \end{itemize}
    $\shO$ is secure if all QPT adversaries have negligible advantage in this game.
\end{itemize}
\end{definition}

Zhandry \cite{zhandry2019quantum} gives a construction of a subspace hiding obfuscator based on one-way functions and $\iO$.
\begin{theorem}[Theorem 6.3 in \cite{zhandry2019quantum}]
If injective one-way functions exist, then any indistinguishability obfuscator, appropriately padded, is also a subspace hiding obfuscator for field $\F$ and dimensions $d_0, d_1$, as long as $|\F|^{n-d_1}$ is exponential.
\end{theorem}

\subsection{Extractable Witness Encryption}
\label{appendix:extractable_witness}
In this subsection, we describe the primitive of witness encryption~\cite{garg2017implausibility} with extractable security, which will we use in our construction of unclonable decryption in Section \ref{sec: unclonable dec witness enc}.
\begin{definition}[Extractable Witness Encryption]
\label{def: extractable witness enc}
An extractable witness encryption scheme for an $\np$
relation $R$ is a pair of algorithms $(\enc, \dec)$:

\begin{itemize}

    \item $\enc(1^\param, x, m) \to \ct:$ 
   takes as input a security parameter $\lambda$ in unary, an instance $x$ and
a message $m$, and outputs a ciphertext $\ct$.

\item $\dec(\ct, w) \to m/\bot:$ takes as input a ciphertext $\ct$ and a witness $w$ and outputs a message $m$ or $\bot$ (for decryption failure).

\end{itemize}

The scheme satisfies the following:

\begin{description}
\item[Correctness:]
For any security parameter $\lambda \in \N$, for any $m \in \{0,1\}$, for any $x$ and $w$ such that $R(x, w) = 1$, we have that:
\begin{align*}
    & \Pr[ \dec(\enc(1^\lambda, x, m), w] = m ] = 1
\end{align*}

\item[Extractable Security:]
 For any QPT adversary $\cA$, polynomial-time sampler $(x, \aux) \gets \samp(1^\lambda)$
and for any polynomial $q(\cdot)$, there exists a QPT extractor $E$ and a polynomial $p(\cdot)$, such that:
\begin{align*}
     & \Pr\left[ \cA(1^\lambda, x, \ct, \aux) = m \middle| \begin{array}{cr}
         m \gets \{0, 1\},  (x, \aux) \gets \samp(1^\lambda),\\
        \ct  \gets \enc(1^\lambda, x, m)
    \end{array}
    \right] \geq \frac{1}{2} + \frac{1}{q(\lambda)} \\
    & \rightarrow \Pr\left[ E(1^\lambda, x, \aux) = w \text{ s.t. } R(x, w) = 1 : (x, \aux) \gets \samp(1^\lambda) \right]
    \geq \frac{1}{p(\lambda)}. 
\end{align*}
\end{description}
\end{definition}

\subsection{Testing Quantum Adversaries: Projective Implementation} \label{sec:unclonable dec ati}

In this section, we include several definitions about measurements, which are relevant to testing whether quantum adversaries are successful in the security games of \Cref{sec: unclonable dec strong ag}. Part of this section is taken verbatim from \cite{aaronsonnew}. As this section only pertains directly to our security definitions for unclonable decryption schemes, the reader can skip ahead, and return to this section when reading \Cref{sec: unclonable dec strong ag}. In particular, this section is not needed to understand Sections \ref{sec: coset states} and \ref{sec: signature tokens}.

\vspace{1mm}

In classical cryptographic security games, the challenger typically gets some information from the adversary and checks if this information satisfies certain properties.
However, in a setting where the adversary is required to return \emph{quantum} information to the challenger, classical definitions of ``testing'' whether a quantum state returned by the adversary satisfies certain properties may result in various failures as discussed in \cite{z20}, as this state may be in a superposition of ``successful'' and ``unsuccessful'' adversaries. We provide here a short description of some of the difficulties in the quantum setting, and we refer the reader to \cite{z20} for a more in-depth discussion. %

\revise{

As an example, consider a security game in which an adversary is required to return some information to a challenger, which enables evaluation of a program on any input. Such a scenario is natural in copy-protection, where the adversary (a ``pirate'') attempts to create two copies of a copy-protected program, given just a single copy (and one can think of these two copies as being returned to the challenger for testing).

Naturally, one would consider a copy-protected program to be ``good'' if it enables correct evaluation on all inputs, or at least on a large fraction of all inputs. Testing correct evaluation on all inputs is of course not possible efficiently (not even classically). Instead, one would typically have the challenger \emph{estimate} the fraction of correct evaluations to high statistical confidence by picking a large enough number of inputs uniformly at random (or from an appropriate distribution), running the copy-protected program on these inputs, and computing the fraction of correct evaluations. Unfortunately, such a test does not easily translate to the quantum setting. 
The reason is that the challenger only gets a single copy of the program, which in a quantum world cannot be generically copied. Moreover, in general, each evaluation may alter the copy-protected program in an irreversible way (if the outcome of the evaluation is not deterministic). Thus, estimating the fraction of inputs on which the copy-protected program received from the adversary evaluates correctly is not in general possible. For instance, consider an adversary who sends a state $\frac{1}{\sqrt{2}} \ket {P_0} + \frac{1}{\sqrt{2}} \ket {P_1}$ to the challenger, where $\ket {P_{0}}$ is a copy-protected program that evaluates perfectly on every input, and $\ket {P_{1}}$ is a useless program. Using this state, evaluation is successful on any input with probability $1/2$. Thus, even a single evaluation collapses the state either to $\ket {P_{0}}$ or to $\ket {P_1}$, preventing the challenger from performing subsequent evaluations on the original state. In fact, it is impossible to have a generic procedure that estimates the ``average success probability of evalutation'' to very high precision, as this would imply a procedure that distinguishes between the state $\frac{1}{\sqrt{2}} \ket {P_0} + \frac{1}{\sqrt{2}} \ket {P_1}$ and the state $\ket {P_0}$ almost perfectly, which is impossible since the two states have large overlap.%
}
\vspace{1em}

\paragraph{Projective Implementation}
Motivated by the discussion above, \cite{z20} formalizes a new measurement procedure for testing a state received by an adversary. We will be adopting this procedure when defining security of single-decryptor encryption schemes in Section \ref{sec: unclonable dec strong ag}.

Consider the following procedure as a binary POVM $\cP$ acting on an alleged-copy-protected program $\rho$: sample a uniformly random input $x$, evaluates the copy-protected program on $x$, and checks if the output is correct.   
In a nutshell, the new procedure consists of applying an appropriate projective measurement which \emph{measures} the success probability of the tested state $\rho$ under $\cP$, and to output ``accept'' if the success probability is high enough. Of course, such measurement will not be able extract the exact success probability of $\rho$, as this is impossible from we have argued in the discussion above. Rather, the measurement will output a success probability from a finite set, such that the expected value of the output matches the true success probability of $\rho$. We will now describe this procedure in more detail.

The starting point is that a POVM specifies exactly the probability distribution over outcomes $\{0,1\}$ (``success'' or ``failure'') on any copy-protected program, but it does not uniquely determine the post-measurement state. Zhandry shows that, for any binary POVM $\cP = (P, I-P)$, there exists a particularly nice implementation of $\cP$ which is projective, and such that the post-measurement state is an eigenvector of $P$. In particular, Zhandry observes that there exists a projective measurement $\cE$ which \emph{measures} the success probability of a state with respect to $\cP$. More precisely,
\begin{itemize}
    \item $\cE$ outputs a \emph{distribution} $D$ of the form $(p, 1-p)$ from a finite set of distribution over outcomes $\{0,1\}$. (we stress that $\cE$ actually outputs a distribution).
    \item The post-measurement state upon obtaining outcome $(p,1-p)$ is an \emph{eigenvector} (or a mixture of eigenvectors) of $P$ with eigenvalue $p$.
\end{itemize}

A measurement $\cE$ which satisfies these properties is the measurement in the common eigenbasis of $P$ and $I-P$ (such common eigenbasis exists since $P$ and $I-P$ commute). 

Note that since $\cE$ is projective, we are guaranteed that applying the same measurement twice will yield the same outcome. Thus, what we obtain from applying $\cE$ is a state with a ``well-defined'' success probability with respect to $\cP$: we know exactly how good the leftover program is with respect to the initial testing procedure $\cP$.

Formally, to complete the implementation of $\cP$, after having applied $\cE$, one outputs the bit $1$ with probability $p$, and the bit $0$ with probability $1-p$. This is summarized in the following definition.

\begin{definition}[Projective Implementation of a POVM]
\label{def:project_implement}
    Let $\cP = (P, Q)$ be a binary outcome POVM. Let $\cD$ be a finite set of distributions $(p, 1-p)$ over outcomes $\{0, 1\}$. Let $\cE = \{E_p\}_{(p, 1-p) \in \cD}$ be a projective measurement with index set $\cD$. Consider the following measurement procedure: 
    \begin{itemize}
        \item[(i)] Apply the projective measurement $\cE$ and obtain as outcome a distribution $(p, 1-p)$ over $\{0, 1\}$;
        \item[(ii)] Output a bit according to this distribution, i.e. output $1$ w.p $p$ and output $0$ w.p $1-p$. 
    \end{itemize}
    We say the above measurement procedure is a projective implementation of $\cP$, which we denote by $\projimp(\cP)$, if it is equivalent to $\cP$ (i.e. it produces the same probability distribution over outcomes).
\end{definition}

Zhandry shows that any binary POVM has a projective implementation, as in the previous definition.

\begin{lemma}[Adapted from Lemma 1 in \cite{z20}]
\label{lem:proj_implement}
    Any binary outcome POVM $\mathcal{P} = (P, Q)$ has a projective implementation $\projimp(\cP)$.
    
    Moreover, if the outcome is a distribution $(p, 1-p)$ when measuring under $\cE$, the collapsed state $\rho'$ is a mixture of eigenvectors of $P$ with eigenvalue $p$, and it is also a mixture of eigenvectors of $Q$ with eigenvalue $1 - p$. %
\end{lemma}

As anticipated, the procedure that we will eventually use to test a state received from the adversary will be to:
\begin{itemize}
    \item[(i)] \emph{Measure} the success probability of the state,
    \item[(ii)] Accept if the outcome is large enough. 
\end{itemize}
As you may guess at this point, we will employ the projective measurement $\cE$ defined previously for step $(i)$. We call this variant of the projective implementation a \emph{threshold implementation}.

\vspace{1em}

\paragraph{Threshold Implementation}
The concept of threshold implementation of a POVM was proposed by Zhandry, and formalized by Aaronson, Liu, Liu, Zhandry and Zhang~\cite{aaronsonnew}.
The following is a formal definition.
\begin{definition}[Threshold Implementation]
\label{def:thres_implement}
Let $\cP = (P, Q)$ be a binary POVM. Let $\projimp(\cP)$ be a projective implementation of $\cP$, and let $\cE$ be the projective measurement in the first step of $\projimp(\cP)$ (using the same notation as in Definition \ref{def:project_implement}). Let $\gamma >0$. We refer to the following measurement procedure as a \emph{threshold implementation} of $\cP$ with parameter $\gamma$, and we denote is as $\ti_\gamma(\cP)$.
\begin{itemize}
        \item Apply the projective measurement $\cE$, and obtain as outcome a vector $(p, 1-p)$;
        \item Output a bit according to the distribution $(p, 1-p)$: output $1$ if $p \geq \gamma$, and $0$ otherwise. 

\end{itemize}
\end{definition}

For simplicity, for any quantum state $\rho$, we denote by $\Tr[\ti_{\gamma}(\cP) \, \rho]$ the probability that the threshold implementation applied to $\rho$ \textbf{outputs} $\mathbf{1}$. Thus, whenever $\ti_{\gamma}(\cP)$ appears inside a trace $\Tr$, we treat $\ti_{\gamma}(\cP)$ as a projection onto the $1$ outcome (i.e. the space spanned by eigenvectors of $P$ with eigenvalue at least $\gamma$).

Similarly to \Cref{lem:proj_implement}, we have the following lemma.

\begin{lemma}
\label{lem:threshold_implementation}
    Any binary outcome POVM $\mathcal{P} = (P, Q)$ has a threshold   implementation $\ti_{\gamma}(\cP)$ for any $\gamma$. 
\end{lemma}

\vspace{1em}

In this work, we are interested in threshold implementations of POVMs with a particular structure. These POVMs represent a challenger's test of a quantum state received from an adversary in a security game (like the POVM described earlier for testing whether a program evaluates correctly on a uniformly random input). These POVMs have the following structure:
\begin{itemize}
    \item Sample a projective measurement from a set of projective measurements $\mathcal{I}$, according to some distribution $D$ over $\mathcal{I}$.
    \item Apply this projective measurement.
\end{itemize}

We refer to POVMs of this form as \emph{mixtures of projective measurements}. The following is a formal definition.

\begin{definition}[Mixture of Projective Measurements] \label{def:mixture_of_projective}
Let $\mathcal{R}$, $\mathcal{I}$ be sets. Let $D: \mathcal{R} \rightarrow \mathcal{I}$. Let $\{(P_i, Q_i)\}_{i \in I}$ be a collection of binary projective measurements. The \emph{mixture of projective measurements} associated to $\mathcal{R}$, $\mathcal{I}, D$ and $\{(P_i, Q_i)\}_{i \in I}$ is the binary POVM $\cP_D = (P_D, Q_D)$ defined as follows: 
\begin{align*}
 P_D = \sum_{i \in \cal I} \Pr[i \gets D(R)] \, P_i ,\,\,\,\,\,\text{  }\,\,\,\,\,  Q_D = \sum_{i \in \cal I} \Pr[i \gets D(R)] \, Q_i, 
\end{align*}
where $R$ is uniformly distributed in $\cR$. 
\end{definition}

In other words, $\cP_D$ is implemented in the following way: sample randomness $r \gets \cR$, compute the index $i = D(r)$, and apply the projective measurement $(P_i, Q_i)$. Thus, for any quantum state $\rho$, $\Tr[P_D \rho]$ is the probability that a projective measurement $(P_i, Q_i)$, sampled according to the distribution induced by $D$, applied to $\rho$ outputs $1$. 

The following lemma will be important in the proof of security for our single-decryptor encryption scheme in Section \ref{sec:unclonable_dec}.

Informally, the lemma states the following. Let $\cP_{D_0}$ and $\cP_{D_1}$ be two mixtures of projective measurements, where $D_0$ and $D_1$ are two computationally indistinguishable distributions. Let $\gamma, \gamma'>0$ be inverse-polynomially close. Then for any (efficiently constructible) state $\rho$, the probabilities of obtaining outcome $1$ upon measuring $\ti_{\gamma}(\cP_{D_0})$ and $\ti_{\gamma'}(\cP_{D_1})$ respectively are negligibly close.

\begin{theorem}[Theorem 6.5 in \cite{z20}] \label{thm:ti_different_distribution}
Let $\gamma >0$. Let $\cP$ be a collection of projective measurements indexed by some set $\cal I$. Let $\rho$ be an efficiently constructible mixed state, and let $D_0, D_1$ be two efficiently sampleable and computationally indistinguishable distributions over $\cal I$. For any inverse polynomial $\epsilon$, there exists a negligible function $\delta$ such that
\begin{align*}
    \Tr[\ti_{\gamma - \epsilon}(\cP_{D_1}) \rho] \geq \Tr[\ti_{\gamma}(\cP_{D_0}) \rho] - \delta \,,
\end{align*}
where $\cP_{D_i}$ is the mixture of projective measurements associated to $\cP$ and $D_i$.
\end{theorem}

\vspace{0.5em}

%

%
%
%
%
%
%
%

%

%
%
%

%

\vspace{1em}

\paragraph{Approximating Threshold Implementation}
\emph{Projective} and \emph{threshold} implementations of POVMs are unfortunately not efficiently computable in general.

However, they can be approximated if the POVM is a mixture of projective measurements, as shown by Zhandry \cite{z20}, using a technique first introduced by Marriott and Watrous \cite{marriott2005quantum} in the context of error reduction for quantum Arthur-Merlin games.

We will make use of the following lemma from a subsequent work of Aaronson et al.~\cite{aaronsonnew}.

\begin{lemma}[Corollary 1 in \cite{aaronsonnew}]\label{cor:ati_thresimp}
    For any $\epsilon, \delta, \gamma \in (0,1)$, any collection of projective measurements $\cP = \{(P_i, Q_i)\}_{i \in \mathcal{I}}$, where $\mathcal{I}$ is some index set, and any distribution $D$ over $\mathcal{I}$, there exists a measurement procedure $\ati^{\epsilon, \delta}_{\cP, D, \gamma}$ that satisfies the following:
    \begin{itemize}
        \item $\ati^{\epsilon, \delta}_{\cP, D, \gamma}$ implements a binary outcome measurement. For simplicity, we denote the probability of the measurement \textbf{outputting} $\mathbf{1}$ on $\rho$ by $\Tr[\ati^{\epsilon, \delta}_{\cP, D, \gamma}\, \rho]$. %
    
        \item For all quantum states $\rho$, $\Tr[\ati^{\epsilon, \delta}_{\cP, D, \gamma}\, \rho]\geq \Tr[\ti_\gamma(\cP_D)\, \rho]-\delta$. 
        
        \item 
        For all quantum states $\rho$, 
        let $\rho'$ be the post-measurement state after applying  $\ati^{\epsilon, \delta}_{\cP, D, \gamma}$ on $\rho$, and obtaining outcome $1$. Then, $\Tr[\ti_{\gamma-2 \epsilon}(\cP_D)\, \rho'] \geq 1 - 2\delta$. 
        \item
        The expected running time is $T_{\cP, D} \cdot \poly(1/\epsilon, 1/(\log \delta))$, where $T_{\cP, D}$ is the combined running time of sampling according to $D$, of mapping $i$ to $(P_{i}, Q_{i})$, and of implementing the projective measurement $(P_{i}, Q_{i})$. 
    \end{itemize}
\end{lemma}

Intuitively the corollary says that if a quantum state $\rho$ has weight $p$ on eigenvectors with eigenvalues at least $\gamma$, then the measurement $\ati^{\epsilon, \delta}_{\cP, D, \gamma}$ will produce with probability at least $p - \delta$ a post-measurement state which has weight $1 - 2 \delta$ on eigenvectors with eigenvalues at least $\gamma - 2 \epsilon$. Moreover, the running time for implementing $\ati^{\epsilon, \delta}_{\cP, D, \gamma}$ is proportional to $\poly(1/\epsilon, 1/(\log \delta))$, which is a polynomial in $\lambda$ as long as $\epsilon$ is any inverse polynomial and $\delta$ is any inverse sub-exponential function. 

Crucially for applications to single-decryption encryption and copy-protection, the above lemma can be generalized to pairs of POVMs on bipartite states.

\begin{lemma}[Lemma 3 in~\cite{aaronsonnew}] \label{lem:ati_2d}
Let $\cP_1$ and $\cP_2$ be two  collections of projective measurements, indexed by elements of $\mathcal{I}$, and let $D_1$ and $D_2$ be probability distributions over $\mathcal{I}$.

For any $\epsilon, \delta, \gamma \in (0,1)$, let $\ati^{\epsilon, \delta}_{\cP_1, D_1, \gamma}$ and $\ati^{\epsilon, \delta}_{\cP_2, D_2, \gamma}$ be the measuring algorithms above. They satisfy: 
\begin{itemize}
    \item For any bipartite (possibly entangled, mixed) quantum state $\rho\in \mathcal{H}_{1}\otimes \mathcal{H}_{2}$,
    \begin{align*}
        \Tr\big[\big(\ati_{\cP_1, D_1, \gamma}^{\epsilon, \delta}\otimes \ati_{\cP_2, D_2, \gamma}^{\epsilon, \delta}\big)\rho\big] \geq \Tr\big[ \big(\ti_{\gamma}(\cP_{1, D_1})\otimes \ti_\gamma(\cP_{2, D_2})\big)\rho \big] - 2\delta,
    \end{align*}
    where $\cP_{i, D_i}$ is the mixture of projective measurement corresponding to $\cP_i, D_i$. 
    \item For any (possibly entangled, mixed) quantum state $\rho \in \mathcal{H}_{1}\otimes \mathcal{H}_{2}$,
        let $\rho'$ be the (normalized) post-measurement state after applying the measurements $\ati^{\epsilon, \delta}_{\cP_1, D_1, \gamma}$ and  $\ati^{\epsilon, \delta}_{\cP_2, D_2, \gamma}$ to $\rho$ and obtaining outcomes $1$ for both. Then, 
        \begin{align*}
            \Tr\big[ \big(\ti_{\gamma-2\epsilon}(\cP_{1, D_1})\otimes \ti_{\gamma-2\epsilon}(\cP_{2, D_2})\big) \rho'\big]\geq  1- 4 \delta.
        \end{align*}
\end{itemize}
\end{lemma}

\section{Coset States}
\label{sec: coset states}

This section is organized as follows. In Section \ref{sec: affine subspaces defs}, we introduce coset states. In Section \ref{sec: direct product}, we show that coset states satisfy both an information-theoretic and a computational \emph{\compuncertain} property. The latter immediately yields a signature token scheme in the plain model assuming $\iO$, (this is described in Section \ref{sec: signature tokens}). In Section \ref{sec: monogamy} we show that coset states satisfy both an information-theoretic \emph{monogamy of entanglement} property (analogous to that satisfied by BB84 states \cite{tomamichel2013monogamy}), and a computational monogamy of entanglement property. The latter is used in Section \ref{sec: unclonable dec witness enc} to obtain an unclonable decryption scheme from $\iO$ and extractable witness encryption. In Section \ref{sec: monogamy conjectured}, we describe a \emph{strong version} of the monogamy property, which we conjecture to be true. The latter is used in Section \ref{sec: unclonable dec from stronger monogamy} to obtain an unclonable decryption scheme which does not assume extractable witness encryption.

\subsection{Definitions}
\label{sec: affine subspaces defs}

In this subsection, we provide the basic definitions and properties of coset states. 

For any subspace $A$, its complement is $A^\perp = \{ b \in \mathbb{F}^n \,|\,  \langle a, b\rangle \bmod 2 = 0 \,,\, \forall a \in A \}$. It satisfies $\dim(A) + \dim(A^\perp) = n$. We also let $|A| = 2^{\dim(A)}$ denote the size of the subspace $A$.
\begin{definition}[Subspace States]
For any subspace $A \subseteq \mathbb{F}_2^n$, the subspace state $\ket A$ is defined as $$ \ket{A} = \frac{1}{\sqrt{|A|}}\sum_{a \in A} \ket a \,.$$
\end{definition}
Note that given $A$, the subspace state $\ket A$ can be constructed efficiently. 

\begin{definition}[Coset States]
For any subspace $A \subseteq \mathbb{F}_2^n$ and vectors $s, s' \in \mathbb{F}_2^n$, the coset state $\ket {A_{s,s'}}$ is defined as:
\begin{align*}
    \ket {A_{s,s'}} = \frac{1}{\sqrt{|A|}} \sum_{a \in A} (-1)^{\langle s', a\rangle} \ket {a + s}\,.
\end{align*}
\end{definition}

Note that by applying $H^{\otimes n}$, which is {\sf QFT} for $\F_2^n$, to the state $\ket {A_{s,s'}}$, one obtains exactly $\ket {A^{\perp}_{s', s}}$. 

Additionally, note that given $\ket A$ and $s, s'$, one can efficiently construct $\ket {A_{s, s'}}$ as follows: 
\begin{align*}
     & \sum_a \ket a \,\xrightarrow[]{\text{add } s}\, \sum_a \ket {a + s} \,\xrightarrow[]{H^{\otimes n}}\, \sum_{a' \in A^\perp} (-1)^{\langle a', s\rangle} \ket {a'} \\
      \xrightarrow[]{\text{adding } s'}\, & \sum_{a' \in A^\perp} (-1)^{\langle a', s\rangle} \ket {a' + s'} \,\xrightarrow[]{H^{\otimes n}}\, \sum_{a \in A} (-1)^{\langle a, s'\rangle} \ket {a + s}
\end{align*}

For a subspace $A$ and vectors $s,s'$, we define $A+s = \{v +s : v \in A\}$, and $A^{\perp}+s' = \{v +s': v \in A^{\perp}\}$.

It is also convenient for later sections to define a canonical representative, with respect to subspace $A$, of the coset $A+s$.

\begin{definition}[Canonical representative of a coset]
\label{def:canonical_vec_func}
    For a subspace $A$, we define the function $\can_A(\cdot)$ such that $\can_A(s)$ is the lexicographically smallest vector contained in $A + s$ (we call this the canonical representative of coset $A+s$).
\end{definition}
Note that if $\tilde{s} \in A + s$, then $\can_A(s) = \can_A(\tilde{s})$. 
Also note that $\can_A$ is polynomial-time computable given the description of $A$. The algorithm to compute $\can_A$ is the following:
\begin{enumerate}
    \item Initialize the answer to be empty. 
    \item In the first step, let the first entry of the answer be $0$ and check if a vector starting with $0$ is in $A+s$. This can be done efficiently by solving a linear system (by knowing $A$ and $s$). If such a vector is not in $A+s$, let the first entry of the answer be $1$. 
    \item Iterate the same procedure for all entries, and output the answer.
\end{enumerate}

When it is clear from the context, for ease of notation, we will write $A+s$ to mean the \emph{program} that checks membership in $A+s$. For example, we will often write $\iO(A+s)$ to mean an $\iO$ obfuscation of the program that checks membership in $A+s$. 

The following equivalences, which follow straightforwardly from the security of $\iO$, will be useful in our security proofs later on. 
\begin{lemma} \label{lem:io_shO_CC_equivalent}
For any subspace $A \subseteq \mathbb{F}_2^n$, 
\begin{itemize}
    \item $\iO(A+s) \approx_c \iO(\shO_A(\cdot-s))\,,$
    \vspace{1mm}
    
    where $\shO_A()$ denotes the program $\shO(A)$, and $\shO$ is the subspace hiding obfuscator defined in Section \ref{sec: shO}. So, $\shO_A(\cdot-s)$ is the program that on input $x$, runs program $\shO(A)$ on input $x-s$. 
    \item $\iO(A+s) \approx_c \iO(\CC[\can_A, \can_A(s)])\,,$ 
    \vspace{1mm}
    
    where recall that $\CC[\can_A, \can_A(s)]$ refers to the compute-and-compare program which on input $x$ outputs $1$ if and only if $\can_A(x) = \can_A(s)$. 
\end{itemize}
\end{lemma}

\subsection{Direct Product Hardness}
\label{sec: direct product}
In this section, we argue that coset states satisfy both an information-theoretic 
and a computational 
\compuncertain  \,\,property.

\subsubsection{Information-Theoretic \CompUncertain}
\label{sec: it direct product}
\begin{theorem}
\label{thm: direct product info}
Let $A \subseteq \mathbb{F}_2^n$ be a uniformly random subspace of dimension $n/2$, and $s, s'$ be uniformly random in $\mathbb{F}_2^n$. Let $\epsilon > 0$ be such that $1/\epsilon = o(2^{n/2})$. Given one copy of $\ket{A_{s,s'}}$, and a quantum membership oracle for $A+s$ and $A^{\perp}+s'$, an adversary needs $\Omega(\sqrt{\epsilon} 2^{n/2})$ queries to output a pair $(v,w)$ such that $v \in A+s$ and $w \in A^{\perp}+s'$ with probability at least $\epsilon$.
\end{theorem}
The proof is a simple random self-reduction to the analogous statement from Ben-David and Sattath \cite{ben2016quantum} for regular subspace states. The proof is given in Section \ref{sec: coset state direct product proof}.

\subsubsection{Computational \compuncertain}
\label{sec:affine_direct_product_comp}

Next, we present the computational version of the {\compuncertain } property. This establishes that \Cref{thm: direct product info} still holds, even if an adversary is given $\iO$ obfuscations of the subspace membership checking programs.

\begin{theorem}
\label{thm: direct product comp}
{Assume the existence of post-quantum $\iO$ and one-way function.} 
Let $A \subseteq \mathbb{F}_2^n$ be a uniformly random subspace of dimension $n/2$, and $s, s'$ be uniformly random in $\mathbb{F}_2^n$. 
Given one copy of $\ket{A_{s,s'}}$, $\mathsf{iO}(A+s)$ and $\mathsf{iO}(A^{\perp}+s')$, any polynomial time adversary outputs a pair $(v,w)$ such that $v \in A+s$ and $w \in A^{\perp}+s'$ with negligible probability.
\end{theorem}

\subsubsection{Proof of \Cref{thm: direct product info}}
\label{sec: coset state direct product proof}
\label{sec:direct product info}

We first present the theorem from Ben-David and Sattath \cite{ben2016quantum}.
\begin{theorem}[\cite{ben2016quantum}]
\label{thm: direct product bds}
Let $A \subseteq \mathbb{F}_2^n$ be a uniformly random subspace of dimension $n/2$, and let $\epsilon > 0$ be such that $1/\epsilon = o(2^{n/2})$. Given one copy of $\ket{A}$, and a quantum membership oracle for $A$ and $A^{\perp}$, an adversary needs $\Omega(\sqrt{\epsilon} 2^{n/2})$ queries to output a pair $(v,w)$ such that $v \in A \setminus \{0\}$ and $w \in A^{\perp} \setminus \{0\}$ with probability $\epsilon$.
\end{theorem}

\jiahui{begin comment out for camera ready ver}
\begin{proof}[Proof of Theorem \ref{thm: direct product info}]
Let $\adv$ be an adversary for Theorem \ref{thm: direct product info} who suceeds with probability $p$, we construct an adversary $\adv'$ for \Cref{thm: direct product bds} with almost the same success probability making the same number of queries. $\adv'$ proceeds as follows.
\begin{itemize}
    \item $\adv'$ receives $\ket{A}$ for some $A \subseteq \mathbb{F}_2^n$. Samples $s,s'$ uniformly at random, and creates the state $\ket{A_{s,s'}}$.
    \item $\adv'$ gives $\ket{A_{s,s'}}$ as input to $\adv$. $\As$ also needs to get access to oracle $A+s$ and $A^\perp+s'$. $\As'$ can simulate them by having access to $A, A^\perp$ and knowing $s, s'$.
    It receives $v,w$ in return from $\As$. $\adv'$ outputs $(v-s, w-s')$.
\end{itemize}
With probability $p$, $\adv$ returns $v, w$ such that $v \in A+s$ and $w\in A^{\perp} + s'$. Thus the output of $\adv'$ $(v-s, w-s')$ is such that $v-s \in A$ and $w-s' \in A^{\perp}$. All that is left to argue is that with overwhelming probability $v-s \neq 0$ and $w-s \neq 0$. Note that there are $2^{n/2} \cdot 2^{n/2}$ pairs $(\tilde{s},\tilde{s}')$ such that $\ket{A_{\tilde{s}, \tilde{s}'} }= \ket{A_{s, s'}}$, since translating $s$ and $s'$ by an element in $A$ and $A^{\perp}$ respectively does not affect the state. Note further that only $2^{n/2+1} - 1$ pairs are such that $v - \tilde{s} = 0$ or $w - \tilde{s}' = 0$.
Since $s$ and $s'$ are sampled uniformly at random, the probability that $v-s = 0$ or $w-s' = 0$ is $\frac{2^{n/2+1}-1}{2^n}$, which is negligible.
\end{proof}

\subsubsection{Proof of \Cref{thm: direct product comp}}
\label{sec:proof_for_direct_product}

\begin{proof}
We consider the following hybrids.
\begin{itemize}
        \item {Hyb 0:}  This is the game of Theorem \ref{thm: direct product comp}: $A \subseteq \mathbb{F}_2^n$, $s, s'$ are sampled uniformly at random. $\adv$ receives $\iO(A+s)$,  $\iO(A^{\perp}+s')$, and $\ket{A_{s,s'}}$. $\adv$ wins if it returns $(v,w) \in (A +s) \times (A^{\perp}+s')$.
        
        \item {Hyb 1:} Same as Hyb 0 except $\adv$ gets $\iO(\shO_A(\cdot-s))$, $\iO(A^\perp+s')$ and $\ket{A_{s,s'}}$. Recall that $\shO_A$ is the program $\shO(A)$, and so $\shO_A(\cdot-s))$ is the program that on input $x$, runs program $\shO(A)$ on input $x-s$.%

        \item {Hyb 2:} Same as Hyb 1 except $\adv$ gets $\iO(\shO_B(\cdot-s))$, $\iO(A^\perp+s')$ and $\ket{A_{s,s'}}$, for a uniformly random superspace $B$ of $A$, of dimension $3/4n$.
        
        \item {Hyb 3:} Same as Hyb 2 except for the following. The challenger samples $s,s', A$, and a uniformly random superspace $B$ of $A$ as before. The challenger sets $t = s + w_B$, where $w_B \leftarrow B$. Sends $\iO(\shO_B(\cdot-t))$, $\iO(A^\perp+s')$ and $\ket{A_{s,s'}}$ to $\adv$.
        
        \item {Hyb 4:} Same as Hyb 3 except $\adv$ gets $\iO(\shO_B(\cdot-t))$, $\iO(\shO_{A^{\perp}}(\cdot-s'))$ and $\ket{A_{s,s'}}$.
        
        \item {Hyb 5:} Same as Hyb 4 except $\adv$ gets $\iO(\shO_{B}(\cdot-t))$, $\iO(\shO_{C^{\perp}}(\cdot-s'))$ and $\ket{A_{s,s'}}$, for a uniformly random superspace $A^\perp \subseteq C^\perp$ of dimension $3n/4$.
        
        \item {Hyb 6:} Same as Hyb 5 except for the following. The challenger sets $t' = s' + w_{C^{\perp}}$, where $ w_{C^{\perp}} \leftarrow C^{\perp}$. $\adv$ gets $\iO(\shO_B(\cdot-t))$, $\iO(\shO_{C^{\perp}}(\cdot-t'))$ and $\ket{A_{s,s'}}$. 
        
        \item {Hyb 7:} Same as Hyb 6 except the challenger sends $B, C, t,t'$ in the clear to $\adv$.
    \end{itemize}

\begin{claim}
\label{lem: direct product hyb 0-1}
For any QPT adversary $\adv$,
$$\left|\Pr[\adv \text{ wins in Hyb 1}] - \Pr[\adv \text{ wins in Hyb 0}] \right| = \negl(\lambda) \,.$$
\end{claim}

\begin{proof}
Suppose for a contradiction there was a QPT adversary $\adv$ such that:
\begin{equation}
\label{eq: difference 1-directproduct}
    \left|\Pr[\adv \text{ wins in Hyb 1}] - \Pr[\adv \text{ wins in Hyb 0}] \right| 
\end{equation}
is non-negligible. 
Such an adversary can be used to construct $\adv'$ which distinguishes $\mathsf{iO}(A+s)$ from $\mathsf{iO}(\shO_A(\cdot -s))$, which is impossible by the security of the (outer) $\mathsf{iO}$, since $A+s$ and $\shO_A(\cdot -s)$ compute the same functionality.

Fix $n$, let $A \subseteq \mathbb{F}_2^n$, $s,s' \in \mathbb{F}_2^n$ be such that the difference in \eqref{eq: difference 1-directproduct} is maximized. Suppose $\Pr[\adv \text{ wins in Hyb 1}] > \Pr[\adv \text{ wins in Hyb 0}]$, the other case being similar.

$\adv'$ proceeds as follows:
\begin{itemize}
    \item Receives as a challenge a circuit $P$ which is either $\mathsf{iO}(A+s)$ or $ \mathsf{iO}(\mathsf{shO}_A(\cdot - s))$. Creates the state $\ket{A_{s,s'}}$. Gives $P$, $\iO(A^{\perp}+s')$ and $\ket{A_{s,s'}}$ as input to $\adv$.
    \item $\adv$ returns a pair $(v,w)$. If $v \in A+s$ and $w \in A^{\perp} +s'$, then $\adv'$ guesses that $P = \mathsf{iO}(\mathsf{shO}_A(\cdot - s))$, otherwise that $P = \mathsf{iO}(A+s)$.
\end{itemize}
It is straightforward to verify that $\adv'$ succeeds at distinguishing with non-negligible probability.
\end{proof}

\begin{claim}
\label{lem: direct product hyb 1-2}
For any QPT adversary $\adv$,
$$\left|\Pr[\adv \text{ wins in Hyb 2}] - \Pr[\adv \text{ wins in Hyb 1}] \right| = \negl(\lambda) \,.$$
\end{claim}

\begin{proof}
Suppose for a contradiction there was a QPT adversary $\adv$ such that:
\begin{equation*}
    \left|\Pr[\adv \text{ wins in Hyb 2}] - \Pr[\adv \text{ wins in Hyb 1}] \right| \,, 
\end{equation*}
is non-negligible. 

We argue that $\adv$ can be used to construct an adversary $\adv'$ that breaks the security of $\shO$. 

Fix $n$. Suppose $\Pr[\adv \text{ wins in Hyb 2}] > \Pr[\adv \text{ wins in Hyb 1}]$, the other case being similar. 

$\adv'$ proceeds as follows:
\begin{itemize}
    \item Sample $A \subseteq \mathbb{F}_2^n$ uniformly at random. Send $A$ to the challenger.
    \item The challenger returns a program $P$ which is either $\shO_A$ or $\shO_B$. $\adv'$ samples uniformly $s,s' \in \mathbb{F}_2^n$, and creates the state $\ket{A_{s,s'}}$. Gives $\iO(P(\cdot - s))$, $\iO(A^{\perp}+s')$ and $\ket{A_{s,s'}}$ as input to $\adv$.
    \item $\adv$ returns a pair $(v,w)$. If $v \in A+s$ and $w \in A^{\perp} + s'$, then $\adv'$ guesses that $P = \shO_B$, otherwise that $P = \shO_A$.
\end{itemize}
It is straightforward to verify that $\adv'$ succeeds at the security game for $\shO$ with non-negligible advantage.
\end{proof}

\begin{claim}
\label{lem: direct product hyb 2-3}
For any QPT adversary $\adv$,
$$\left|\Pr[\adv \text{ wins in Hyb 3}] - \Pr[\adv \text{ wins in Hyb 2}] \right| = \negl(\lambda) \,.$$
\end{claim}
\begin{proof}
The proof is similar to the proof of Lemma \ref{lem: direct product hyb 0-1}, and follows from the security of $\iO$ and the fact that  $\shO_B(\cdot - s)$ and $\shO_B(\cdot - t)$ compute the same functionality. This is because for any vector $w_B \in B$, $B + w_b$ is the same subspace as $B$.
\end{proof}

\begin{claim}
For any QPT adversary $\adv$, and $j = 4, 5, 6$, we have
\begin{align*}
\left|\Pr[\adv \text{ wins in Hyb j}] - \Pr[\adv \text{ wins in Hyb (j-1)}] \right| = \negl(\lambda)
\end{align*}
\end{claim}
\begin{proof}
    The proofs are analogous to those of Lemmas \ref{lem: direct product hyb 0-1}, \ref{lem: direct product hyb 1-2},  \ref{lem: direct product hyb 2-3}.
\end{proof}

%

\begin{lemma}
For any QPT adversary $\adv$ for Hyb 6, there exists an adversary $\adv'$ for Hyb 7 such that
$$ \Pr[\adv' \text{ wins in Hyb 7}] \geq \Pr[\adv \text{ wins in Hyb 6}] \,.$$
\end{lemma}
\begin{proof}
This is immediate.
\end{proof}

\begin{lemma}
For any (unbounded) adversary $\adv$,
$$\Pr[\adv \text{ wins in Hyb 7}] = \negl(\lambda) \,.$$
\end{lemma}
\begin{proof}
Suppose there exists an adversary $\adv$ for Hyb 7 that wins with probability $p$. 

We first show that, without loss of generality, one can take $B$ to be the subspace of vectors such that the last $n/4$ entries are zero (and the rest are free), and one can take $C$ to be such that the last $3/4n$ entries are zero (and the rest are free). We construct the following adversary $\adv'$ for the game where $B$ and $C$ have the special form above with trailing zeros, call these $B_*$ and $C_*$, from an adversary $\adv$ for the game of Hyb 7. 
\begin{itemize}
\item $\adv'$ receives a state $\ket{A_{s,s'}}$, together with $t$ and $t'$, for some $C_* \subseteq A \subseteq B_*$, where $t = s + w_{B_*}$ for $w_{B_*} \leftarrow B_*$, and $t' = s' + w_{C_*^{\perp}}$, where $w_{C_*^{\perp}} \leftarrow C_*^{\perp}$.
\item $\adv'$ picks uniformly random subspaces $B$ and $C$ of dimension $\frac{3}{4}n$ and $\frac{n}{4}$ respectively such that $C \subseteq B$, and a uniformly random isomorphism $\mathcal{T}$ mapping $C_*$ to $C$ and $B_*$ to $B$ (which can be sampled efficiently). We think of $\mathcal{T}$ as a change-of-basis matrix (in particular when we take its transpose). $\adv'$ applies to  $\ket{A_{s,s'}}$ the unitary $U_{\mathcal{T}}$ which acts as $\mathcal{T}$ on the standard basis elements. $\adv'$ gives $U_{\mathcal{T}}\ket{A}$ to $\adv$ together with $B$, $C$, $\mathcal{T}(t)$ and $(\mathcal{T}^{-1})^T(t')$. $\adv'$ receives a pair $(v,w)$ from $\adv$. $\adv'$ outputs $(\mathcal{T}^{-1}(v),\mathcal{T}^T(w))$.
\end{itemize}

First, notice that
\begin{align*}
    U_{\mathcal{T}} \ket{A_{s,s'}} &= U_{\mathcal{T}}  \sum_{v \in A} (-1)^{\langle v,s' \rangle}\ket{v+s} \\
    & =  \sum_{v \in A} (-1)^{\langle v,s' \rangle}\ket{\mathcal{T}(v)+\mathcal{T}(s)} \\
    & =  \sum_{w \in \mathcal{T}(\mathcal{A})} (-1)^{\langle \mathcal{T}^{-1}(w), s' \rangle}\ket{w+\mathcal{T}(s)} \\
    & = \sum_{w \in \mathcal{T}(A)} (-1)^{\langle w, (\mathcal{T}^{-1})^T(s') \rangle}\ket{w+\mathcal{T}(s)} \\
    & = \ket{\mathcal{T}(A)_{z, z'}} \,,
\end{align*}
where $z = \mathcal{T}(s)$ and $z' = (\mathcal{T}^{-1})^T (s')$.

Notice that $\mathcal{T}(A)$ is a uniformly random subspace between $C$ and $B$, and that $z$ and $z'$ are uniformly random vectors in $\mathbb{F}_2^n$. Moreover, we argue that:
\begin{itemize}
\item[(i)] $\mathcal{T}(t)$ is distributed as a uniformly random element of $z+B$.
\item[(ii)] $(\mathcal{T}^{-1})^T(t')$ is distributed as a uniformly random element of $z' + C^{\perp}$. 
\end{itemize}

For (i), notice that 
$$\mathcal{T}(t) = \mathcal{T}(s+w_{B_*}) = \mathcal{T}(s)+\mathcal{T}(w_{B_*}) = z + \mathcal{T}(w_{B_*})\,,$$ where $w_{B_*}$ is uniformly random in $B_*$. Since $\mathcal{T}$ is an isomorphism with $\mathcal{T}(B_*) = B$, then $\mathcal{T}(w_{B_*})$ is uniformly random in $B$. Thus, $\mathcal{T}(t)$ is distributed as a uniformly random element in $z+B$.

For (ii), notice that 
$$(\mathcal{T}^{-1})^T(t') = (\mathcal{T}^{-1})^T(s'+w_{C_*^{\perp}}) = (\mathcal{T}^{-1})^T(s') + (\mathcal{T}^{-1})^T(w_{C_*^{\perp}}) = z' + (\mathcal{T}^{-1})^T(w_{C_*^{\perp}}) \,,$$
where $w_{C_*^{\perp}}$ is uniformly random in $C_*^{\perp}$. We claim that $(\mathcal{T}^{-1})^T(w_{C_*^{\perp}})$ is uniformly random in $C^{\perp}$. Notice, first, that the latter belongs to $C^{\perp}$. Let $x \in C$, then 
$$\langle (\mathcal{T}^{-1})^T(w_{C_*^{\perp}}), x \rangle = \langle w_{C_*^{\perp}}, \mathcal{T}^{-1} (x) \rangle  = 0\,,$$
where the last equality follows because $w_{C_*^{\perp}} \in C_*^{\perp}$, and $\mathcal{T}^{-1}(C) = C_*$. The claim follows from the fact that $(\mathcal{T}^{-1})^T$ is a bijection.

Hence, $\adv$ receives inputs from the correct distribution, and thus, with probability $p$, $\adv$ returns a pair $(v,w)$ such that $v \in \mathcal{T}(A)+z$ and $w \in \mathcal{T}(A)^{\perp} + z'$, where $z = \mathcal{T}(s)$ and $z' \in (\mathcal{T}^{-1})^T (s')$. $\adv'$ returns $(v', w') = (\mathcal{T}^{-1}(v),  \mathcal{T}^{T}(w))$.

Notice that:
\begin{itemize}
    \item If $v \in \mathcal{T}(A)+z$, where $z = \mathcal{T}(s)$, then $\mathcal{T}^{-1}(v) \in A + s$.
    \item If $w \in \mathcal{T}(A)^{\perp} + z'$, where $z' \in (\mathcal{T}^{-1})^T (s')$, then $\mathcal{T}^{T}(w) \in A^{\perp} + s'$. This is because, for any $u \in A$:
    $$ \langle \mathcal{T}^T(w), u \rangle = \langle w, \mathcal{T}(u) \rangle = 0 \,. $$
\end{itemize}

Thus, with probability $p$, $\adv'$ returns a pair $(v',w')$ where $v' \in A + s$ and $w' \in A^{\perp} + s'$, as desired.
\vspace{2mm}

So, we can now assume that $B$ is the space of vectors such that the last $\frac{n}{4}$ entries are zero, and $C$ is the space of vectors such that the last $\frac34 n$ entries are zero. Notice then that the sampled subspace $A$ is uniformly random subspace subject to the last $\frac{n}{4}$ entries being zero, and the first $\frac{n}{4}$ entries being free. From an adversary $\adv$ for Hybrid 7 with such $B$ and $C$, we will construct an adversary $\adv'$ for the information-theoretic direct-product game where the ambient subspace is $\mathbb{F}_2^{n'}$, where $n' = \frac{n}{2}$.

\begin{itemize}
    \item $\adv'$ receives $\ket{A_{s,s'}}$, for uniformly random $A \subseteq \mathbb{F}_2^{n'}$ of dimension $n'/2$ and uniformly random $s, s' \in \mathbb{F}_2^{n'}$. $\adv'$ samples $\tilde{s}, \tilde{s}', \hat{s}, \hat{s}' \leftarrow \mathbb{F}_2^{\frac{n}{4}}$. 
    
    Let  $\ket{\phi} = \frac{1}{2^{n/8}} \sum_{x \in \{0,1\}^{n/4} } (-1)^{\langle x,\tilde{s}' \rangle}\ket{x + \tilde{s}}$. $\adv_0'$ creates the state 
$$ \ket{W} = \ket{\phi} \otimes \ket{A_{s,s'}} \otimes \ket{\hat{s}}\,,$$
$\adv'$ gives to $\adv$ as input the state $\ket{W}$, together with $t = 0^{3n/4}|| \hat{s} + w_B$ for $w_B \leftarrow B$ and $t' = \hat{s}'||0^{3n/4} + w_{C^{\perp}}$, for $w_{C^{\perp}} \leftarrow C^{\perp}$. $\adv$ returns a pair $(v,w) \in \mathbb{F}_2^{n} \times \mathbb{F}_2^{n}$. Let $v' = [v]_{\frac{n}{4}+1, \frac{3}{4}n} \in \mathbb{F}_2^{n/2}$ be the ``middle'' $n/2$ entries of $v$. Let $w' =  [w]_{\frac{n}{4}+1, \frac{3}{4}n} \in \mathbb{F}_2^{n/2} $. $\adv'$ outputs $(v', w')$. 
\end{itemize}

Notice that 
\begin{align*}
\ket{W} &=  \ket{\phi} \otimes \ket{A_{s,s'}} \otimes \ket{\hat{s}} \\
&= \sum_{x \in \{0,1\}^{n/4}, v \in A} (-1)^{\langle x,\tilde{s}' \rangle} (-1)^{\langle v,s' \rangle}\Big| (x+\tilde{s})||(v+s)||\hat{s}\Big\rangle \\
&= \sum_{x \in \{0,1\}^{n/4}, v \in A} (-1)^{\langle (x||v||0^{n/4}), (\tilde{s}'|| s'||\hat{s}') \rangle} \Big| x||v||0^{n/4} + \tilde{s}||s||\hat{s} \Big \rangle \\
&= \sum_{w \in \tilde{A}} (-1)^{\langle w, z' \rangle} \ket{w + z} = \ket{\tilde{A}_{z,z'}} \,,
\end{align*}
where $z = \tilde{s}||s||\hat{s}$, $z' = \tilde{s}'||s'||\hat{s}'$, and $\tilde{A} \subseteq \mathbb{F}_2^n$ is the subspace in which the first $n/4$ entries are free, the middle $n/2$ entries belong to subspace $A$, and the last $n/4$ entries are zero (notice that there is a freedom for the choice of $s'$ in the above calculation).

Notice that the subspace $\tilde{A}$, when averaging over the choice of $A$, is distributed precisely as in the game of Hybrid 7 (with the special choice of $B$ and $C$); $z, z'$ are uniformly random in $\mathbb{F}_2^{n}$; $t$ is uniformly random from $z + B$, and $t'$ is uniformly random from $z'+C^{\perp}$. Thus, with probability $p$, $\adv$ returns to $\adv'$ a pair $(v,w)$ such that $v \in \tilde{A}+z$ and $w \in \tilde{A}^{\perp} + z'$. It follows that, with probability $p$, the answer $(v',w')$ returned by $\adv'$ is such that $v' \in A+s$ and $w' \in A^{\perp} + s'$.

Thus, by Theorem \ref{thm: direct product info}, we deduce that $p$ must be negligible.

\end{proof}

Therefore, we have shown that the advantage in distinguishing Hybrid 0 and Hybrid 6 is negligible, and the success probability in Hybrid 6 is at most the success probability in Hybrid 7, which is negligible). Hence, the probability of success in the original game is also negligible. 
\end{proof}

\jiahui{comment out for camera ready ver}

\subsection{Monogamy-of-Entanglement Property}
\label{sec: monogamy}
In this subsection, we argue that coset states satisfy an information-theoretic 
and a computational 
monogamy-of-entanglement property. We will not make use of these properties directly, instead we will have to rely on a stronger conjectured monogamy-of-entanglement property, which is presented in subsection \ref{sec: monogamy conjectured}. Thus, the properties that we prove in this subsection serve merely as ``evidence'' in support of the stronger conjecture.

\subsubsection{Information-Theoretic Monogamy-of-Entanglement}
\label{sec: it monogamy}
Let $n \in \mathbb{N}$. Consider the following game between a challenger and an adversary $(\Advz, \adv_1, \adv_2)$.
\begin{itemize}
    \item The challenger picks a uniformly random subspace $A \subseteq \mathbb{F}_2^n$ of dimension $\frac{n}{2}$, and two uniformly random elements $s, s' \in \mathbb{F}_2^n$. Sends $\ket{A_{s,s'}}$ to $\Advz$. 
    \item $\Advz$ creates a bipartite state on registers $\mathsf{B}$ and $\mathsf{C}$. Then, $\Advz$ sends register $\mathsf{B}$ to $\adv_1$, and $\mathsf{C}$ to $\adv_2$. 
    \item The description of $A$ is then sent to both $\As_1, \As_2$. 
    \item $\adv_1$ and $\adv_2$ return respectively $(s_1,s_1')$ and $(s_2, s_2')$.
\end{itemize}
$(\Advz, \adv_1, \adv_2)$ wins if, for $i \in \{1,2\}$, $s_i \in A +s$ and  $s_i' \in A^{\perp} + s' \,.$
\vspace{2mm}

Let $\mathsf{ITMonogamy}((\Advz, \adv_1, \adv_2), n)$ be a random variable which takes the value $1$ if the game above is won by adversary $(\Advz, \adv_1, \adv_2)$, and takes the value $0$ otherwise. We have the following theorem. 
\begin{theorem}
\label{thm: monogamy info}
There exists a sub-exponential function $\subexp$ such that, for any (unbounded) adversary $(\Advz, \adv_1, \adv_2)$, 
$$\Pr[\mathsf{ITMonogamy}((\Advz, \adv_1, \adv_2), n) = 1] \leq 1/\subexp(n)\,.$$
\end{theorem}

We refer the reader to \Cref{proof:monogamy} for the proof.

\subsubsection{Computational monogamy}
\label{sec: comp monogamy}
We describe a computational version of the monogamy game from the previous section. In the computational version, $\adv_0$ additionally receives the programs $\iO(A+s)$ and $\iO(A'+s')$. The game is between a challenger and an adversary $(\Advz, \adv_1, \adv_2)$.
\begin{itemize}
    \item The challenger picks a uniformly random subspace $A \subseteq \mathbb{F}^n$ of dimension $\frac{n}{2}$, and two uniformly random elements $s, s' \in \mathbb{F}_2^n$. It sends $\ket{A_{s,s'}}$, $\iO(A+s)$, and $\iO(A^{\perp}+s')$ to $\Advz$.
    \item $\Advz$ creates a bipartite state on registers $\mathsf{B}$ and $\mathsf{C}$. Then, $\Advz$ sends register $\mathsf{B}$ to $\adv_1$, and $\mathsf{C}$ to $\adv_2$. 
    \item The description of $A$ is then sent to both $\As_1, \As_2$. 
    \item $\adv_1$ and $\adv_2$ return respectively $(s_1,s_1')$ and $(s_2, s_2')$.
\end{itemize}
$(\Advz, \adv_1, \adv_2)$ wins if, for $i \in \{1,2\}$, $s_i \in A +s$ and  $s_i' \in A^{\perp} + s' \,.$
\vspace{2mm}

Let $\mathsf{CompMonogamy}((\Advz, \adv_1, \adv_2), n)$ be a random variable which takes the value $1$ if the game above is won by adversary $(\Advz, \adv_1, \adv_2)$, and takes the value $0$ otherwise.

\begin{theorem} \label{thm: monogamy comp}
{Assume the existence of post-quantum $\iO$ and one-way function,} there exists a negligible function $\negl(\cdot)$, for any QPT adversary $(\Advz, \adv_1, \adv_2)$, 
$$\Pr[\mathsf{CompMonogamy}((\Advz, \adv_1, \adv_2), n) =1] = \negl(n)\,.$$
\end{theorem}

The proof is very similar to the proof of \Cref{thm: direct product comp}. We refer the reader to \Cref{sec: appendix monogamy comp} for the full details.

\subsection{Conjectured Strong Monogamy Property}
\label{sec: monogamy conjectured}
\jiahui{add comment about the follow-up?}
In this section, we describe a stronger version of the monogamy property, which we conjecture to hold.  The monogamy property is a slight (but significant) variation of the one stated in the last section (which we proved to be true). Recall that there $\adv_1$ and $\adv_2$ are required to return pairs $(s_1, s_1')$ and $(s_2, s_2')$ respectively, such that both $s_1, s_2 \in A+s$ and $s_1', s_2' \in A^{\perp} + s'$. Now, we require that it is hard for $\adv_1$ and $\adv_2$ to even return a single string $s_1$ and $s_2$ respectively such that $s_1 \in A+s$ and $s_2 \in A^{\perp} + s'$.

Formally, consider the following game between a challenger and an adversary $(\Advz, \adv_1, \adv_2)$. 
\begin{itemize}
    \item The challenger picks a uniformly random subspace $A \subseteq \mathbb{F}_2^n$ of dimension $\frac{n}{2}$, and two uniformly random elements $s, s' \in \mathbb{F}_2^n$. It sends $\ket{A_{s,s'}}$ to $\Advz$.
    \item $\Advz$ creates a bipartite state on registers $\mathsf{B}$ and $\mathsf{C}$. Then, $\Advz$ sends register $\mathsf{B}$ to $\adv_1$, and $\mathsf{C}$ to $\adv_2$. 
    \item The description of $A$ is then sent to both $\As_1, \As_2$. 
    \item $\adv_1$ and $\adv_2$ return respectively $s_1$ and $s_2$.
\end{itemize}
Let $\mathsf{ITStrongMonogamy}((\Advz, \adv_1, \adv_2), n)$ be a random variable which takes the value $1$ if the game above is won by adversary $(\Advz, \adv_1, \adv_2)$, and takes the value $0$ otherwise. We conjecture the following: 
\begin{conjecture}
\label{conj:strong_monogamy_it}
There exists a sub-exponential function $\subexp$ such that, for any (unbounded) adversary $(\Advz, \adv_1, \adv_2)$, 
$$\Pr[\mathsf{ITStrongMonogamy}((\Advz, \adv_1, \adv_2), n) = 1] \leq 1/\subexp(n)\,.$$
\end{conjecture}

\begin{remark}
This conjecture is later proved in a follow-up work by Culf and Vidick after the first version of this paper. We refer the readers to \cite{culfvidick2021cosetsproof} for details of the proof.
\end{remark}

Assuming the conjecture is true, and assuming post-quantum $\iO$ and one-way functions, we are able to prove the following computational strong monogamy statement.  Consider a game between a challenger and an adversary $(\Advz, \adv_1, \adv_2)$, which is identical to the one described above except that all $\Advz$ additionally gets the membership checking programs $\iO(A+s)$ and $\iO(A^{\perp}+s')$. 

\begin{itemize}
    \item The challenger picks a uniformly random subspace $A \subseteq \mathbb{F}_2^n$ of dimension $\frac{n}{2}$, and two uniformly random elements $s, s' \in \mathbb{F}_2^n$. It sends $\ket{A_{s,s'}}$, $\iO(A+s)$, and $\iO(A^{\perp}+s')$ to $\Advz$.
    \item $\Advz$ creates a bipartite state on registers $\mathsf{B}$ and $\mathsf{C}$. Then, $\Advz$ sends register $\mathsf{B}$ to $\adv_1$, and $\mathsf{C}$ to $\adv_2$. 
    \item The description of $A$ is then sent to both $\As_1, \As_2$. 
    \item $\adv_1$ and $\adv_2$ return respectively $s_1$ and $s_2$.
\end{itemize}
$(\Advz, \adv_1, \adv_2)$ wins if, for $s_1 \in A +s$ and $s_2 \in A^{\perp} + s'$.
\vspace{2mm}

Let $\mathsf{CompStrongMonogamy}((\Advz, \adv_1, \adv_2), n)$ be a random variable which takes the value $1$ if the game above is won by adversary $(\Advz, \adv_1, \adv_2)$, and takes the value $0$ otherwise.

\begin{theorem}
\label{conj:strong_monogamy}
Assuming Conjecture \ref{conj:strong_monogamy_it} holds, and assuming the existence of post-quantum $\iO$ and one-way functions, then there exists a negligible function $\negl(\cdot)$,
for any QPT adversary $(\Advz, \adv_1, \adv_2)$, 
$$\Pr[\mathsf{CompStrongMonogamy}((\Advz, \adv_1, \adv_2), n) = 1] = \negl(n)\,.$$
\end{theorem}

We can further show a `sub-exponential strong monogamy property' if we additionally assume sub-exponentially secure $\iO$ and one-way functions. 
%

\begin{theorem}
\label{conj:subexp_strong_monogamy}
Assuming Conjecture \ref{conj:strong_monogamy_it} holds, and assuming the existence of sub-exponentially secure post-quantum $\iO$ and one-way functions, then
for any QPT adversary $(\Advz, \adv_1, \adv_2)$, 
$$\Pr[\mathsf{CompStrongMonogamy}((\Advz, \adv_1, \adv_2), n) = 1] \leq 1/\subexp(n)\,.$$
\end{theorem}
The proof is almost identical to that of \Cref{thm: monogamy comp}, therefore we omit the proof here and refer to the proof of %
\Cref{thm: monogamy comp}.

In the rest of the work, whenever we mention `strong monogamy property' or `strong monogamy-of-entanglement property', we refer to the computational monogamy property in \Cref{conj:strong_monogamy} above. Whenever  we mention `sub-exponentially strong monogamy property' or `sub-exponentially  strong monogamy-of-entanglement property', we refer to the computational monogamy property in \Cref{conj:subexp_strong_monogamy}.
\section{Tokenized Signature Scheme from iO}
\label{sec: signature tokens}

\revise{In this section, we present a construction for tokenized signatures with unforgeability security based on the computational \compuncertain  (\Cref{thm: direct product comp}). We improved upon the scheme in  \cite{ben2016quantum} by removing the need of (highly structured) oracles or post-quantum VBB obfuscation. 
} 

\subsection{Definitions}
\begin{definition}[Tokenized signature scheme] 
\label{def: ts}A tokenized signature (TS) scheme consists of a tuple of QPT algorithms $(\kgen, \tokengen, \sign, \verify)$ with the following properties:
\begin{itemize}
    \item $\kgen(1^\param) \to (\sk, \pk)$: Takes as input $1^{\lambda}$, where $\lambda$ is a security parameter, and outputs a secret key, public (verification) key pair $(\sk,\pk)$.
    
    \item $\tokengen(\sk) \to \qtoken$: Takes as input a secret key $\sk$ and outputs a signing token $\qtoken$. 
    
    \item $\sign(m, \qtoken) \to (m, \sig)/\bot$: Takes as input a message $m \in \{0,1\}^*$ and a token $\qtoken$, and outputs either a message, signature pair $(m, \sig)$ or $\bot$.
    
    \item $\verify(\pk, m, \sig) \to 0/1$: Takes as input an verification key, an alleged message, signature pair $(m,\sig)$, and outputs $0$ (``reject'') or $1$ (``accept'').
    
    \item $\revoke(\pk, \qtoken) \to 0/1$:  Takes in public key $\pk$ and a claimed token $\qtoken$, and outputs $0$ (``reject'') or $1$ (``accept'').
\end{itemize}
\end{definition}

These algorithms satisfy the following. First is correctness. There exists a negligible function $\negl(\cdot)$, for any $\lambda \in \mathbb{N}$, $m \in \{0,1\}^*$,
\begin{align*}
    \Pr[\verify(\pk, m, \sig) = 1: 
    & (m,\sig) \leftarrow \sign(m, \qtoken), \qtoken \leftarrow \tokengen(\sk), \\ 
    &(\sk, \pk) \leftarrow \kgen(1^{\lambda})  ] 
    \geq 1 - \negl(\lambda) \,. \nonumber \label{eq: ts scheme}
\end{align*}

\begin{definition}[Length restricted TS scheme]
A TS scheme is $r$-restricted if it holds only for $m \in \{0,1\}^r$. We refer to a scheme that is $1$-restricted as a one-bit TS scheme.
\end{definition}

Notation-wise, we introduce an additional algorithm $\verify_\ell$. The latter takes as input a public key $\pk$ and $\ell$ pairs $(m_\ell, \sig_\ell), \ldots, (m_\ell, \sig_\ell)$. It checks that $m_i \neq m_j$ for all $i\neq j$, and $\verify (m_i, \sig_i) = 1$ for all $i \in [\ell]$; it outputs $1$ if and only if they all hold. %

Next we define unforgeability. 

\begin{definition}[$1$-Unforgeability] A TS scheme is $1$-unforgeable if for every QPT adversary $\adv$, there exists a negligible function $\negl(\cdot)$, for every $\lambda$:
\begin{align*}
    \Pr\left[\begin{array}{cc} (m_0, \sig_0, m_1, \sig_1) \gets \As(\pk, \ket {\sf tk}) \\  \verify_2(\pk, m_0,\sig_0 ,m_1, \sig_1) = 1 \\ \end{array}: \begin{array}{cc} (\sk, \pk) \leftarrow \kgen(1^{\lambda})\\ \ket{\mathsf{tk}} \leftarrow \tokengen(\sk) \end{array} \right] \leq \negl(\lambda) \,.
\end{align*}
\end{definition}

\begin{definition}[Unforgeability] 
\label{def: ts unforgeability}A TS scheme is unforgeable if for every QPT adversary $\adv$, there exists a negligible function $\negl(\cdot)$, for every $\lambda$, $l = \poly(\lambda)$:
\begin{align*}
    \Pr\left[ \begin{array}{cc} \{m_i, \sig_i\}_{i \in [l+1]} \gets \As(\pk, \{\ket {{\sf tk}_i}\}_{i \in [l]})  \\ \verify_{l+1}(\pk, \{m_i, \sig_i\}_{i \in [l + 1]}) = 1 \end{array}: \begin{array}{cc}
   (\sk, \pk) \leftarrow \kgen(1^{\lambda}) \\
    \ket {{\sf tk}_1} \gets \tokengen(\sk) \\
    \vdots \\
    \ket {{\sf tk}_l} \gets \tokengen(\sk) 
    \end{array} \right] \leq \negl(\lambda) \,.
\end{align*}
\end{definition}

Finally we have revocability. 

\begin{definition}[Revocability]
A revocable tokenized signature scheme satisfies:
\begin{itemize}
    \item Correctness:
    
       $ \Pr\left[\revoke(\pk, \qtoken) = 1 \middle| (\pk,\sk) \gets \kgen(1^\lambda), \qtoken \gets \tokengen(\sk) \right] = 1$.
       
      \item Revocability: 
       For every $\ell \leq \poly(\lambda), t \leq \ell$, and every QPT $\cA$ with $\ell$ signing tokens $\ket{\tk_1} \otimes \cdots \otimes \ket{\tk_\ell}$ and $\pk$, which has generated $t$ signatures $(m_1,\sig_1), \cdots,$ $(m_t, \sig_t)$ and a state $\sigma$:
    \begin{align*}
    \Pr\left[ \verify_t(\pk, (m_1,\sig_1), \cdots, (m_t, \sig_t)) = 1  \wedge \revoke_{\ell-t+1}(\sigma) = 1 \right] \leq \negl(\lambda)
     \end{align*}
    Here $\revoke_{\ell-t+1}$ means applying $\revoke$ on all $\ell-t+1$ registers of $\sigma$, and outputs $1$ if they all output $1$.
\end{itemize}
\end{definition}

\revise{
The revocability property follows straightforwardly from unforgeability \cite{ben2016quantum}. Thus to show a construction is secure, we only need to focus on proving unforgeability. The following theorem says $1$-unforgeability is sufficient to achieve a full blown TS scheme. 
}

\begin{theorem}[\cite{ben2016quantum}]
\label{thm: onebit onetime to full}
A one-bit  $1$-unforgeable TS scheme implies a (full blown) TS scheme, assuming the existence of a quantum-secure digital signature scheme.
\end{theorem}

In the next section, we give our construction of a one-bit  $1$-unforgeable TS scheme from coset states. 

\subsection{Tokenized Signature Construction}
\begin{construction*} 
$\,$

\label{cons: token sig}
\begin{itemize}
    \item $\mathsf{KeyGen}(1^\param)$: Set $n = \poly(\lambda)$. Sample uniformly $A \subseteq \mathbb{F}_2^n$. Sample $s,s' \leftarrow \mathbb{F}_2^n$. Output $\sk = (A, s,s')$ (where by $A$ we mean a description of the subspace $A$) and $\pk = (\iO(A+s), \iO(A^{\perp}+s'))$.
    
    \item $\tokengen(\sk)$: Takes as input $\sk$ of the form $(A, s,s')$. Outputs $\qtoken = \ket{A_{s,s'}}$.
    
    \item $\sign(m, \qtoken)$: Takes as input $m \in \{0,1\}$ and a state $\qtoken$ on $n$ qubits. Compute $H^{\otimes n} \qtoken$ if $m=1$, otherwise do nothing to the quantum state. It then measures in the standard basis. Let $\sig$ be the outcome. Output $(m, \sig)$.
    
    \item $\verify(\pk, (m, \sig))$: Parse $\pk$ as $\pk = (C_0, C_1)$ where $C_0$ and $C_1$ are circuits. Output $C_m(\sig)$.

    \item $\revoke(\pk, \qtoken)$: Parse $\pk$ as $\pk = (C_0, C_1)$. Then:
    \begin{itemize}
        \item Coherently compute $C_0$ on input $\qtoken$, and measure the output of the circuit. If the latter is $1$, uncompute $C_0$, and proceed to the next step. Otherwise halt and output $0$.
        \item Apply $H^{\otimes n}$. Coherently compute $C_1$ and measure the output of the circuit. If the latter is $1$, output $1$.
    \end{itemize}
\end{itemize}
\end{construction*}

\begin{theorem}
\label{thm: sig tokens main}
Assuming {post-quantum} $\iO$ {and one-way function}, the scheme of Construction \ref{cons: token sig} is a one-bit $1$-unforgeable tokenized signature scheme. 
\end{theorem}

\begin{proof}
Security follows immediately from Theorem \ref{thm: direct product comp}.
\end{proof}

\begin{corollary}
Assuming {post-quantum} $\iO$, {one-way function(which implies digital signature)} and a quantum-secure digital signature scheme, there exists a (full blown) tokenized signature scheme.
\end{corollary}

\begin{proof}
This is an immediate consequence of Theorems \ref{thm: onebit onetime to full} and \ref{thm: sig tokens main}.
\end{proof}
\section{Single-Decryptor Encryption}
\label{sec:unclonable_dec}

In this section, we formally introduce unclonable decryption, i.e. single-decryptor encryption \cite{georgiou-zhandry20}. Then we describe two constructions and prove their security.

\revise{Our first construction (\Cref{sec: unclonable dec from stronger monogamy}) relies on the strong monogamy-of-entanglement property (Conjecture \ref{conj:strong_monogamy_it}), the existence of post-quantum one-way function, indistinguishability obfuscation and compute-and-compare obfuscation for (sub-exponentially) unpredictable distributions (whose existence has been discussed in \Cref{sec:cc} and \Cref{sec:CC_quantum_aux}). Our second construction (\Cref{sec: unclonable dec witness enc}) has a similar structure. It does not rely on the strong monogamy-of-entanglement property for coset states, but on the (weaker) \compuncertain{} property (Theorem \ref{thm: direct product comp}). However, the construction additionally relies on a much stronger cryptographic primitive -- post-quantum extractable witness encryption (as well post-quantum one-way functions and indistinguishability obfuscation). 
}

\subsection{Definitions}

\begin{definition}[Single-Decryptor Encryption Scheme]
A single-decryptor encryption scheme consists of the following efficient algorithms:
\begin{itemize}
     \item $\setup(1^\lambda) \to (\sk, \pk): $ a (classical) probabilistic algorithm that takes as input a security parameter $\lambda$ and outputs a classical secret key $\sk$ and public key $\pk$.
     
     \item $\keygen(\sk) \to \qsk: $ a quantum algorithm that takes as input a secret key $\sk$ and outputs a quantum secret key $\qsk$.
     
     \item $\enc(\pk, m) \to \ct:$ a (classical) probabilistic algorithm that takes as input a public key $\pk$, a message $m$ and outputs a classical ciphertext $\ct$.
     
     \item $\dec(\qsk, \ct) \to m/\bot: $ a quantum algorithm that takes as input a quantum secret key $\qsk$ and a ciphertext $\ct$, and outputs a message $m$ or a decryption failure symbol $\bot$.
\end{itemize}
\end{definition}

A secure single-decryptor encryption scheme should satisfy the following:
\begin{description}

\item \textbf{Correctness}:  There exists a negligible function $\negl(\cdot)$, for all $\lambda \in \N$, for all $m \in \cM$, 
     \begin{align*}
         & \Pr\left[ \dec(\qsk, \ct) = m \,\middle| \begin{array}{cc} 
              & (\sk, \pk) \gets \setup(1^\lambda), \qsk \gets \keygen(\sk) \\
              & \ct \gets \enc(\pk, m) \\
         \end{array}\right] \geq 1-\negl(\lambda)  
     \end{align*}

Note that correctness \revise{implies that a honestly generated quantum decryption key} can be used to decrypt correctly polynomially many times, from the gentle measurement lemma~\cite{Aaronson05}.

\item \textbf{CPA Security}: The scheme should satisfy (post-quantum) CPA security, i.e. indistinguishability under chosen-plaintext attacks: for every (stateful) QPT adversary $\cA$, there exists a negligible function $\negl(\cdot)$ such that for all $\lambda \in \N$, the following holds:
\begin{gather*}
    \Pr\left[\cA(\ct) = b :
      \begin{array}{cl}
        (\sk, \pk) \gets \setup( 1^\lambda) \\
        ((m_0,m_1) \in \cM^2) \gets \cA(1^\lambda, \pk) \\
        b \gets \{0,1\}; \ct \gets \enc(\pk, m_b)
      \end{array}
    \right] \leq \dfrac{1}{2} + \negl(\lambda),
  \end{gather*}
\end{description}

\paragraph{Anti-Piracy Security}
Next, we define anti-piracy security via the anti-piracy game below. Recall that, intuitively, anti-piracy security says that it is infeasible for a pirate who receives a quantum secret key to produce two quantum keys, which both allow successful decryption. This can be formalized into ways:
\begin{itemize}
    \item (\emph{CPA-style anti-piracy}) We can ask the pirate to provide a pair of messages $(m_0, m_1)$ along with two quantum secret keys, and we test whether the two keys allow to (simultanoeusly) distinguish encryptions of $m_0$ and $m_1$.
    \item (\emph{random challenge anti-piracy}) We do \emph{not} ask the pirate to provide a pair of plaintext messages, but only a pair of quantum secret keys, and we test whether the two quantum secret keys allow for simultaneous decryption of encryptions of uniformly random messages.
    
\end{itemize}
The reader might expect that, similarly to standard definitions of encryption security, the former implies the latter, i.e. that CPA-security (it is infeasible to distinguish encryptions of chosen plaintexts with better than negligible advantage) implies that it is infeasible to decrypt uniformly random challenges with non-negligible probability. However, for the case of anti-piracy security, this implication does not hold, as we explain in more detail in \Cref{sec:unclonable_dec_unified}. This subtlety essentially arises due to the fact that there are two parties involved, having to simultaneously make the correct guess. Therefore, we will state both definitions here, and we will later argue that our construction satisfies both. 

\revise{In Section \ref{sec: unclonable dec strong ag}, we will introduce an even stronger definition of CPA-style anti-piracy (and a stronger definition for random challenge anti-piracy in \Cref{sec:strong_anti_piracy_random}). We will eventually prove that our constructions satisfy both of the strong definitions. We chose to start our presentation of unclonable decryption with the definitions in this section since they are much more intuitive than the stronger version of Section \ref{sec: unclonable dec strong ag}.} %

In order to describe the security games, it is convenient to first introduce the concept of a \emph{quantum decryptor}. The following definition is implicitly with respect to some single-decryptor encryption scheme $(\sf{Setup, QKeyGen, Enc, Dec} )$.

\begin{definition}[Quantum decryptor]
\label{def: quantum decryptor}
A \emph{quantum decryptor} for ciphertexts of length $n$, is a pair $(\rho, U)$ where $\rho$ is a state, and $U$ is a general quantum circuit acting on $n + m$ qubits, where $m$ is the number of qubits of $\rho$. 

For a ciphertext $c$ of length $n$, we say that we run the quantum decryptor $(\rho, U)$ on ciphertext $c$ to mean that we execute the circuit $U$ on inputs $\ket{c}$ and $\rho$.
\end{definition}

We are now ready to describe the CPA-style anti-piracy game.

\begin{definition}[Anti-Piracy Game, CPA-style]
\label{def: regular antipiracy cpa}
Let $\lambda \in \mathbb{N}^+$.
The CPA-style anti-piracy game is the following game between a challenger and an adversary $\mathcal{A}$.
\begin{enumerate}
    \item \textbf{Setup Phase}: The challenger samples keys $(\sk, \pk) \gets \setup(1^\lambda)$.

    \item \textbf{Quantum Key Generation Phase}:
    The challenger sends $\cA$ the classical public key $\pk$ and one copy of quantum decryption key $\qsk \leftarrow \keygen(\sk)$. %

    \item \textbf{Output Phase}: $\As$ outputs a pair of distinct messages $(m_0, m_1)$. It also outputs a (possibly mixed and entangled) state $\sigma$ over two registers $R_1, R_2$ and two general quantum circuits $U_1$ and $U_2$. We interpret $\As$'s output as two (possibly entangled) quantum decryptors $\D_1 = (\sigma[R_1], U_1)$ and $\D_2 = (\sigma[R_2],U_2)$.
    
    \item \textbf{Challenge Phase:} 
    The challenger samples $b_1, b_2$ and $r_1, r_2$ uniformly at random and generates ciphertexts $c_1 = \enc(\pk, m_{b_1}; r_1)$ and $c_2 = \enc(\pk, m_{b_2}; r_2)$. The challenger runs quantum decryptor $\D_1$ on $c_1$ and $\D_2$ on $c_2$, and checks that $\D_1$ outputs $m_{b_1}$ and $\D_2$ outputs $m_{b_2}$. If so, the challenger outputs $1$ (the game is won by the adversary), otherwise outputs $0$. %
\end{enumerate}
We denote by $\sf AntiPiracyCPA(1^{\lambda}, \mathcal{A})$ a random variable for the output of the game.

\end{definition}

Note that an adversary can succeed in this game with probability at least $1/2$. It simply gives $\rho_{\sk}$ to the first quantum decryptor and the second decryptor randomly guesses the plaintext.

We remark that one could have equivalently formulated this definition by having the pirate send registers $R_1$ and $R_2$ to two separated parties Bob and Charlie, who then receive ciphertexts from the challenger sampled as in the Challenge Phase above. The two formulations are equivalent upon identifying the quantum circuits $U_1$ and $U_2$. 

\begin{definition}[Anti-Piracy Security, CPA-style]  \label{def:weak_ag}
Let $\gamma: \mathbb{N}^+ \rightarrow [0,1]$. A single-decryptor encryption scheme satisfies $\gamma$-anti-piracy security, if for any QPT adversary $\cA$,  there exists a negligible function $\negl(\cdot)$ such that the following holds for all $\lambda \in \N$: 
  \begin{align}
    \Pr\left[b = 1, b \gets \sf{AntiPiracyCPA}(1^{\lambda}, \mathcal{A}) \right]\leq \frac{1}{2} + \gamma(\lambda) + \negl(\lambda)
    \end{align}
\end{definition}

Unless specified otherwise, when discussing anti-piracy security of an unclonable encryption scheme in this work, we refer to CPA-style anti-piracy security. %

It is not difficult to show that if $\gamma$-anti-piracy security holds for all inverse poly $\gamma$, then this directly implies CPA security (we refer the reader to the appendix (\Cref{sec:unclonable_dec_cpa_ag_implies_cpa}) for the proof of this implication). %

\vspace{1em}

Next, we define an anti-piracy game with \emph{random challenge plaintexts}. This quantifies how well an efficient adversary can produce two ``quantum decryptors'' both of which enable successful decryption of encryptions of uniformly random plaintexts. This security notion will be directly useful in the security proof for copy-protection of PRFs in \Cref{sec:cp_wPRF}.

\begin{definition}[Anti-Piracy Game, with random challenge plaintexts] 
\label{def: regular antipiracy random challenges}Let $\lambda \in \mathbb{N}^+$.
The anti-piracy game with random challenge plaintexts is the following game between a challenger and an adversary $\mathcal{A}$.
\begin{enumerate}
    \item \textbf{Setup Phase}: The challenger samples keys $(\sk, \pk) \gets \setup(1^\lambda)$.

    \item \textbf{Quantum Key Generation Phase}:
    The challenger sends $\cA$ the classical public key $\pk$ and one copy of quantum decryption key $\qsk \leftarrow \keygen(\sk)$. %
    
    \item \textbf{Output Phase}: $\As$ outputs a (possibly mixed and entangled) state $\sigma$ over two registers $R_1, R_2$ and two general quantum circuits $U_1$ and $U_2$. We interpret $\As$'s output as two (possibly entangled) quantum decryptors $\D_1 = (\sigma[R_1], U_1)$ and $\D_2 = (\sigma[R_2],U_2)$.
    \item \textbf{Challenge Phase:} 
    The challenger samples $m_1, m_2 \gets \cM$  and $r_1, r_2$ uniformly at random, and generates ciphertexts $c_1 = \enc(\pk, m_1; r_1)$ and $c_2 = \enc(\pk, m_2; r_2)$. The challenger runs quantum decryptor $\D_1$ on $c_1$ and $\D_2$ on $c_2$, and checks that $\D_1$ outputs $m_1$ and $\D_2$ outputs $m_2$. If so, the challenger outputs $1$ (the game is won by the adversary), otherwise outputs $0$.  
\end{enumerate}
We denote by $\sf{AntiPiracyGuess}(1^{\lambda}, \mathcal{A})$ a random variable for the output of the game. 
\end{definition}

Note that an adversary can succeed in this game with probability at least $1/|\cM|$. The adversary simply gives $\rho_{\sk}$ to the first quantum decryptor and the second decryptor randomly guesses the plaintext.

\begin{definition}[Anti-Piracy Security, with random challenge plaintexts] 
\label{def:weak_ag_random}
Let $\gamma: \mathbb{N}^+ \rightarrow [0,1]$. 
 A single-decryptor encryption scheme satisfies $\gamma$-anti-piracy security with random challenge plaintexts, if for any QPT adversary $\cA$,  there exists a negligible function $\negl(\cdot)$ such that the following holds for all $\lambda \in \N$: 
  \begin{align}
    \Pr\left[b = 1, b \gets \sf{AntiPiracyGuess}(1^{\lambda}, \mathcal{A}) \right]\leq \frac{1}{|\cM|} + \gamma(\lambda) + \negl(\lambda)
    \end{align}
where $\cM$ is the message space. 
\end{definition}

\begin{remark}
    In the rest of the section, we will mainly focus on \Cref{def:weak_ag} and the stronger version of it from the next section. We will appeal to \Cref{def:weak_ag_random} when we prove security of our copy-protection scheme for PRFs. %
\end{remark}

\subsection{Strong Anti-Piracy Security}
\label{sec: unclonable dec strong ag}

The stronger definition of anti-piracy security that we introduce in this section is more technically involved, and less intuitive, than the definitions in the previous section, but is easier to work with when proving security of our constructions. This section relies on preliminary concepts introduced in Section \ref{sec:unclonable dec ati}. We will refer to the anti-piracy security notions defined in the previous section as \emph{regular} anti-piracy'' to distinguish them from \emph{strong} anti-piracy defined in this section.

In order to describe the anti-piracy game in this section, we first need to formalize a procedure to test good quantum decryptors and the notion of a \emph{good quantum decryptor}. Again, the following definitions are implicitly with respect to some single-decryptor encryption scheme $(\sf{Setup, QKeyGen, Enc, Dec} )$.

%

We first describe a procedure to test good quantum decryptors. The procedure is parametrized by a threshold value $\gamma$. We are guaranteed that, if the procedure passes, then the post-measurement state is a $\gamma$-good decryptor.

\begin{definition}[Testing a quantum decryptor] 
\label{def:gamma_good_decryptor}
   Let $\gamma \in [0,1]$. Let $\pk$ be a public key, and $(m_0, m_1)$ a pair of messages. We refer to the following procedure as a {test for a $\gamma$-good quantum decryptor} with respect to $\pk$ and $(m_0, m_1)$:
   \begin{itemize}
       \item The procedure takes as input a quantum decryptor $(\rho, U)$.
       \item Let $\mathcal{P} = (P, I - P)$ be the following mixture of projective measurements (in the sense of Definition \ref{def:mixture_of_projective}) acting on some quantum state $\rho'$:
       \begin{itemize}
       \item Sample a uniform $b \leftarrow \{0,1\}$. Compute $c \leftarrow \enc(\pk, m_b)$.
       \item Run the quantum decryptor $(\rho', U)$ on input $c$. Check whether the outcome is $m_b$. If so, output $1$, otherwise output $0$.
      \end{itemize}
       \item Let $\ti_{1/2 + \gamma}(\cP)$ be the threshold implementation of $\cP$ with threshold value $\frac{1}{2} + \gamma$, as defined in \Cref{def:thres_implement}. Run $\ti_{1/2 + \gamma}(\cP)$ on $\rho$, and output the outcome. If the output is $1$, we say that the test passed, otherwise the test failed.
   \end{itemize}
\end{definition}

By Lemma \ref{lem:threshold_implementation}, we have the following corollary. 

\begin{corollary}[$\gamma$-good Decryptor]
\label{cor: gamma good dec}
    Let $\gamma \in [0,1]$. Let $(\rho, U)$ be a quantum decryptor. 
    Let $\ti_{1/2 + \gamma}(\cP)$ be the test for a $\gamma$-good decryptor defined above. Then, the post-measurement state conditioned on output $1$ is a mixture of states which are in the span of all eigenvectors of $P$ with eigenvalues at least $1/2+\gamma$. We refer to 
    the latter state as a $\gamma$-good decryptor with respect to $(m_0, m_1)$. %
\end{corollary}

Now we are ready to define the strong $\gamma$-anti-piracy game. 

\begin{definition}[Strong Anti-Piracy Game]
\label{def:gamma_anti_piracy_game}
 Let $\lambda \in \mathbb{N}^+$, and $\gamma \in [0,1]$.
The strong $\gamma$-anti-piracy game is the following game between a challenger and an adversary $\mathcal{A}$.
\begin{enumerate}
   \item \textbf{Setup Phase}: The challenger samples keys $(\sk, \pk) \gets \setup(1^\lambda)$.

    \item \textbf{Quantum Key Generation Phase}:
    The challenger sends $\cA$ the classical public key $\pk$ and one copy of quantum decryption key $\qsk \leftarrow \keygen(\sk)$. %

    \item \textbf{Output Phase}: $\As$ outputs a pair of distinct messages $(m_0, m_1)$. It also outputs a (possibly mixed and entangled) state $\sigma$ over two registers $R_1, R_2$ and two general quantum circuits $U_1$ and $U_2$. We interpret $\As$'s output as two (possibly entangled) quantum decryptors $\D_1 = (\sigma[R_1], U_1)$ and $\D_2 = (\sigma[R_2],U_2)$.
    
    \item \textbf{Challenge Phase}: The challenger runs the test for a $\gamma$-good decryptor
    with respect to $\pk$ and $(m_0, m_1)$ on $\D_1$ and $\D_2$. The challenger outputs $1$ if both tests pass, otherwise outputs $0$.
\end{enumerate}
We denote by $\sf{StrongAntiPiracy}(1^{\lambda}, \gamma, \mathcal{A})$ a random variable for the output of the game. 
\end{definition}

\begin{definition}[Strong Anti-Piracy-Security] \label{def:strong_ag}
 Let $\gamma: \mathbb{N}^+ \rightarrow [0,1]$. A single-decryptor encryption scheme satisfies strong $\gamma$-anti-piracy security, if for any QPT adversary $\cA$,  there exists a negligible function $\negl(\cdot)$ such that the following holds for all $\lambda \in \N$: 
  \begin{align}
    \Pr\left[b = 1, b \gets \sf{StrongAntiPiracy}(1^{\lambda}, \gamma(\lambda), \mathcal{A}) \right]\leq \negl(\lambda)
    \end{align}
\end{definition}

\Cref{def:strong_ag} implies \Cref{def:weak_ag}. %
\begin{theorem}
\label{thm: strong ap implies regular}
Let $\gamma: \mathbb{N}^+ \rightarrow [0,1]$. Suppose a single-decryptor encryption scheme satisfies strong $\gamma$-anti-piracy security (Definition \ref{def:strong_ag}). Then, it also satisfies $\gamma$-anti-piracy security (Definition \ref{def:weak_ag}).
\end{theorem}

\begin{proof}
We refer the reader to appendix (\Cref{sec:unclonable dec strong_implies_weak}) for the proof.
\end{proof} 

In a similar way, one can define a stronger version of random challenge anti-piracy security (\Cref{def:weak_ag_random}). We leave the details to (\Cref{sec:strong_anti_piracy_random}).

\subsection{Construction from Strong  Monogamy Property}
\label{sec: unclonable dec from stronger monogamy}

In this section, we give our first construction of a single-decryptor encryption scheme, whose security relies on the strong monogamy-of-entanglement property from \Cref{sec: monogamy conjectured}.

In the rest of the paper, to simplify notation, whenever it is clear from the context, we will denote a program that checks membership in a set $S$ simply by $S$.
\begin{construction} 
$\,$
\label{cons: unclonable dec}
\begin{itemize}
    \item $\setup(1^\lambda) \to (\sk, \pk):$
    \begin{itemize}
        \item Sample $\kappa$  random $(n/2)$-dimensional subspaces $A_i \subseteq \F^n_2$ for $i = 1, 2, \cdots, \kappa$, where $n = \lambda$ and $\kappa =  \kappa(\lambda)$ is a polynomial in $\lambda$. 
        
        \item For each $i \in [\kappa]$, choose two uniformly random vectors $s_i, s_i' \in \F^n_2$. %
        
        \item Prepare the programs $\iO(A_i + s_i)$ and $\iO(A^\perp_i + s'_i)$ (where we assume that the programs $A_i + s_i$ and $A^\perp_i + s'_i$ are padded to some appropriate length).

        \item Output $\sk = \{A_i, s_i, s_i'\}_{i \in [\kappa]}, \pk = \{\iO({A_i+s_i}), \iO({A_i^\perp+s_i'})\}_{i \in [\kappa]}$.
    \end{itemize}
    
    \item $\keygen(\sk) \to \qsk: $ %
    on input $\sk =  \{A_i, s_i, s_i'\}_{i \in [\kappa]}$, output the ``quantum secret key'' $\qsk = \{ \ket{A_{i, s_i,s_i'}} \}_{i \in [\kappa]}$. \revise{ Recall that each $\ket {A_{i, s_i, s'_i}}$ is 
    \begin{align*}
        \ket {A_{i, s_i, s'_i}} = \frac{1}{\sqrt{|A_i|}} \sum_{a \in A_i} (-1)^{\langle a, s'_i\rangle} \ket {a + s_i}.
    \end{align*}
    }
    
    %
    %
    %

    \item $\enc(\pk, m) \to \ct:$ %
    on input a public key $\pk = \{\iO({A_i+s_i}), \iO({A_i^\perp+s_i'})\}_{i \in [\kappa]}$ and message $m$:
    \begin{itemize}
        \item Sample a uniformly random string $r \gets \{0,1\}^{\kappa}$.
        \item Let $r_i$ be the $i$-th bit of $r$. Define $R_i^{0} = \iO({A_i+s_i})$ and $R_i^1 = \iO({A^\perp_i+s'_i})$. Let $\P_{m,r}$ be the following program:
\begin{figure}[hpt]
\centering
\begin{mdframed}[
  linecolor=black,
  leftmargin =8em,
  rightmargin=8em,
  usetwoside=false,
]
\revise{ 
On input $u = u_1 || u_2 || \cdots || u_\kappa$ (where each $u_i \in \mathbb{F}_2^n$):
\begin{enumerate}
\item If for all $i \in [\kappa]$, $R_i^{r_i}(u_i) = 1$:

    \quad Output $m$
\item Else:

    \quad Output $\bot$
\end{enumerate}
}
\end{mdframed}
\caption{Program $P_{m,r}$}
\label{fig:program_P_mr}
\end{figure}

        \item Let $\hatP_{m,r} = \iO(\P_{m,r})$. Output ciphertext $\ct = (\hatP_{m,r}, r)$.
    \end{itemize}
    
     \item $\dec(\qsk, \ct) \to m/\bot:$ %
     on input $\qsk = \{\ket{A_{i, s_i, s_i'}}\}_{i \in [\kappa]}$ and $\ct = (\hatP_{m,r}, r)$: 
     \begin{itemize}
         \item For each $i \in [\kappa]$, if $r_i = 1$, apply $H^{\otimes n}$ to the $i$-th state $\ket{A_{i, s_i, s_i'}}$; if $r_i = 0$, leave the $i$-th state $\ket{A_{i, s_i, s_i'}}$ unchanged. Denote the resulting state by $\qsk^*$.
         
         \item Evaluate the program $\hatP_{m,r}$ on input $\qsk^*$ in superposition; measure the evaluation register and denote the outcome by $m'$. Output $m'$.
         
         \revise{
         \item Rewind by applying the operations in the first step again.

         }
     \end{itemize}
\end{itemize}
\end{construction}

\paragraph{Correctness.} %

Honest evaluation applies $H^{\otimes n}$ to $\ket{A_{i, s_i, s_i'}}$ whenever $r_i = 1$. Clearly, the coherent evaluation of $\iO(A_i + s_i)$ on $\ket{A_{i, s_i, s_i'}}$ always outputs $1$, and likewise the coherent evaluation of $\iO(A^\perp_i + s'_i)$ on $H^{\otimes n} \ket{A_{i, s_i, s_i'}}$ also always outputs $1$. Therefore, by definition of $\hatP_{m,r}$, the evaluation $\hatP_{m,r}(\qsk^*)$ outputs $m$ with probability $1$.

%

%
%
%
%

\begin{theorem}[Strong Anti-Piracy]
\label{thm:strong_antipiracy_unclonable_dec1}

\unpredictableassumptions, the single-decryptor encryption scheme of Construction \ref{cons: unclonable dec} has strong $\gamma$-anti-piracy security for any inverse polynomial $\gamma$.

\subexpassumptions, the single-decryptor encryption scheme of Construction \ref{cons: unclonable dec} has strong $\gamma$-anti-piracy security for any inverse polynomial $\gamma$.

\end{theorem}

In the above theorem, ELFs and the quantum hardness of LWE are for building the corresponding compute-and-compare obfuscation (see \Cref{thm:CC__from_ELF_iO} and \Cref{thm:CC_subexp_from_LWE_iO}). 
We will prove this theorem in Section \ref{sec: strong antipiracy proof}.

We remark that this does \emph{not} immediately imply that there exists a negligible $\gamma$ such that strong $\gamma$-anti-piracy holds. The slightly subtle reason is that the parameter $\gamma$ in strong $\gamma$-anti-piracy is actually a parameter of the security game (rather than a measure of the success probability of an adversary in the game).

From Theorem \ref{thm: strong ap implies regular}, we know that strong $\gamma$-anti-piracy security implies regular $\gamma$-anti-piracy security. Thus, for any inverse-polynomial $\gamma$, the scheme of Construction \ref{cons: unclonable dec} has regular $\gamma$-anti-piracy security. For regular anti-piracy security, it is straightforward to see that a scheme that satisfies the notion for any inverse-polynomial $\gamma$, also satisfies it for $\gamma = 0$. Thus, we have the following.

\begin{corollary}[Regular Anti-Piracy]
\label{thm:antipiracy_unclonable_dec1}
\unpredictableassumptions, the single-decryptor encryption scheme of Construction \ref{cons: unclonable dec} has regular $\gamma$-anti-piracy security for $\gamma = 0$.

\subexpassumptions, the single-decryptor encryption scheme of Construction \ref{cons: unclonable dec} has regular $\gamma$-anti-piracy security for $\gamma = 0$.

\end{corollary}

As mentioned earlier, it is not clear whether anti-piracy security, CPA-style (Definition \ref{def: regular antipiracy cpa}) implies anti-piracy with random challenge inputs (Definition \ref{def: regular antipiracy random challenges}). Thus, we will also separately prove the latter, since in Section \ref{sec:cp_wPRF} we will reduce security of our PRF copy-protection scheme to it. 

\begin{theorem}[Regular Anti-Piracy, For Random Challenge Plaintexts]
\label{thm:antipiracy_random_unclonable_dec1}
\unpredictableassumptions, the single-decryptor encryption scheme has $\gamma$-anti-piracy security against random challenge plaintexts for $\gamma = 0$.

\subexpassumptions, the single-decryptor encryption scheme has $\gamma$-anti-piracy security against random challenge plaintexts for $\gamma = 0$.

\end{theorem}

\begin{proof}
We refer the reader to \Cref{sec:strong_anti_piracy_random}. The proof follows a similar outline as the proof of CPA-style anti-piracy.
\end{proof}

\subsection{Proof of Strong Anti-Piracy Security of Construction \ref{cons: unclonable dec}}
\label{sec: strong antipiracy proof}
In this section, we prove \Cref{thm:strong_antipiracy_unclonable_dec1}.
We only focus on the first half of the theorem, as the second half follows the same outline of the first one. The only differences between them are: 
\begin{itemize}
    \item They either base on  strong monogamy-of-entanglement or \emph{sub-exponentially} strong monogamy-of-entanglement. 
    \item Thus, they rely on either compute-and-compare obfuscation for any unpredictable distribution (post-quantum ELFs) or sub-exponentially unpredictable distributions (the quantum hardness of LWE).
\end{itemize}

The proof proceeds via two hybrids.
We will mark changes between consecutive hybrids in {\color{red} red}. 
We denote the advantage of adversary $\cA$ in Hybrid $i$ by $\advantage_{\cA, i}$. Let $\gamma$ be any inverse-polynomial.

\paragraph{Hybrid 0.}
\revise{This is the strong $\gamma$-anti-piracy game from Definition \ref{def:gamma_anti_piracy_game}}:

\begin{enumerate}
    \item The challenger samples $(\sk, \pk) \gets \setup(1^\lambda)$ and sends $\pk$ to adversary $\cA$. 
    Let $\sk = \{A_{i},s_i, s_i'\}_{i \in [\kappa]}$ and $\pk = \{\iO({A_i+s_i}), \iO({A_i^\perp+s_i'})\}_{i \in [\kappa]}$.
    
    \item The challenger prepares the quantum key $\qsk =  \{\ket{A_{i,s_i, s_i'}}\}_{i \in [\kappa]}  \gets \keygen(\sk)$, and sends $\qsk$ to $\cA$.
    
    \item $\cA$ produces a quantum state $\sigma$ over two registers $R_1, R_2$, unitaries $U_{1}, U_{2}$. and a pair of messages $(m_0, m_1)$. Sends this information to the challenger. 
    
     \item For $i \in \{1,2\}$, let $\mathcal{P}_{i,D}$ be the following mixture of projective measurements acting on some quantum state $\rho'$:
       \begin{itemize}
       \item Sample a uniform $b \leftarrow \{0,1\}$. Compute $c \leftarrow \enc(\pk, m_b)$.
       \item Run the quantum decryptor $(\rho', U_i)$ on input $c$. Check whether the outcome is $m_b$. If so, output $1$, otherwise output $0$.
       \end{itemize}
       Formally, let $D$ be the distribution over pairs $(b,c)$ of (bit, ciphertext) defined in the first bullet point, and let $\cP_i = \{M^i_{(b,c)}\}_{b,c}$ be a collection of projective measurements where $M^i_{(b,c)}$ is the projective measurement described in the second bullet point. Then, $\mathcal{P}_{i,D}$ is the mixture of projective measurements associated to $D$ and $\cP_i$ (as in Definition \ref{def:mixture_of_projective}).

    \item The challenger runs $\ti_{\frac{1}{2}+\gamma}(\cP_{1,D})$ and $\ti_{\frac{1}{2}+\gamma}(\cP_{2,D})$ on quantum decryptors $(\sigma[R_1],U_{1})$  and $(\sigma[R_2], U_{2})$ respectively. $\cA$ wins if both measurements output $1$.
\end{enumerate}

\paragraph{Hybrid 1.}
\revise{Hybrid 1 is the same as Hybrid 0, except in step 5 the challenger runs the \emph{approximate} threshold implementations $\ati^{\epsilon, \delta}_{\cP_1, D, \frac12 + \gamma}$ and $\ati^{\epsilon, \delta}_{\cP_2, D, \frac12 + \gamma}$, where $\epsilon = \frac{\gamma}{4}$, and $\delta$ is some negligible function of $\lambda$.}

\begin{enumerate}
    \item  The challenger samples $(\sk, \pk) \gets \setup(1^\lambda)$ and sends $\pk$ to adversary $\cA$. 
    Let $\sk = \{A_{i},s_i, s_i'\}_{i \in [\kappa]}$ and $\pk = \{\iO({A_i+s_i}), \iO({A_i^\perp+s_i'})\}_{i \in [\kappa]}$.
    
    \item The challenger prepares the quantum key $\qsk =  \{\ket{A_{i,s_i, s_i'}}\}_{i \in [\kappa]}  \gets \keygen(\sk)$, and sends $\qsk$ to $\cA$.
    
    \item $\cA$ produces a quantum state $\sigma$ over two registers $R_1, R_2$, unitaries $U_{1}, U_{2}$. and a pair of messages $(m_0, m_1)$. Sends this information to the challenger.  
    
    \revise{
    \item Let $\cP_{1,D}$ and $\cP_{2,D}$ be as in Hybrid 0. 
    }
    
    \item The challenger {\color{red} runs $\ati^{\epsilon, \delta}_{\cP_1, D, \frac{1}{2} + \gamma}$ and $\ati^{\epsilon, \delta}_{\cP_2, D, \frac{1}{2} + \gamma}$} on $(\sigma[R_1],U_{1})$ and $(\sigma[R_2], U_{2})$. $\cA$ wins if both measurements output $1$.
\end{enumerate}

By \Cref{lem:ati_2d}, we have $\advantage_{\As, 1} \geq \advantage_{\As, 0} - 2\delta$. Moreover, by Lemma \ref{cor:ati_thresimp}, for each $i \in \{1,2\}$, $\ati^{\epsilon, \delta}_{\cP_i, D, \frac12 + \gamma}$ runs in time $\poly(\log (1/\delta), 1/\epsilon)$. The latter is polynomial for our choice of $\epsilon$ and $\delta$.

\vspace{1em}

We complete the proof of \Cref{thm:strong_antipiracy_unclonable_dec1} by showing that the advantage $\advantage_{\As,1}$ in Hybrid 1 is negligible. 

\begin{lemma}\label{lemm:adv1isnegl}
    $\advantage_{\As,1}$ is negligible. 
\end{lemma}
\begin{proof}
    Suppose for a contradiction that $\advantage_{\As,1}$ is non-negligible. Then, applying $\ati^{\epsilon, \delta}_{\cP_1, D, \frac{1}{2} + \gamma}$ and $\ati^{\epsilon, \delta}_{\cP_2, D, \frac{1}{2} + \gamma}$ on $\sigma[R_1]$ and $\sigma[R_2]$ results in two outcomes $1$ with non-negligible probability. Let $\sigma'$ be the bipartite state conditioned on both outcomes being $1$. %
    
    From the second bullet point of \Cref{lem:ati_2d}, we have the following:
    \begin{equation}
        \Tr\big[ \big(\ti_{\frac{1}{2} + \gamma-2\epsilon}(\cP_{1, D})\otimes \ti_{\frac{1}{2} + \gamma-2\epsilon}(\cP_{2, D})\big) \sigma' \big]\geq  1- 4 \delta, \label{eq: ti bound}
    \end{equation}
    where recall that, for ease of notation, when we write $\ti$ inside of a trace we are referring to the projection on the $1$ outcome.
    The observation says that the collapsed state $\sigma'[R_1]$ is negligibly close to being a $(\gamma-2\epsilon)$-good decryptor with respect to ciphertexts generated according to distribution $D$. 
    Similarly for $\sigma'[R_2]$.
    
    We then define a different but computationally close distribution $D'$ over pairs $(b,c)$ of (bit, ciphertext). Let $(m_0, m_1)$ be the pair of messages chosen by $\mathcal{A}$. %

        \begin{enumerate}
        \item Sample $b \gets \{0,1\}$ and $r \gets \{0,1\}^\kappa$.  %

        \item Let $\can_{i, 0}(\cdot) = \can_{A_i}(\cdot)$ and $\can_{i,1}(\cdot) = \can_{A^\perp_i}(\cdot)$ where $ \can_{A_i}(\cdot),\can_{A^\perp_i}(\cdot)$ are the functions defined in \Cref{def:canonical_vec_func}.
        
        \item Define function $f$ as follows:
        \begin{align*}
            f(u_1, \cdots, u_\kappa) = \can_{1, r_1}(u_1) || \cdots || \can_{\kappa, r_\kappa}(u_\kappa).
        \end{align*}
         Let $s_{i, 0} = s_i$ and $s_{i, 1} = s_i'$. Let the ``lock value'' $y$ be the following:
        \begin{align*}
            y = \can_{1, r_1}(s_{1, r_1}) || \cdots || \can_{\kappa, r_\kappa}(s_{\kappa, r_\kappa}).
        \end{align*}
       Let $C_{m_b, r}$ be the compute-and-compare program $\CC[f, y, m_b]$. 

        \item Run the obfuscation algorithm $\CC.\obf$ on $C_{m_b,r}$ and obtain the obfuscated program $\obfCC_{m_b,r} = \CC.\obf(C_{m_b,r})$. Let $\hatCC_{m_b,r} = \iO(\obfCC_{m_b,r})$. 
        
       \item Let $c = (\hatCC_{m_b,r}, r)$. Output $(b, c)$.
       \end{enumerate}

Since the programs $C_{m_b,r}$ and $P_{m_b, r}$ (from Fig. \ref{fig:program_P_mr}) are functionally equivalent, the two distributions $D$ and $D'$ are computationally indistinguishable assuming post-quantum security of $\iO$. A direct application of \Cref{thm:ti_different_distribution}, together with \eqref{eq: ti bound}, gives the following corollary. Let $\cP_{1, D'}$ be the same mixture of projective measurements as $\cP_{1, D}$, except that pairs of (bit, ciphertext) are sampled according to distribution $D'$ instead of $D$.
\begin{corollary}
\label{cor: sigma'}
Let $D, D'$ be the distributions defined above. 
Let $\sigma'$ be the post-measurement state conditioned on $\ati^{\epsilon, \delta}_{\cP_1, D, \frac{1}{2} + \gamma}$ and $\ati^{\epsilon, \delta}_{\cP_2, D, \frac{1}{2} + \gamma}$ %
both outputting $1$ on $\sigma[R_1]$ and $\sigma[R_2]$. For any inverse polynomial $\epsilon'$, there exists a negligible function $\delta'$ such that: 
    \begin{align*}
        \Tr\big[ \big(\ti_{\frac{1}{2} + \gamma-2\epsilon-\epsilon'}(\cP_{1, D'})\otimes \ti_{\frac{1}{2} + \gamma-2\epsilon-\epsilon'}(\cP_{2, D'})\big) \sigma' \big]\geq  1- 4 \delta - \delta'.
    \end{align*}
\end{corollary}
Intuitively, the above corollary says that if $\sigma'[R_1]$ and $\sigma'[R_2]$ are both negligibly close to being $(\gamma-2 \epsilon)$-good decryptors with respect to ciphertexts generated according to $D$, then, for any inverse polynomial $\epsilon'$, they are also negligibly close to being $(\gamma- 2 \epsilon - \epsilon')$-good decryptors with respect to ciphertexts generated according to $D'$. By setting $\epsilon' = \epsilon = \frac{\gamma}{4}$, we have that there exists a negligible function $\negl$ such that:
\begin{align} \label{eq:sigma_prime_good}
    \Tr\big[ \big(\ti_{\frac{1}{2} + \frac{\gamma}{4}}(\cP_{1, D'})\otimes \ti_{\frac{1}{2} + \frac{\gamma}{4}}(\cP_{2, D'})\big) \sigma' \big]\geq  1- \negl(\lambda). 
\end{align}

\vspace{1em}

The rest of the proof amounts to showing that all of the above implies that there exists an efficient algorithm breaking the computational strong monogamy-of-entanglement property. Before giving the full details, we provide a sketch of how the proof proceeds.
\begin{itemize}
\item[(i)] First notice that if the ``lock value'' $y$, as defined in the description of $D'$, is computationally unpredictable given the quantum decryptor $\sigma[R_1]$ (and the additional classical auxiliary information), then the security of compute-and-compare obfuscation implies that we can replace $D'$ with a distribution $D''$ that contains no information at all about the plaintext (with respect to which no quantum decryptor can have any advantage beyond random guessing).
\item[(ii)] From \eqref{eq:sigma_prime_good}, we know that conditioned on the (approximate) threshold implementation measurement accepting on both sides (which happens with non-negligible probability), each side is (close to) a $\frac{\gamma}{4}$-good decryptor with respect to $D'$. Notice that this implies the existence of an efficient algorithm that takes $\sigma[R_1]$ as auxiliary information, and distinguishes between $D'$ and $D''$. By the security of compute-and-compare obfuscation, this implies that the lock value in the left ciphertext must be predictable given $\sigma[R_1]$ (and the classical auxiliary information). Similarly the lock value in the right ciphertext must be predictable given $\sigma[R_2]$ (and the classical auxiliary information).
\item[(iii)] Since lock values consist of concatenations of canonical representatives of either the coset $A_i + s_i$ or $A_i^{\perp} + s_i'$, one would like to conclude that it is possible to extract (with non-negligible probability) one representative on each side (i.e. one using the information in register $R_1$, and one using the information in register $R_2$). However, one has to be cautious, since this deduction does not work in general! In fact, successfully extracting on $R_1$'s side might destroy the (entangled) quantum information on $R_2$'s side, preventing a successful simultaneous extraction. 
\item[(iv)] The key is that if each side is (close to) a $\frac{\gamma}{4}$-good decryptor, then no matter what measurement is performed on the left side, and no matter what outcome is obtained, the state on the right side is still in the support of $\frac{\gamma}{4}$-good decryptors. This means that extraction will still succeed with non-negligible probability on the right side. This implies a strategy that succeeds at extracting canonical representatives simultaneously on both sides. Finally, since for each $i$ the choice of whether to encrypt using $A_i +s_i$ or $A_i^{\perp} + s_i'$ is independent and uniformly random, with overwhelming probability there will be some $i$ such that the extracting algorithm will recover $s_i$ and $s_i'$ simultaneously, breaking the strong monogamy-of-entanglement property.
\end{itemize}

\paragraph{\textbf{Extracting from register $R_1$.}}
Let $\cP_{1, D'} = (P_{1, D'}, I- P_{1, D'})$. Recall that from \Cref{eq:sigma_prime_good} we have
\begin{align*}
    \Tr[\ti_{1/2 + \frac{\gamma}{4}}(\cP_{1, D'}) \, \sigma'[R_1]] \geq 1 - \negl(\lambda).
\end{align*}
This implies that $\sigma'[R_1]$ has $1-\negl(\lambda)$ weight over eigenvectors of $P_{1, D'}$ whose eigenvalues are at least $1/2 + \frac{\gamma}{4}$. Therefore,
\begin{align}
\label{eq: D' bound}
    \Tr[P_{1, D'} \, \sigma'[R_1]] \geq  \frac{1}{2} + \frac{\gamma}{4} - \negl. 
\end{align}

Hence, if we view $(\sigma'[R_1], U_{1})$ as a quantum decryptor, its advantage on challenges sampled from $D'$ is noticeably greater than random guessing. 

Let $\Sim$ be an efficient simulator for the compute-and-compare obfuscation scheme that we employ (using the notation of Definition \ref{def: cc obf}). We define $D''$ to be the following distribution over pairs $(b,c)$ of (bit, ciphertext).
\begin{itemize}
    \item  Let $\widetilde{\Sim} = \Sim(1^\lambda, {\sf param})$ where ${\sf param}$ consists of the parameters of the compute-and-compare program being obfuscated in the description of $D'$ (input size, output size, circuit size - these are the parameters of $C_{m_b, r}$).

    \item Let $c = (\iO(\widetilde{\Sim}), r)$. Output $(b, c)$. %
\end{itemize}

Because the simulated ciphertext generated in $D''$ is independent of $m_b$, the quantum decryptor $(\sigma'[R_1], U_{1})$ cannot enable guessing $b$ with better than $1/2$ probability. More concretely, let $\mathcal{P}_{1, D''} = (P_{1,D''}, P_{2,D''})$ be the mixture of projective measurements which is the same as $\mathcal{P}_{1, D'}$ except pairs of (bit, ciphertext) are sampled according to $D''$ instead of $D'$. Then,
\begin{align}
\label{eq: D'' bound}
    \Tr[P_{1, D''} \, \sigma'[R_1]] = \frac{1}{2}. 
\end{align}

Since the quantum decryptor $(\sigma'[R_1], U_{1})$ behaves noticeably differently on distributions $D'$ and $D''$, it can be used as a distinguisher for the two related distributions $\widehat{D}'$ and $\widehat{D}''$, defined as follows (these implicitly depend on some adversary $\mathcal{A}$):

\begin{enumerate}
        \item $\widehat{D}'$: a distribution over pairs of programs and auxiliary information
        $$ (C, \AUX) \,,$$
 where $\AUX = (\pk, m_b, r, \sigma'[R_1], U_1)$, with the latter being sampled like the homonymous parameters in Hybrid $0$, and $C = \obfCC_{m_b,r}$ is an obfuscated compute-and-compare as in $D'$.

        \item $\widehat{D}''$: a distribution over pairs of programs and auxiliary information
         $$(C, \AUX)\,,$$
         where $\AUX = (\pk, m_b, r, \sigma'[R_1], U_1)$, with the latter being sampled like the homonymous parameters in Hybrid $0$, and $C = \widetilde{\Sim}$ as in $D''$.
\end{enumerate}

To distinguish $\widehat{D}'$ and $\widehat{D}''$, a distinguisher runs the quantum decryptor $(\sigma'[R_1], U_1)$ on input $(\iO(C), r)$ and checks whether the output is equal to $m_b$. By the definition of $\widehat{D}'$ and $\widehat{D}''$, it is straightforward to see that in the first case this procedure is equivalent to performing the measurement $\cP_{1, D'}$ (recall that the latter depends on $U_1$) on $\sigma'[R_1]$, and in the second case this procedure is equivalent to performing the measurement $\cP_{1, D''}$ on $\sigma'[R_1]$. By \eqref{eq: D' bound} and \eqref{eq: D'' bound}, the advantage of this distinguisher is at least $\frac{\gamma}{4} - \negl$, which is noticeable.  %

Thus, by the security of compute-and-compare obfuscation, the distribution $\widehat{D}'$ over pairs of (compute-and-compare program, auxiliary information) is not computationally unpredictable (as in Definition \ref{def:cc_unpredictable_dist}). In particular, the contrapositive of Definition \ref{def:cc_unpredictable_dist} is that there exists an efficient algorithm $\cM_1$ that succeeds at the following with non-negligible probability:
\begin{itemize}
    \item Let $(\CC[f, y, m], \AUX) \gets \widehat{D}'$. 
    \item $\cM_1$ receives $f$ and $\AUX$, and outputs $y'$. $\cM_1$ is successful if $y' = y$.
\end{itemize}

    In our case, for a fixed $\{A_i, s_i, s'_i\}_{i \in [\kappa]}$ and $r$, the function $f$ is defined as: 
\begin{align*}
   f(u_1, \cdots, u_\kappa) &= \can_{1, r_1}(u_1) || \cdots || \can_{\kappa, r_\kappa}(u_\kappa)   \end{align*}
Notice that this function is efficiently computable given descriptions of the subspaces $A_i$. Thus, in our case, the contrapositive of Definition \ref{def:cc_unpredictable_dist} says that there exists an adversary which receives the description of the function $f$ and $\AUX$, which in particular includes the description of the subspaces $A_i$, and is able to guess the appropriate coset representatives, depending on the bits of $r$. The existence of this adversary will be crucial in our reduction to an adversary for the computational strong monogamy-of-entanglement game. Notice that in the monogamy-of-entanglement game, each of the two parties $\mathcal{A}_1$ and $\mathcal{A}_2$ receives the descriptions of the subspaces $A_i$, but not of the cosets, which they have to guess.

\paragraph{\textbf{Extracting on register $R_2$.}} If two registers are entangled, then performing measurements on one register will generally result in destruction of the quantum information on the other side. We show that this does not happen in our case, and that one can extract coset representatives from register $R_1$, in a way that the leftover state on $R_2$ also allows for extraction of coset representatives. Recall the definition of $\gamma$-good decryptor from Corollary \ref{cor: gamma good dec}. Informally,  a quantum decryptor $(\rho, U)$ is a $\gamma$-good decryptor with respect to a distribution $D$ over pairs of (bit, ciphertext) if $\rho$ is a mixture of states which are all in the span of eigenvectors of the $\gamma$-good decryptor test with respect to $D$, with  eigenvalues greater than $\frac12+\gamma$.

\begin{claim}
\label{claim:unchanged_r2}
Let $\rho$ be a bipartite state on registers $R_1$ and $R_2$. Let $U_1$ and $U_2$ be general quantum circuits acting respectively on $R_1$ and $R_2$ (plus a register containing ciphertexts). Suppose that $(\rho[R_1], U_1)$ and $(\rho[R_2], U_2)$ are both $\gamma$-good decryptors with respect to a distribution $D$ over pairs of (bit, ciphertext). Let $M$ be any POVM on $R_1$. Then, the post-measurement state on $R_2$ (together with $U_2$) conditioned on any outcome is still a $\gamma$-good decryptor with respect to distribution $D$.
\end{claim}

\begin{proof}
Assume $\rho$ is a pure state. The general statement follows from the fact that a mixed state is a convex mixture of pure states.

We can write $\rho$ in an eigenbasis of products of eigenvectors of $P_{1,D}$ and $P_{2,D}$. The hypothesis that both $(\rho[R_1], U_1)$ and $(\rho[R_2], U_2)$ are both $\gamma$-good decryptors implies that $\rho$ can be written as follows: $\rho = \sum_{i,j: p_i,q_j \geq 1/2 + \gamma} \alpha_{i,j} \ket{x_i} \otimes \ket{y_{j}}$, where $x_i$  is an eigenvector of $P_{1,D}$ with eigenvalue $p_i$ and $y_j$ is an eigenvector of $P_{2, D}$ with eigenvalue $q_j$. In particular, note that only the eigenvectors corresponding to eigenvalues $p_i, q_j \geq 1/2 + \gamma$ have non-zero weight. %
Finally, notice that applying any POVM $M$ on $R_1$ (or in general any quantum operation on $R_1$) does not change the support of the resulting traced out state on $R_2$: the support still consists of eigenvectors of $P_{2,D}$ with eigenvalues $\geq 1/2 + \gamma$.

\end{proof}

Now, let $\sigma'$ be as defined in Corollary \ref{cor: sigma'}. Then, \Cref{eq:sigma_prime_good} implies that both $\sigma'[R_1]$ and $\sigma'[R_2]$ are negligibly close to being $\frac{\gamma}{4}$-good decryptors with respect to $D'$. %
By \Cref{claim:unchanged_r2} (together with simple triangle inequalities) this implies that,  conditioned on algorithm $\cM_1$ successfully outputting the lock value (which happens with non-negligible probability), 
the remaining state $\sigma''[R_2]$ is still a $\frac{\gamma}{4}$-good decryptor with respect to $D'$. 

By the same argument as for $\cM_1$, there exists an algorithm $\cM_2$ that takes the description of the function $f$ and auxiliary information $\AUX = (\pk, m_b, r, \sigma''[R_2], U_2)$ and outputs the lock value with non-negligible probability. Thus, $\cM_1$ and $\cM_2$ simultaneously output a lock value with non-negligible probability. 

\paragraph{\textbf{Breaking the strong monogamy-of-entanglement property}.} We now give a formal description of an adversary $(\As_0, \As_1, \As_2)$ breaking the strong monogamy-of-entanglement property (\Cref{conj:strong_monogamy}) of coset states, given an adversary $\mathcal{A}$ that breaks the $\gamma$-anti-piracy-game, and algorithms $\cM_1$ and $\cM_2$ as described above. %
\begin{enumerate}
    \item The challenger picks a uniformly random subspace $A \subseteq \mathbb{F}_2^n$ and two uniformly random elements $s, s' \in \mathbb{F}_2^n$. It sends $\ket {A_{s, s'}}$ and $\iO(A+s), \iO(A^\perp + s')$ to $\As_0$.
    \item $\As_0$ simulates the game in Hybrid 1: 
    \begin{itemize}
        \item Samples $i^* \leftarrow [\kappa]$. Generates $\rho_\sk = \{\ket {A_i, s_i, s'_i}\}_{i \in [\kappa]}$ and $\pk = \{\iO(A_i + s_i), \iO(A^\perp + s'_i)\}_{i \in [\kappa]}$ where $A_i, s_i, s_i'$ are uniformly random except $\ket {A_{i^*, s_{i^*}, s_{i^*}'}} = \ket {A_{s, s'}}$,  $\iO(A_{i^*} + s_{i^*}) = \iO(A + s)$ and 
        $\iO(A_{i^*}^\perp + s'_{i^*}) = \iO(A^\perp + s')$. 
        
        \item $\As_0$ gives $\rho_\sk$ and $\pk$ to the adversary $\As$ for the $\gamma$-anti-piracy game. $\As_0$ obtains $\sigma[R_1], \sigma[R_2]$ together with $U_{1}, U_{2}$ and $(m_0, m_1)$. 
        $\As_0$ applies $\ati^{\epsilon, \delta}_{\cP_1, D, \frac{1}{2} + \gamma}$ and $\ati^{\epsilon, \delta}_{\cP_2, D, \frac{1}{2} + \gamma}$ to $\sigma$, where $D$ is the distribution over pairs of (bit, ciphertext) defined in Hybrid 0 (which is efficiently sampleable given $\pk$.) If if any of the two outcomes is $0$, %
        $\As_0$ halts. If both outcomes are $1$, let the collpased state be $\sigma'$. 

    \end{itemize}
    It sends $\sigma'[R_1], \pk$ to $\As_1$ and $\sigma'[R_2], \pk$ to $\As_2$. Both $\As_1, \As_2$ also get the description of $A_i$ (for all $i \ne i^*$). 
    
    \item The challenger gives the description of $A$ (equivalently, $A_{i^*}$) to $\As_1$ and $\As_2$.
    
    \item $\As_1$ samples $m_{b_1, R_1}$ as a uniformly random message in $\{m_0, m_1\}$ and $r_{R_1} \gets \{0,1\}^\kappa$. Let $\AUX_{R_1} = (\pk, m_{b_1, R_1}, r_{R_1}, \sigma'[R_1])$ and $f_{R_1}$ be the function corresponding to $r_{R_1}$.  Because $\As_1$ gets the description of all $A_i$, the description of $f_{R_1}$ is efficiently computable. It runs $\cM_1$ on $f_{R_1}, \AUX_{R_1}$ and gets the outcome $y_{R_1}$.

    \item Similarly for $\As_2$, it prepares $f_{R_2}, \AUX_{R_2}$, runs $\cM_2$ and gets the outcome $y_{R_2}$.  
\end{enumerate}

It follows from the previous analysis that, with non-negligible probability, both $y_{R_1}$ and $y_{R_2}$ are correct lock values. Since $r_{R_1, i^*} \ne r_{R_2, i^*}$ with overwhelming probability, this violates the strong monogamy-of-entanglement property. %

\end{proof}

Thus, the advantage in Hybrid 1 is negligible. 
This implies that the advantage in Hybrid 0 is also negligible, which concludes the proof of \Cref{thm:strong_antipiracy_unclonable_dec1}.
\qed

\subsection{Construction from Extractable Witness Encryption} %
\label{sec: unclonable dec witness enc}

In this section, we give an alternative construction of a single-decryptor encryption scheme. This construction uses a quantum signature token scheme as a black box. The construction is conceptually very similar to that of Section \ref{sec: unclonable dec from stronger monogamy}, but it uses extractable witness encryption instead of compute-and-compare obfuscation to deduce simultaneous extraction. Because the extraction guarantee from extractable witness encryption is stronger than the one from compute-and-compare obfuscation (we elaborate on this difference in Section \ref{sec: security uncl dec construction 2}), we do not need to reduce security of the scheme to the strong monogamy-of-entanglement property, but instead we are able to reduce security of the scheme to security of the signature token scheme (which, recall, is a primitive that we show how to construct using the computational direct product hardness property of coset states in Section \ref{sec: signature tokens}).

In the following construction, let $\we = (\we.\enc, \we.\enc)$ be an extractable witness encryption scheme (as in Definition \ref{def: extractable witness enc}), and let $\TS = (\TS.\kgen, \TS.\tokengen, \TS.\sign, \TS.\verify)$ be an unforgeable signature token scheme (as in Definitions \ref{def: ts} and \ref{def: ts unforgeability}). The construction below works to encrypt single bit messages, but can be extended to messages of polynomial length without loss of generality.

\begin{construction} 
$\,$
\label{cons: unclonable dec 2}
\begin{itemize}
\item {$\setup(1^\param) \to (\pk, \sk)$}: Let $\kappa = \kappa(\lambda)$ be a polynomial.
\begin{itemize}
    \item For each $i \in [\kappa]$, compute $(\TS.\sk_i, \TS.\pk_i) \gets \TS.\kgen(1^\lambda)$. 
    
    \item Output $\pk = \{\TS.\pk_i\}_{i \in [\kappa]}$ and $\sk = \{\TS.\sk_i\}_{i \in [\kappa]}$.
\end{itemize}

\item $\keygen(\sk) \to \qsk:$ On input $\sk = \{\TS.\sk_i\}_{i \in [\kappa]}$:
\begin{itemize}
    \item For $i \in [\kappa]$, compute $\ket{\tk_i} \gets \TS.\tokengen(\TS.\sk_i)$
    
    \item Output $\qsk = \{\ket{\tk_i}\}_{i \in [\kappa]}$
\end{itemize}

\item {$\enc(\pk, m) \to \ct$:}
On input a public key $\pk =  \{\TS.\pk_i\}_{i \in [\kappa]}$ and a message $m \in \{0,1\}$; 
\begin{itemize}
    \item Sample a random string $r \gets \{0,1\}^\kappa$
    
    \item Compute $\ct_{r,m} \gets \we.\enc(1^n, r, m)$, where $r$ is an instance of the language $L$, defined by the following $NP$ relation $R_L$. In what follows, let $w$ be parsed as $w = w_1||\cdots ||w_\kappa $, for $w_i$'s of the appropriate length.
\begin{align}
\label{eq:function_p}
R_L(r, w)=\begin{cases}
1 &\text{if } \TS.\verify(\TS.\pk_i, r_i, w_i) = 1 \text{ for all } i \in [\kappa],\\
0 &\text{ otherwise}.
\end{cases}
\end{align}
That is, each $w_i$ should be a valid signature of $r_i$. 

\item Output the ciphertext $\ct = (\ct_{r,m}, r)$.
\end{itemize}

\item {$\dec(\ct, \qsk) \to m/ \bot$:} On input a ciphertext $\ct = (\ct_{r,m}, r)$ and a quantum secret key $\qsk = \{\ket{\tk_i}\}_{i \in [\kappa]}$. 

\begin{itemize}
    \item For each $i \in [\kappa]$, sign message $r_i$ by running $(r_i, \sig_i) \gets \TS.\sign(\ket{\tk_i}, r_i)$. Let $ w = \sig_1 || \cdots || \sig_\kappa$.
    
    \item Output $m/\bot \gets\we.\dec(\ct_{r,m},\sig_1 || \cdots || \sig_\kappa)$.

\end{itemize}
 Note in the decryption algorithm $\dec$, we run $\TS.\sign$ and $\we.\dec$ coherently, so that (by the gentle measurement lemma) an honest user can rewind and use the quantum key polynomially many times.

\end{itemize}
\end{construction}

\jiahui{comment out the following section for camera ready}
\subsection{Security of Construction \ref{cons: unclonable dec 2}}
\label{sec: security uncl dec construction 2}
The proofs of security are straightforward, and similar to the proofs given in \cite{georgiou-zhandry20}, except that here we have a new definition of $\gamma$-anti-piracy security, and we use a tokenized signature scheme instead of a one-shot signature scheme. The proof also resembles our proof of security for the construction from the strong monogamy-of-entanglement property. 

We sketch the proofs here and omit some details.

\paragraph{Correctness and Efficiency. } It is straightforward to see that all procedures are efficient and
that correctness follows from the correctness 
of the $\we$ and $\TS$ schemes.

\paragraph{CPA Security. } 
CPA security relies on extractable security of the witness encryption scheme and on unforgeability of the tokenized signature scheme.
Suppose that there exists a QPT adversary $\cA$ that succeeds with non-negligible probability in its CPA security game, by the extractable security of witness encryption, there exists an extractor that extracts witness $w = \sig_i || \cdots || \sig_\kappa$, where each $\sig_i$ is the signature of a random bit $r_i$ that can $\TS.\verify$. This clearly violates the unforgeability of $\TS$, since the adversary $\cA$ in CPA security game is not given any tokens.

\paragraph{(Strong) $\gamma$-Anti-Piracy. }
Strong $\gamma$-anti-piracy security for any inverse-polynomial $\gamma$ also follows from extractable security of the witness encryption scheme and unforgeability of tokenized signature scheme scheme.

Suppose that there exists a QPT adversary $(\cA_0, \cA_1, \cA_2)$ that succeeds with non-negligible probability in the $\gamma$-anti-piracy game. Then, with non-negligible probability over the randomness of the challenger, $\cA_0$ outputs a state $\sigma$ such that $\sigma[R_1]$ and $\sigma[R_2]$ simultaneously pass the $\gamma$-good decryptor test with non-negligible probability. Therefore, by applying the corresponding approximation threshold implementation ($\ati$), the resulting state $\sigma'[R_1]$ and $\sigma'[R_2]$ are negligibly close to being $(\gamma-\epsilon)$-good decryptors, for any inverse polynomial $\epsilon$. %

Since $(\gamma-\epsilon)$ is inverse-polynomial, then by the extractable security of the witness encryption scheme, there must exist an extractor $E_1$ on the $R_1$ side that extracts, with non-negligible probability, a witness $w_1 = \sig_{1,1} || \cdots || \sig_{1,\kappa}$ (where each $\sig_{1,i}$ is a signature for bit $r_{R_1,i}$). Similarly, by \Cref{claim:unchanged_r2}, there also exists an extractor $E_2$ on the $R_2$ side that extracts witness $w_2 = \sig_{2,1} || \cdots || \sig_{2,\kappa}$ (where each $\sig_{2,i}$ is a signature for bit $r_{R_2,i}$) from the leftover state after extraction on $R_1$. Since $r_{R_1}, r_{R_2}$ are independently sampled, with probability $(1- 1/2^\kappa)$, there exists a  position $i^*$ where $r_{R_1,i^*} \neq r_{R_2,i^*}$. We can then construct an adversary that breaks the 1-unforgeability of tokenized signatures by getting one token $\ket{\tk_{i^*}}$ and successfully producing signatures on two different messages $r_{R_1,i^*} \neq r_{R_2,i^*}$.
\section{Copy-Protection of Pseudorandom Functions}
\label{sec:cp_wPRF}

In this section, we formally define copy-protection of pseudorandom functions. 
Then, we describe a construction that essentially builds on the single-decryptor encryption scheme described in Section \ref{sec: unclonable dec from stronger monogamy} (together with post-quantum sub-exponentially secure one-way functions and $\iO$). The same construction can be based on the single-decryptor encryption scheme from Section \ref{sec: unclonable dec witness enc}, but we omit the details to avoid redundancy.
In Section \ref{sec: PRF prelim}, we give definitions of certain families of PRFs which we use in our construction. We remark that all of the PRFs that we use can be constructed from post-quantum one-way functions.

\subsection{Definitions}
In what follows, the PRF $F: [K]\times [N] \rightarrow [M]$, implicitly depends on a security parameter $\lambda$. We denote by $\setup(\cdot)$ the procedure which on input $1^{\lambda}$, outputs a PRF key. 

\begin{definition}[Copy-Protection of PRF]
    A copy-protection scheme for a  PRF $F:[K] \times [N] \to [M]$ 
    consists of the following QPT algorithms: 
    
    \begin{description}
        \item[] $\keygen(K)$: takes a key $K$ and outputs a quantum key $\rho_K$; %
        \item[] $\eval(\rho_K, x)$: takes a quantum key $\rho_K$ and an input $x \in [N]$. It outputs a classical string $y \in [M]$.
    \end{description}
\end{definition}

A copy-protection scheme should satisfy the following properties: 
\begin{definition}[Correctness]
There exists a negligible function $\negl(\cdot)$, such that for all $\lambda$, all $K \gets \setup(1^\lambda)$,
all inputs $x$, 
    
$$\Pr[ \eval(\rho_K, x) = F(K,x) : \rho_K \leftarrow \keygen(K)] \geq 1 - \negl(\lambda) \, .
$$ 
\end{definition}
Note that the correctness property implies that the evaluation procedure has an ``almost unique'' output. This means that the PRF can be evaluated (and rewound) polynomially many times, without disturbing the quantum key $\rho_K$, except negligibly.

\begin{definition}[Anti-Piracy Security] \label{def:weak anti piracy}
   Let $\lambda \in \mathbb{N}^+$. Consider the following game between a challenger and an adversary $\As$:
    \begin{enumerate}
        \item The challenger samples $K \gets \setup(1^\lambda)$ and $\rho_K \gets \keygen(K)$.  It gives $\rho_K$ to $\As$; 
        \item $\As$ returns to the challenger a bipartite state $\sigma$ on registers $R_1$ and $R_2$, as well as general quantum circuits $U_{1}$ and $U_{2}$.%
        \item The challenger samples uniformly random $u, w \gets [N]$. Then runs $U_1$ on input $(\sigma[R_1], u)$, and runs $U_2$ on input $(\sigma[R_2], w)$. The outcome of the game is $1$ if and only if the outputs are $F(K, u)$ and $F(K, w)$ respectively. 
    \end{enumerate}
    
    Denote by $\sf{CopyProtectionGame}(1^{\lambda}, \mathcal{A})$ a random variable for the output of the game.
    
    We say the scheme has anti-piracy security if for every polynomial-time quantum algorithm $\As$, there exists a negligible function $\negl(\cdot)$, for all $\lambda \in \mathbb{N}^+$, 
    \begin{align*}
        \Pr\left[b = 1, b \gets \sf{CopyProtectionGame}(1^{\lambda}, \mathcal{A}) \right] = \negl(\lambda) \,.
    \end{align*}
\end{definition}

We give a stronger anti-piracy definition, which is an indistinguishability definition and is specifically for copy-protecting PRFs. We will show that our construction also satisfies this definition. 

\begin{definition}[Indistinguishability Anti-Piracy Security for PRF]
\label{def:indistinguishable_anti_piracy_PRF}
    Let $\lambda \in \mathbb{N}^+$. Consider the following game between a challenger and an adversary $\As$: 
    \begin{enumerate}
        \item The challenger runs $K \gets \setup(1^\lambda)$, and $\rho_K \gets \keygen(K)$. It gives $\rho_K$ to $\As$; 
        \item  $\As$ returns to the challenger a bipartite state $\sigma$ on registers $R_1$ and $R_2$, as well as general quantum circuits $U_{1}$ and $U_{2}$.%
        \item The challenger samples two uniformly random inputs $u, w \gets [N]$ and two uniformly random strings $y_1, y_2 \gets [M]$ (these are of the same length as the PRF output).
        
        \item The challenger flips two coins independently: $b_1, b_2 \gets \{0,1\}$. If $b_1 = 0$, it gives $(u, F(K, u), \sigma[R_1])$ as input to $U_{1}$; else  it gives $(u, y_1, \sigma[R_1])$ as input to $U_{1}$. Let $b_1'$ be the output. Similarly, if 
         $b_2 = 0$, it gives $(w,F(K, w), \sigma[R_2])$ as input to $U_{2}$; else it gives $(w,y_2, \sigma[R_2])$ as input to $U_{2}$. Let $b_2'$ be the output.
        \item The outcome of the game is $1$ if $b_1' = b_1$ and $b_2' = b_2$.
    \end{enumerate}
    
    Denote by $\sf{IndCopyProtectionGame}(1^{\lambda}, \mathcal{A})$ a random variable for the output of the game.
    
    We say the scheme has indistinguishability anti-piracy security if for every polynomial-time quantum algorithm $\As$, there exists a negligible function $\negl(\cdot)$, for all $\lambda \in \mathbb{N}^+$, 
    \begin{align*}
        \Pr\left[b = 1, b \gets \sf{IndCopyProtectionGame}(1^{\lambda}, \mathcal{A}) \right] = \frac{1}{2} + \negl(\lambda) \,.
    \end{align*}
\end{definition}

Similarly to the relationship between CPA-style unclonable decryption (\Cref{def: regular antipiracy cpa}) and anti-piracy with random challenge inputs (Definition \ref{def: regular antipiracy random challenges}), it is not clear whether \Cref{def:indistinguishable_anti_piracy_PRF} implies \Cref{def:weak anti piracy} (this subtlety arises due to the fact that there are two parties involved, having to simultaneously make the correct guess). Thus, we will give separate statements and security proofs in the next section. 

\subsection{Preliminaries: Puncturable PRFs and related notions}
\label{sec: PRF prelim}

A \emph{puncturable} PRF is a PRF augmented with a procedure that allows to ``puncture'' a PRF key $K$ at a set of points $S$, in such a way that the adversary with the punctured key can evaluate the PRF at all points except the points in $S$. Moreover, even given the punctured key, an adversary cannot distinguish between a uniformly random value and the evaluation of the PRF at a point $S$ with respect to the original unpunctured key. Formally:

\begin{definition}[(Post-quantum) Puncturable PRF]
A PRF family $F: \{0,1\}^{n(\lambda)} \rightarrow \{0,1\}^{m(\lambda)}$ with key generation procedure $\kgen_F$ is said to be puncturable if there exists an algorithm $\puncture_F$, satisfying the following conditions:

\begin{itemize}
    \item \textbf{Functionality preserved under puncturing:} %
Let $S \subseteq \{0,1\}^{n(\lambda)}$. For all $x \in \{0, 1\}^{n(\lambda)}$ where $x \notin S$, we have that:
\begin{align*}
    \Pr[F(K,x) = F(K_S, x): K \gets \kgen(1^\lambda), K_S \gets \puncture_F(K, S)] = 1
\end{align*}

\item \textbf{Pseudorandom at punctured points:}
For every $QPT$ adversary $(A_1, A_2)$, there exists a negligible function $\negl$ such that the following holds. Consider an experiment where $K \gets \kgen_F(1^\lambda)$, $(S,\sigma) \gets A_1(1^{\lambda})$, and $K_S \gets \puncture_F(K, S)$. Then, for all $x \in S$,
\begin{align*}\left|\Pr[A_2(\sigma, K_S, S, F(K, x)) = 1] - \Pr_{r \gets \{0,1\}^{m(\lambda)}}[A_2(\sigma, K_S, S, r) = 1] \right|  \leq \negl(\lambda)
\end{align*}

\end{itemize}
\end{definition}

\begin{definition}
A statistically injective (puncturable) PRF family with (negligible) failure probability $\epsilon(\cdot)$ is a (puncturable) PRF family $F$ such that with probability $1-\epsilon(\lambda)$ over the random choice of key $K \gets \kgen_F (1^\lambda)$,
we have that $F(K, \cdot)$ is injective.
\end{definition}

We will also make use of extracting PRFs: these are PRFs that are strong extractors on their inputs in the following sense.

\begin{definition}[Extracting PRF]
An extracting (puncturable) PRF with error $\epsilon(\cdot)$ for min-entropy $k(\cdot)$ is a (puncturable)
PRF $F$ mapping $n(\lambda)$ bits to $m(\lambda)$ bits such that for all $\lambda$, if $X$ is any distribution over $n(\lambda)$
bits with min-entropy greater than $k(\lambda)$, then the statistical distance between $(K, F(K, X))$
and $(K, r \gets \{0,1\}^{m(\lambda)})$ is at most $\epsilon(\cdot)$, where $K \gets \kgen(1^\lambda)$. 
\end{definition}

Puncturable PRFs can be straightforwardly built by modifying the \cite{GGM86} tree-based construction of PRFs, which only assumes one-way functions. \cite{sahai2014use} showed that puncturable statistically injective PRFs and extracting puncturable PRFs with the required input-output size can be built from one-way functions as well.
These constructions can all be made post-quantum as shown in \cite{zhandry2012quantumprf}. Thus, the following theorems from \cite{sahai2014use} hold also against bounded quantum adversaries.

\begin{theorem}[\cite{sahai2014use} Theorem 1,  \cite{GGM86}]
If post-quantum one-way functions exist, then for all efficiently computable
functions $n(\lambda)$ and $m(\lambda)$, there exists a post-quantum puncturable PRF family that maps $n(\lambda)$ bits to $m(\lambda)$ bits.
\end{theorem}

\begin{theorem}[\cite{sahai2014use} Theorem 2]
\label{sw14_thm2}
If post-quantum one-way functions exist, then for all efficiently computable functions $n(\lambda)$, $m(\lambda)$, and $e(\lambda)$
such that $m(\lambda) \geq 2n(\lambda) + e(\lambda)$, there exists a post-quantum puncturable statistically injective PRF family with failure
probability $2^{-e(\lambda)}$
that maps $n(\lambda)$ bits to $m(\lambda)$ bits.
\end{theorem}

\begin{theorem}[\cite{sahai2014use} Theorem 3]
\label{sw14_thm3}
If post-quantum one-way functions exist, then for all efficiently computable functions $n(\lambda)$, $m(\lambda)$, $k(\lambda)$, and
$e(\lambda)$ such that $n(\lambda) \geq k(\lambda) \geq m(\lambda) + 2e(\lambda) + 2$, there exists a post-quantum extracting puncturable PRF family that
maps $n(\lambda)$ bits to $m(\lambda)$ bits with error $2^{-e(\lambda)}$
for min-entropy $k(\lambda)$.
\end{theorem}

\subsection{Construction}
In this section, we describe a construction of a copy-protection scheme for a class of PRFs. We will eventually reduce security of this construction to security of the single-decryptor encryption scheme of Section \ref{sec: unclonable dec from stronger monogamy}, and we will therefore inherit the same assumptions. A similar construction can be based on the single-decryptor encryption scheme of Section \ref{sec: unclonable dec witness enc}. %

Let $\lambda$ be the security parameter. Our construction copy-protects a PRF $F_1: [K_{\lambda}] \times  [N_{\lambda}] \rightarrow [M_{\lambda}]$ where $N = 2^{n(\lambda)}$ and $M = 2^{m(\lambda)}$, for some polynomials $n(\lambda)$ and $m(\lambda)$, satisfying $n(\lambda) \geq m(\lambda) + 2 \lambda + 4$. For convenience, we will omit writing the dependence on $\lambda$, when it is clear from the context. Moreover, $F_1$ should be a puncturable extracting PRF with error $2^{-\lambda-1}$ for min-entropy $k(\lambda) = n(\lambda)$ (i.e., a uniform distribution over all possible inputs). %
By \Cref{sw14_thm3}, such PRFs exist assuming post-quantum one-way functions. 

In our construction, we will parse the input $x$ to $F_1(K_1, \cdot)$ as three substrings $x_0 || x_1 || x_2$, where each $x_i$ is of length $\ell_i$ for $i \in \{0, 1, 2\}$ and $n = \ell_0 + \ell_1 + \ell_2$. $\ell_2 - \ell_0$ should also be large enough (we will specify later how large). Our copy-protection construction for $F_1$ will make use of the following additional building blocks:

\begin{enumerate}

    \item A puncturable statistically injective PRF $F_2$ with failure probability $2^{-\lambda}$ that accepts inputs of length $\ell_2$ and outputs strings of length $\ell_1$. By \Cref{sw14_thm2}, such a PRF exists assuming one-way functions exist, and as long as $\ell_1 \geq 2 \ell_2 + \lambda$.

    \item A puncturable PRF $F_3$ that accepts inputs of length $\ell_1$ and outputs strings of length $\ell_2$. By \Cref{lem:sw14_lem1} in \cite{sahai2014use},  assuming one-way functions exist, $F_3$ is a puncturable PRF. 
\end{enumerate}

Note that PRF $F_1$ is the PRF that we will copy-protect. The PRFs $F_2$ and $F_3$ are just building blocks in the construction.

\revise{Next, we describe a copy-protection scheme for the PRF $F_1$, using the above building blocks. The description is contained in Figures \ref{fig:prf_key} and \ref{fig:program_p}.  %
}

\begin{figure}[hpt]
    \centering
    \begin{gamespec}
    \begin{description}

    \item[]{ $\keygen(K_1)$:} Sample uniformly random subspaces $A_i$ of dimension $\lambda/2$ %
    and vectors $s_i, s'_i$ for $i = 1, 2, \cdots, \ell_0$.
   Sample PRF keys $K_2, K_3$ for $F_2, F_3$. Let $P$ be the program described in Figure \ref{fig:program_p}. Output the quantum key $\rho_{K} = (\{\ket{A_{i, s_i, s'_i}} \}_{i \in [\ell_0]}, \iO(P))$. 
    
    \item[] {$\eval(\rho_K, x)$:} Let $\rho_{K} = (\{\ket{A_{i, s_i, s'_i}} \}_{i \in [\ell_0]}, \iO(P))$. Parse $x$ as $x = x_0 || x_1 || x_2$ where $x_0$ is of length $\ell_0$. For all $i\in [\ell_0]$, if $x_{0, i}$ is $1$, apply $H^{\otimes n}$ to $\ket{A_{i, s_i, s'_i}}$. Otherwise, leave the state unchanged. 
        
        Let $\sigma$ be the resulting state (which can be interpreted as a superposition over tuples of $l_0$ vectors). Run $\iO(P)$ coherently on input $x$ and $\sigma$, and measure the final output register to obtain $y$. 

    \end{description}
    \end{gamespec}
    \caption{Quantum copy-protection scheme for PRFs.}
    \label{fig:prf_key}
\end{figure}

The program $P$, described in Figure \ref{fig:program_p}, takes as input $x$ and $\ell_0$ vectors $v_1, \cdots, v_{\ell_0}$, and has two modes. 
If $x$ is not in the sparse hidden trigger set (not passing the `if' check in the first line), the program is in the \emph{normal mode}: it outputs a PRF evaluation of $x$ if and only if every $v_i$ is in the appropriate coset. Otherwise, the program is in the \emph{hidden trigger mode}. It will compute a circuit $Q'$ from the input $x$ and output $Q'(v_1, \cdots, v_{\ell_0})$.
On almost all inputs except a sparse set of hidden triggers, the program runs in its normal mode. For $i \in [l_0]$, define the programs $R^0_i = \iO(A_i+s_i)$ and $R^1_i = \iO(A^\perp_i + s'_i)$ (where the inputs to $\iO$ should be appropriately padded).

\begin{figure}[hpt]
\centering
\begin{gamespec}
\textbf{Hardcoded:} Keys $K_1, K_2, K_3$, $R^0_i, R^1_i$ for all $i \in [\ell_0]$. 

On input $x = x_0 || x_1 || x_2$ and vectors $v_1, \cdots, v_{\ell_0}$:  
\begin{enumerate}
\item If $F_3(K_3, x_1) \oplus x_2 = x_0' || Q'$  and $x_0 = x'_0$ and $x_1 = F_2(K_2, x'_0 || Q')$: 

    \quad \textbf{Hidden Trigger Mode}: Treat $Q'$ as a (classical) circuit and output $Q'(v_1, \cdots, v_{\ell_0})$. 
\item Otherwise, check if the following holds: for all $i \in [\ell_0]$, $R^{x_{0, i}}_i(v_i) = 1$ (where $x_{0, i}$ is the $i$-th bit of $x_0$).
    
    \quad \textbf{Normal Mode}: If so, output $F_1(K_1, x)$. Otherwise, output $\bot$. 
\end{enumerate}
\end{gamespec}
\caption{Program $P$}
\label{fig:program_p}
\end{figure}

We prove the following theorem:%

\revise{
\begin{theorem}
\label{thm:PRF_antipiracy}
\unpredictableassumptions, our construction satisfies 
anti-piracy security (as in \Cref{def:weak anti piracy}).

\subexpassumptions, our construction satisfies 
anti-piracy security.
\end{theorem}
}

We show correctness of our construction in \Cref{sec: PRF correctness}, and anti-piracy security in \Cref{sec: PRF proof anti piracy}.

 The following theorem states that our construction also satisfies \Cref{def:indistinguishable_anti_piracy_PRF}.

\begin{theorem}
\label{thm:PRF_indistinguishable_antipiracy}
\unpredictableassumptions, our construction satisfies   indistinguishability-based anti-piracy security (as in Definition \ref{def:indistinguishable_anti_piracy_PRF}).

\subexpassumptions, our construction satisfies  indistinguishability-based anti-piracy security. 
\end{theorem}

We include the proof of the latter theorem in \Cref{sec:indistinguishable_anti_piracy_PRF}.

\subsection{Proof of Correctness}
\label{sec: PRF correctness}

First, it is easy to see that all procedures are efficient. 
We  then show that our construction satisfies correctness.

\begin{lemma}
    The above construction has correctness.%
\end{lemma}

\begin{proof}
First, we observe that for an input $x$, keys $K_2, K_2$, if the step 1 check in the program $P$ is not met, then the output of $\eval(\rho_K, x)$ will be the same as $F_1(K_1, \cdot)$ with probability $1$.

Therefore, let us assume there exists a fixed input $x^* = x_0^* || x_1^* || x_2^*$ such that for an inverse polynomial fraction of possible keys $K_2, K_3$, the step 1 check is passed. Define $\hat{x}_2^*$ be the first $\ell_0$ bits of $x_2^*$ and $\hat{F}_3(K_3, \cdot)$ be the function that outputs the first $\ell_0$ bits of $F_3(K_3, \cdot)$. $\hat{F}_3$ is a PRF because it is a truncation of another PRF $F_3$. To pass the step 1 check, $(x^*_0, x^*_1, \hat{x}^*_2)$ should at least satisfy: 
\begin{align*}
    \hat{F}_3(K_3, x^*_1) \oplus x^*_0 = \hat{x}^*_2. 
\end{align*}

Thus, for an inverse polynomial fraction of $K_3$, the above equation holds. This gives a non-uniform algorithm for breaking the security of $\hat{F}_3$ and violates the security of $F_3$ as a consequence: given oracle access to $\hat{F}_3(K_3, \cdot)$ for a random $K_3$, or a truly random function $f(\cdot)$, the algorithm simply queries on $x^*_1$ and checks if the output is $x^*_0 \oplus \hat{x}^*_2$; if yes, it outputs $1$ (indicating the function is $\hat{F}_3(K_3,\cdot))$; otherwise, it outputs $0$ (indicating the function is a truly random funtion).  Since the above equation holds for an inverse polynomial fraction of $K_3$, our non-uniform algorithm succeeds with an inverse polynomial probability. 

Since non-uniform security of PRFs can be based on non-uniform security of OWFs, the correctness of our construction relies on the existence of non-uniform secure post-quantum OWFs. 
\end{proof}

\subsection{Proof of Anti-Piracy Security}
\label{sec: PRF proof anti piracy}

In this subsection, we prove the anti-piracy security.  
Before proving anti-piracy, we give the following helper lemma from \cite{sahai2014use}. 
\begin{lemma}[Lemma 1 in \cite{sahai2014use}] \label{lem:sw14_lem1}
    Except with negligible probability over the choice of the key $K_2$, the following two statements hold: 
    \begin{enumerate}
        \item For any fixed $x_1$, there exists at most one pair $(x_0, x_2)$ that will cause the step 1 check in Program $P$ to pass. 
        \item There are at most $2^{\ell_2}$ values of $x$ that can cause the step 1 check to pass. 
    \end{enumerate}
\end{lemma}

The proof will exploit the sparse hidden triggers in the program $P$. Intuitively, we want to show that sampling a unifomly random input is indistinguishable from sampling an element from the sparse hidden trigger set. Then, we will reduce an adversary that successfully evaluates on hidden triggers to an adversary that breaks the single decryptor encryption scheme of Section \ref{sec:unclonable_dec}.

\begin{definition}[Hidden Trigger Inputs]
   An input $x$ is a hidden trigger input of the program $P$ (defined in \Cref{fig:program_p}) if it makes the step 1 check in the program be satisfied. 
\end{definition}

We will prove a lemma says that no efficient algorithm, given the quantum key, can distinguish between the following two cases: (i) sample two uniformly random inputs, and (ii) sample two inputs in the hidden trigger set.

Before describing the lemma, we describe an efficient procedure which takes as input an input/output pair for $F_1$, PRF keys $K_2, K_3$, descriptions of cosets, and produces a hidden trigger input.
\begin{definition}
\label{def: sample hidden triggers}
The procedure $\gentrigger$ takes as input $x_0$ (of length $\ell_0$), $y$ (of length $m$, where $m$ is the length of the output of $F_1$), two PRF keys $K_2, K_3$ and hidden cosets $\{A_i, s_i, s'_i\}_{i \in [\ell_0]}$: %
    \begin{enumerate} 
        \item Let $Q$ be the program (padded to length $\ell_2 - \ell_0$) that takes as input $v_1, \cdots, v_{\ell_0}$ and outputs $y$ if and only if for every input $v_i$, if $x_{0, i} = 0$, then $v_i$ is in $A_i + s_i$ and otherwise it is in $A^\perp_i + s'_i$. 
        \item $x'_1 \gets F_2(K_2, x_0 || Q)$;
        \item $x'_2 \gets F_3(K_3, x'_1) \oplus (x_0 ||Q)$.
        \item Output $x' = x_0||x'_1||x'_2$. 
    \end{enumerate}
\end{definition}
Note that for any $x_0, y$, $\gentrigger$ will produce an input $x'$ such that it starts with $x_0$ and the evaluation of $P$ on input $x'$ and valid vectors $v_1, \cdots, v_{\ell_0}$ is $y$.

\vspace{0.5em}

The following lemma says that any efficient algorithm cannot distinguish if it gets two inputs sampled  uniformly at random, or two hidden trigger inputs (sampled according to Definition \ref{def: sample hidden triggers}):
\begin{lemma}\label{lem:hidden_trigger}
    Assuming post-quantum $\iO$ and one-way functions, any efficient QPT algorithm $\As$ cannot win the following game with non-negligible advantage:
    \begin{itemize}
    \item  A challenger samples $K_1 \gets \setup(1^\lambda)$ and prepares a quantum key $\rho_K = (\{\ket{A_{i, s_i, s'_i}} \}_{i \in [\ell_0]}, \iO(P))$ (where recall that $P$ has keys $K_1, K_2, K_3$ hardcoded).

        \item The challenger then samples a random input $u \gets [N]$. Let $y_u = F_1(K_1, u)$. Parse the input as $u = u_0||u_1||u_2$. 
        
        Let $u' \gets \gentrigger(u_0, y_u, K_2, K_3, \{A_i, s_i, s'_i\}_{i \in [\ell_0]})$.
        
        \item Similarly, it samples a random input $w \gets [N]$. Let $y_w = F_1(K_1, w)$. Parse the input as $w = w_0||w_1||w_2$. 
        
        Let $w' \gets \gentrigger(w_0, y_w, K_2, K_3, \{A_i, s_i, s'_i\}_{i \in [\ell_0]})$.
        
        \item The challenger flips a coin $b$, and sends $(\rho_K, u, w)$ or $(\rho_K, u', w')$ to $\adv$, depending on the outcome. $\adv$ wins if it guesses $b$.
    \end{itemize}
\end{lemma}
One might wonder whether it is sufficient to just show a version of the above lemma which says that any efficient algorithm cannot distinguish if it gets \emph{one} uniformly random input or \emph{one} random hidden trigger input, and use a hybrid argument to show indistinguishability in the case of two samples. However, this is not the case, as one cannot efficiently sample a random hidden trigger input when given only the public information in the security game (in particular $\gentrigger$ requires knowing $K_2, K_3$), and so the typical reduction would not go through.

\vspace{1em}

Next, we show that if \Cref{lem:hidden_trigger} holds, then our construction satisfies anti-piracy security \Cref{thm:PRF_antipiracy}. 
After this, to finish the proof, we will only need to prove \Cref{lem:hidden_trigger}. The core of the latter proof is the ``hidden trigger'' technique used in~\cite{sahai2014use}, which we will prove in \Cref{sec: PRF hidden trigger}. %

%

%
%

%
\begin{proof}[Proof for \Cref{thm:PRF_antipiracy}]
  We mark the changes between the current hybrid and the previous {\color{red} in red}.

     \paragraph{Hybrid 0.} 
     Hybrid 0 is the original anti-piracy security game.
     
    \begin{enumerate}
        \item  The challenger samples $K_1 \gets \setup(1^\lambda)$ and $\rho_K = (\{\ket{A_{i, s_i, s'_i}} \}_{i \in [\ell_0]}, \iO(P)) \gets \keygen(K_1)$. Note that here $P$ hardcodes $K_1, K_2, K_3$.

        \item $\As$ upon receiving $\rho_K$, it runs and prepares a pair of (potentially entangled) quantum states $\sigma[R_1], \sigma[R_2]$ as well as quantum circuits $U_1, U_2$. 
        \item The challenger also prepares two inputs $u, w$ as follows:
            \begin{itemize}
                \item It samples $u$ uniformly at random. 
                Let $y_u  = F_1(K_1, u)$.

                \item It samples $w$ uniformly at random.
                Let $y_w = F_1(K_1, w)$.

            \end{itemize}
        \item The outcome of the game is 1 if and only if both quantum programs successfully produce $y_u$ and $y_w$ respectively. 
    \end{enumerate}

    \paragraph{Hybrid 1} 
    The difference between Hybrids 0 and 1 corresponds exactly to the two cases that the adversary needs to distinguish between in the game of \Cref{lem:hidden_trigger}.

    \begin{enumerate}
        \item  The challenger samples $K_1 \gets \setup(1^\lambda)$ and $\rho_K = (\{\ket{A_{i, s_i, s'_i}} \}_{i \in [\ell_0]}, \iO(P)) \gets \keygen(K_1)$. Note that here $P$ hardcodes $K_1, K_2, K_3$.

        \item $\As$ upon receiving $\rho_K$, it runs and prepares a pair of (potentially entangled) quantum states $\sigma[R_1], \sigma[R_2]$ as well as quantum circuits $U_1, U_2$. 
        \item The challenger also prepares two inputs $u', w'$ as follows:
            \begin{itemize}
                \item It samples $u = u_0||u_1||u_2$ uniformly at random. 
                Let $y_u  = F_1(K_1, u)$. 
                
                {\color{red}
                Let $u' \gets \gentrigger(u_0, y_u, K_2, K_3, \{A_i, s_i, s'_i\}_{i \in [\ell_0]})$. }

                \item It samples $w = w_0||w_1||w_2$ uniformly at random.
                Let $y_w = F_1(K_1, w)$. 
                
                {\color{red}
                Let $w' \gets \gentrigger(w_0, y_w, K_2, K_3, \{A_i, s_i, s'_i\}_{i \in [\ell_0]})$. }

            \end{itemize}
        \item The outcome of the game is 1 if and only if both quantum programs successfully produce $y_u$ and $y_w$ respectively. 
    \end{enumerate}

 Assume there exists an algorithm that  distinguishes Hybrid 0 and 1 with \emph{non-negligible} probability $\epsilon(\lambda)$, then these exists
      an algorithm that breaks the game in \Cref{lem:hidden_trigger} with probability $\epsilon(\lambda) - \negl(\lambda)$.

     The reduction algorithm receives $\rho_k$ and $u, w$ or $u', w'$ from the challenger in \Cref{lem:hidden_trigger}; it computes $y_u, y_w$ using $\iO(P)$ on the received inputs respectively and gives them to the quantum decryptor states $\sigma[R_1], \sigma[R_2]$. If they both decrypt correctly, then the reduction outputs 0 (i.e. it guess that sampling was uniform), otherwise it outputs 1 (i.e. it guesses that hidden trigger inputs were sampled).

    \paragraph{Hybrid 2.}  In this hybrid, if $u_0 \ne w_0$ (which happens with overwhelming probability),  $F_1(K_1, u)$ and $F_1(K_1, w)$ are replaced with truly random strings. Since both inputs have enough min-entropy $\ell_1 + \ell_2 \geq m + 2\lambda  + 4$ (as $u_1||u_2$ and $w_1||w_2$ are completely uniform and not given to the adversary) and $F_1$ is an extracting puncturable PRF, both outcomes $y_u, y_w$ are statistically close to independently random outcomes.
    Thus, Hybrid 1 and Hybrid 2 are statistically close. 
    
    \begin{enumerate}
        \item  The challenger samples $K_1 \gets \setup(1^\lambda)$ and $\rho_K = (\{\ket{A_{i, s_i, s'_i}} \}_{i \in [\ell_0]}, \iO(P)) \gets \keygen(K_1)$. Note that here $P$ hardcodes $K_1, K_2, K_3$. 

        \item $\As$ upon receiving $\rho_K$, it runs and prepares a pair of (potentially entangled) quantum states $\sigma[R_1], \sigma[R_2]$ as well as quantum circuits $U_1, U_2$. 
        \item The challenger also prepares two inputs $u', w'$ as follows:
            \begin{itemize}
                \item It samples {\color{red} $u_0$ uniformly at random}. It then samples {\color{red} a uniformly random $y_u$}. 
                
                Let $u' \gets \gentrigger(u_0, y_u, K_2, K_3, \{A_i, s_i, s'_i\}_{i \in [\ell_0]})$.

                \item It samples {\color{red} $w_0$ uniformly at random}. It then samples {\color{red} a uniformly random $y_w$}. 
                
                Let $w' \gets \gentrigger(w_0, y_w, K_2, K_3, \{A_i, s_i, s'_i\}_{i \in [\ell_0]})$. 

            \end{itemize}
        \item The outcome of the game is 1 if and only if both quantum programs successfully produce $y_u$ and $y_w$ respectively. 
    \end{enumerate}

        \paragraph{Hybrid 3.} The game in this hybrid has exactly the same distribution as that of Hybrid 2 (in the sense that all sampled values are distributed identically). We only change the order in which some values are sampled, and recognize that certain procedures become identical to encryptions in our single-decryptor encryption scheme from Section \ref{sec:unclonable_dec}.
        Thus, $\As$ wins the game with the same probability as in Hybrid 2. 

        \begin{enumerate}
            \item {\color{red} The challenger first samples $\{A_i, s_i, s'_i\}_{i \in [\ell_0]}$ and prepares the quantum states $\{\ket{A_{i, s_i, s'_i}} \}_{i \in [\ell_0]}$}. It treat the the quantum states $\{\ket{A_{i, s_i, s'_i}} \}_{i \in [\ell_0]}$ as the quantum decryption key $\rho_{\sk}$ for our single-decryptor encryption scheme and the secret key $\sk$ is $\{A_i, s_i, s'_i\}_{i \in [\ell_0]}$.  Similarly, let $\pk = \{R^0_i, R^1_i\}_{i \in [\ell_0]}$ where $R^0_i = \iO(A_i+s_i)$ and $R^1_i = \iO(A^\perp_i + s'_i)$. 
        
            \item  It samples $y_u, y_w$  uniformly at random. {\color{red} Let $(u_0, Q_0) \gets \enc(\pk, y_u)$ and $(w_0, Q_1) \gets \enc(\pk, y_w)$ where $\enc(\pk, \cdot)$ is the encryption algorithm of the single-decryptor encryption scheme of Construction \ref{cons: unclonable dec}}. 
            
            \item The challenger sets $\rho_K = (\{\ket{A_{i, s_i, s'_i}} \}_{i \in [\ell_0]}, \iO(P))$.
            
            \item $\As$ upon receiving $\rho_K$, it runs and prepares a pair of (potentially entangled) quantum states $\sigma[R_1], \sigma[R_2]$ as well as quantum circuits $U_1, U_2$. 
            \item The challenger also prepares two inputs $u', w'$ as follows (as $\gentrigger$ does):
                \begin{itemize}
                    \item 
                    Let $u_1' \gets F_2(K_2, u_0 || Q_0)$ and $u_2'\gets F_3(K_3, u_1') \oplus (u_0||Q_0)$. Let $u' = u_0 || u_1' || u_2'$. 
                    
                    \item 
                    Let $w_1' \gets F_2(K_2, w_0 || Q_1)$ and $w_2'\gets F_3(K_3, w_1') \oplus (w_0||Q_1)$. Let $w' = w_0 || w_1' || w_2'$. 
                \end{itemize}
            \item The outcome of the game is 1 if and only if both quantum programs successfully produce ${y_u}$ and ${y_w}$ respectively. 
        \end{enumerate}
        Note that the only differences of Hybrids 2 and 3 are the orders of executions. Namely, in Hybrid 3, $\{A_i, s_i, s'_i\}$ are sampled much earlier than when $\rho_k$ is prepared. Similarly, the obfuscation programs sampled in $\gentrigger$ are now sampled much earlier than sampling $u'$ and $w'$. We write Hybrid 3 in a way that is similar to the weak anti-piracy security game of the single-decryptor encryption scheme of Construction \ref{cons: unclonable dec}.
        
        \vspace{1em}
    
        Given an algorithm $\As$ that wins the game in Hybrid 3 with non-negligible probability $\gamma(\lambda)$, we can build another algorithm $\Bs$ that breaks the (regular) $\gamma$-anti-piracy security with random challenge plaintexts (see \Cref{def:weak_ag_random}) of the underlying single-decryptor encryption scheme. 
        \begin{itemize}
            \item $\Bs$ plays as the challenger in the game of Hybrid 3. 
            \item $\Bs$ receives $\rho_{\sk} = \{\ket{A_{i, s_i, s'_i}} \}_{i \in [\ell_0]}$ and $\pk = \{\iO(A_i + s_i), \iO(A^\perp_i + s'_i)\}_{i \in [\ell_0]}$ in the anti-piracy game of single-decryptor encryption.
            
            \item $\Bs$ prepares $K_1, K_2, K_3$ and the program $P$. Let $\rho_K = (\{\ket{A_{i, s_i, s'_i}} \}_{i \in [\ell_0]}, \iO(P))$.
            
            \item $\Bs$ gives $\rho_{K}$ to $\As$, and $\As$ prepares a pair of (potentially entangled) quantum states $\sigma[R_1], \sigma[R_2]$ as well as quantum circuits $U_1, U_2$.  

            \item $\Bs$ outputs the decryptors $(\sigma[R_1], \P_1)$ and $(\sigma[R_2], \P_2)$, where $\P_1$ and $\P_2$ are defined as follows: on input ($\rho_1,\ct_1 = (u_0 || Q_1))$ and  $(\rho_2,\ct_2 = (w_0 || Q_2))$ respectively (where $\ct_1$ and $\ct_2$ represent encryptions of random $y_u$ and $y_w$), $\P_1$ and $\P_2$ behave respectively as follows:
            \begin{itemize}
                    \item $\P_1$:
                   Let $u_1' \gets F_2(K_2, u_0 || Q_0)$ and $u_2'\gets F_3(K_3, u_1') \oplus (u_0||Q_0)$. Let $u' = u_0 || u_1' || u_2'$. Run $(\rho_1, U_1)$ on $u'$.
                    
                    \item $\P_2$:
                    Let $w_1' \gets F_2(K_2, w_0 || Q_1)$ and $w_2'\gets F_3(K_3, w_1') \oplus (w_0||Q_1)$. Let $w' = w_0 || w_1' || w_2'$. Run $(\rho_2, U_{2})$ on $w'$ respectively. 
                \end{itemize}  

        \end{itemize}

        We know that whenever $\As$ succeeds in the game of Hyb 3, it outputs $y_u, y_w$ correctly. Thus, the programs prepared by $\Bs$ successfully decrypts encryptions of uniformly random plaintexts.  Thus, $\Bs$ breaks $\gamma$-anti-piracy security with random challenge plaintexts. 
\end{proof}

\printbibliography

\appendix

\section{Additional Preliminaries}
%

\subsection{Quantum Computation and Information}
\label{appendix:quantum_info}
A quantum system $Q$ is defined over a finite set $B$ of classical states. In this work we will consider $B = \{0,1\}^n$. A \textbf{pure state} over $Q$ is a unit vector in $\mathbb{C}^{|B|}$, which assigns a complex number to each element in $B$. In other words, let $|\phi\rangle$ be a pure state in  $Q$, we can write $|\phi\rangle$ as:
\begin{equation*}
    |\phi\rangle = \sum_{x \in B} \alpha_x |x\rangle
\end{equation*}
where $\sum_{x \in B} |\alpha_x|^2 = 1$ and $\{|x\rangle\}_{x \in B}$ is called 
the ``\textbf{computational basis}'' of $\mathbb{C}^{|B|}$. The computational basis forms an orthonormal basis of $\mathbb{C}^{|B|}$.

Given two quantum systems $R_1$ over $B_1$ and $R_2$ over $B_2$, we can define a \textbf{product} quantum system $R_1 \otimes R_2$ over the set $B_1 \times B_2$. Given $|\phi_1\rangle \in R_1$ and $|\phi_2\rangle \in R_2$, we can define the product state $|\phi_1\rangle \otimes |\phi_2\rangle \in R_1 \otimes R_2$. 

We say $|\phi\rangle \in R_1 \otimes R_2$ is \textbf{entangled} if there does not exist 
$|\phi_1\rangle \in R_1$ and $|\phi_2\rangle \in R_2$ such that $|\phi\rangle = |\phi_1\rangle \otimes |\phi_2\rangle$. For example, consider $B_1 = B_2 = \{0,1\}$
and $R_1 = R_2 = \mathbb{C}^2$, $|\phi\rangle = \frac{|00\rangle + |11\rangle}{\sqrt{2}}$ is entangled. Otherwise, we say $|\phi\rangle$ is un-entangled. 

A mixed state is a collection of pure states $\ket{\phi_i}$ for $i\in [n]$, each with associated probability $p_i$, with the condition $p_i\in [0,1]$ and $\sum_{i=1}^n p_i = 1$. A mixed state can also be represented by the density matrix: $\rho:= \sum_{i=1}^n p_i \ket{\phi_i}\bra{\phi_i}$. 

\textbf{Partial Trace}. 
For two subsystems $R_1$ and $R_2$ making up the composite system described
by the density matrix $\rho$. The partial trace over the $R_2$ subsystem, denoted $\Tr_{R_2}$, is defined as
$\Tr_{R_2}[\rho] :=  \sum_j
(I_{R_1} \otimes \langle j|_{R_2}) \rho (I_{R_1} \otimes |j\rangle_{R_2})$. 
where $\{|j\rangle \}$ is any orthonormal basis for subsystem $R_2$.

For a quantum state $\sigma$ over two registers $R_1, R_2$, we denote the state in $R_1$ as $\sigma[R_1]$, where $\sigma[R_1]= \Tr_{R_2}[\sigma]$ is a partial trace of $\sigma$. Similarly, we denote $\sigma[R_2]= \Tr_{R_1}[\sigma]$.

\textbf{Purification of mixed states}. For a mixed state $\rho$ over system $Q$, there exists another space $Q'$ and a pure state $\ket \psi$ over $Q \otimes Q'$ such that $\rho$ is a partial trace of $\ket \psi \bra \psi$ with respect to $Q'$.

\vspace{1em}

A pure state $|\phi\rangle$ can be manipulated by a unitary transformation $U$. The resulting state $|\phi'\rangle = U |\phi\rangle$. 

We can extract information from a state $|\phi\rangle$ by performing a \textbf{measurement}. A measurement specifies an orthonormal basis, typically the computational basis, and the probability of getting result $x$ is $|\langle x | \phi \rangle|^2$. After the measurement, $|\phi\rangle$ ``collapses'' to the state $|x\rangle$ if the result is $x$. 
                
                For example, given the pure state $|\phi\rangle = \frac{3}{5} |0\rangle + \frac{4}{5} |1\rangle$ measured under $\{|0\rangle ,|1\rangle \}$, with probability $9/25$ the result is $0$ and $|\phi\rangle$ collapses to $|0\rangle$; with probability $16/25$ the result is $1$ and $|\phi\rangle$ collapses to $|1\rangle$.

    We finally assume a quantum computer can  implement any unitary transformation (by using these basic gates, Hadamard, phase, CNOT and $\frac{\pi}{8}$ gates), especially the following two unitary transformations:
        \begin{itemize}
            \item \textbf{Classical Computation:} Given a function $f : X \to Y$, one can implement a unitary $U_f$ over $\mathbb{C}^{|X|\cdot |Y|} \to \mathbb{C}^{|X| \cdot |Y|}$ such that for any $|\phi\rangle = \sum_{x \in X, y \in Y} \alpha_{x, y} |x, y\rangle$, 
            \begin{equation*}
                U_f |\phi\rangle = \sum_{x \in X, y \in Y} \alpha_{x, y} |x, y \oplus f(x)\rangle
            \end{equation*}
            
            Here, $\oplus$ is a commutative group operation defined over $Y$.
            
            \item \textbf{Quantum Fourier Transform:} Let $N = 2^n$. Given a quantum state $|\phi\rangle = \sum_{i=0}^{2^n-1} x_i |i\rangle$, by applying only $O(n^2)$ basic gates, 
    	one can compute $|\psi\rangle =  \sum_{i=0}^{2^n-1} y_i |i\rangle$ where the sequence $\{y_i\}_{i=0}^{2^n-1}$ is the sequence achieved by applying the 
	classical Fourier transform ${\sf QFT}_N$ to the sequence $\{x_i\}_{i=0}^{2^n-1}$: 
		\begin{equation*}
			y_k = \frac{1}{\sqrt{N}} \sum_{i=0}^{2^n-1} x_i \omega_n^{i k} 
		\end{equation*} 
	where $\omega_n = e^{2 \pi i / N}$, $i$ is the imaginary unit.

	One property of {\sf QFT} is that by preparing $|0^n\rangle$ and 
	applying ${\sf QFT}_2$ to each qubit, 	$\left({\sf QFT}_2 |0\rangle\right)^{\otimes n} = \frac{1}{\sqrt{2^n}} \sum_{x \in \{0,1\}^n} |x\rangle$ which is a uniform superposition over all possible $x \in \{0,1\}^n$.
        \end{itemize}
    
    For convenience, we sometimes omit writing the normalization of a pure state.

\section{Compute-and-Compare Obfuscation for (Sub-Exponentially) Unpredictable Distributions}
\label{sec:CC_quantum_aux}

In this section, we prove compute-and-compare obfuscation for sub-exponentially unpredictable distributions exists assuming the existence of post-quantum \iO{} and the quantum hardness of LWE. We show a similar statement about compute-and-compare obfuscation for any unpredictable distributions assuming post-quantum $\iO$ and post-quantum extremely lossy functions. We focus on the first result and it extends to the second case with little effort. 

Our proof follows the steps below:

\begin{enumerate}
    \item Assuming the quantum hardness of LWE, there exist lossy functions with any sub-linear  residual leakage~\cite{peikert2011lossy}. 
    \item Assuming lossy functions with any sub-linear residual leakage, there exist PRGs with sub-exponentially unpredictable seeds (quantum auxiliary input). The proof constitutes that of \cite{zhandry2019magic}, with the building block ELFs (extremely loss functions) being replaced with plain lossy functions and the last step of invoking Goldriech-Levin \cite{goldreich1989hard} being replaced with a quantum version of Goldreich-Levin \cite{adcock2002quantum}. For the quantum version of Goldreich-Levin, we prove a variant which holds against quantum auxiliary input.  
    \item Finally, assuming PRGs with sub-exponentially unpredictable seeds and post-quantum \iO{}, there exists compute-and-compare obfuscation for sub-exponentially unpredictable distributions~\cite{wichs2017obfuscating}. 
\end{enumerate}
As proved in \cite{wichs2017obfuscating}, in Step 3, we can build such compute-and-compare obfuscation solely based on the quantum hardness of LWE. However, as all the constructions in this work require \iO{} as a building block, we focus on the simpler construction which is based on \iO. Thus, we have the following theorem: 
\begin{theorem}\label{thm:CC_subexp_from_LWE_iO}
    Assuming the existence of post-quantum $\iO$ and the
quantum hardness of LWE, there exist obfuscators as in \Cref{def: cc obf}.  for sub-exponentially unpredictable distributions.
\end{theorem}

In the rest of the section, we will introduce all building blocks and prove Step 2. Step 1 and 3 follow directly from previous work. 

Similarly, we can prove the following theorem: 
\begin{theorem}\label{thm:CC__from_ELF_iO}
    Assuming the existence of post-quantum $\iO$ and post-quantum
extremely lossy functions, there exist obfuscators as in \Cref{def: cc obf}. for any unpredictable distributions.
\end{theorem}
\Cref{thm:CC__from_ELF_iO} directly follows all three steps above without even replacing the building block ELFs with plain lossy functions. Thus, we omit the proof here.  However, currently we do not know any post-quantum construction for extremly lossy functions. 

\subsection{Preliminaries}

We first introduce lossy functions. For the purpose of this work, we ignore the need of trapdoors in the definition. The definition is taken verbatim from \cite{peikert2011lossy}. 

Define the following quantities: the security
parameter is $\lambda$, $n(\lambda) = \poly(\lambda)$ represents the input length of the function, $m(\lambda) = \poly(\lambda)$ represents the output length and $k(\lambda) \leq n(\lambda)$ represents
the lossiness of the collection. For convenience, we also define the residual leakage $r(\lambda) := n(\lambda)-k(\lambda)$. For all these quantities, we often omit the dependence on $\lambda$.
\begin{definition}[Lossy Functions~\cite{peikert2011lossy}]
A collection of $(n, k)$-lossy functions is given by a tuple of (possibly probabilistic) polynomial-time algorithms $(S_{\lf}, F_{\lf})$ having the properties below. For notational convenience, define the algorithms $S_{\inj}(\cdot) := S_{\lf}(\cdot, 1)$ and $S_\lossy(\cdot) := S_{\lf}(\cdot, 0)$.
\begin{enumerate}
    \item \emph{Easy to sample an injective function:} $S_\inj(1^\lambda)$ outputs $s$ where $s$ is a function index, with overwhelming probability, $F_{\lf}(s, \cdot)$ computes a (deterministic) injective function $f_s(\cdot)$ over the domain $\{0,1\}^{n(\lambda)}$. 
    
    For notational convenience, we assume $S_\inj(1^\lambda)$ samples a function description $f_s(\cdot)$.  
    \item \emph{Easy to sample a lossy function:} $S_\lossy(1^\lambda)$ outputs $s$ where $s$ is a function index, $F_{\lf}(s, \cdot)$ computes a (deterministic) function $f_s(\cdot)$ over the domain $\{0,1\}^{n(\lambda)}$ whose image has size at most $2^{r} = 2^{n - k}$, with overwhelming probability. 
    
    For notational convenience, we also assume $S_\lossy(1^\lambda)$  samples a function description $f_s(\cdot)$.  
    \item \emph{Hard to distinguish injective from lossy:} the outputs (function descriptions) of  $S_\inj(1^\lambda)$ and $S_\lossy(1^\lambda)$ are computationally indistinguishable. 
\end{enumerate}
\end{definition}

\begin{theorem}[Theorem 6.4, \cite{peikert2011lossy}]
    Assuming ${\sf LWE}_{q, \chi}$ is hard for some $q, \chi$, there exists a collection of $(n , k)$ lossy functions where the residual leakage $r$ is $r = n^c$ for any constant $c > 0$.
\end{theorem}
\begin{remark}
This can be done by carefully choosing parameters $c_1 = n^{1-c}$, $c_2$ as some constant, $c_3 = 1/c$ in Theorem 6.4 of \cite{peikert2011lossy}. 
\end{remark}
\begin{remark}
For the definition of extremely lossy functions, 
\begin{itemize}
    \item In bullet (2): $S_\lossy(1^\lambda)$ takes another parameter $r \in [2^{n}]$ and $f_s$ sampled from $S_\lossy(1^\lambda, r)$ has image size $r$, with overwhelming probability; 
    \item In bullet (3): For any polynomial $p$ and inverse polynomial function $\delta$ (in $n$), there is a polynomial $q$ such that: for any adversary $\As$ running in time at most $p$, and any $r \in [q(n), M]$, it can not distinguish the outputs of $S_\inj(1^\lambda)$ between $S_\lossy(1^\lambda)$ with advantage more than $\delta$. 
\end{itemize}
\end{remark}

\vspace{1em}

Second, we introduce PRGs with sub-exponentially unpredictable seeds~\cite{zhandry2019magic, wichs2017obfuscating}. 
\begin{definition}[PRG with Sub-Exponentially Unpredictable Seeds~\cite{zhandry2019magic}]
A family of pseudorandom generators $H: \Xs \to \Ys$ is secure for sub-exponentially unpredictable seeds if, for any sub-exponentially unpredictable distribution on $(X, \mathcal{H}_Z)$, no efficient
adversary can distinguish $(H, \rho_z, H(x))$ from $(H, \rho_z, S)$ where $(x, \rho_z) \gets D$ and $S \gets U_Y$, where $\rho_z$ is a quantum auxiliary input.
\end{definition}

The following theorem follows from Appendix A in \cite{wichs2017obfuscating}. Moreover, \iO{} in the theorem statement can be further replaced with LWE using the construction in their work.  
\begin{theorem}[\cite{wichs2017obfuscating}]
    Assuming PRGs with sub-exponentially unpredictable seeds and \iO{}, there exists compute-and-compare obfuscation for sub-exponentially unpredictable distributions. 
\end{theorem}

\subsection{PRGs with Sub-Exponentially Unpredictable Seeds}
To prove \Cref{thm:CC_subexp_from_LWE_iO}, we only need to prove PRGs with sub-exponentially unpredictable seeds can be built from lossy functions. 

\label{sec:hcbit}

Most of the proof follows \cite{zhandry2019magic}, except we are working with plain lossy functions (not extremely lossy functions), sub-exponentially unpredictable distributions and potentially quantum auxiliary information. We first look at the construction.

\begin{construction}\label{constr:hcbit}Let $q$ be the input length and $m$ be the output length.  Let $\lambda$ be a security parameter.  We will consider inputs $x$ as $q$-dimensional vectors $\xv\in\F_2^q$.  Let $\LF$ be a lossy function (with some sub-linear residual leakage, which will be specified later).  Let $M=2^{m+\lambda+1}$, and let $n$ be the output length of the lossy function.  Set $N=2^{n}$.  Let $\ell$ be some polynomial in $m,\lambda$ to be determined later.  First, we will construct a function $H'$ as follows.

Choose random $f_1,\dots,f_\ell\gets \LF.S_\inj(1^\lambda)$ where $f_i:[M]\rightarrow[N]$, and let $h_1,\dots,h_{\ell-1}:[N]\rightarrow[M/2]=[2^{m+\lambda}]$ and $h_\ell:[N]\rightarrow[2^m]$  be sampled from pairwise independent and uniform function families.  Define $\fv=\{f_1,\dots,f_\ell\}$ and $\hv=\{h_1,\dots,h_\ell\}$. Define $H'_i:\{0,1\}^i\rightarrow [M/2]$ (and $H'_\ell:\{0,1\}^\ell\rightarrow[2^m]$) as follows:
\begin{itemize}
	\item $H'_0()=1\in[2^{m+\lambda}]$
	\item $H'_i(\bv_{[1,i-1]},b_i):$ compute $y_i=H'_{i-1}(\bv_{[1,i-1]})$, $z_i\gets f_i(y_i||b_i)$, and output $y_{i+1}\gets h_i(z_i)$
\end{itemize}

Then we set $H'=H'_\ell$.  To define $H$, choose a random matrix $\Rm\in\F_2^{\ell\times q}$.  The description of $H$ consists of $\fv,\hv,\Rm$.  We set $H(x)=H'(\Rm\cdot \xv)$.  A diagram of $H$ is given in Figure~\ref{fig:hcbit}.
\end{construction} 

\begin{figure}[ht]
	\begin{center}\includegraphics[scale=0.45]{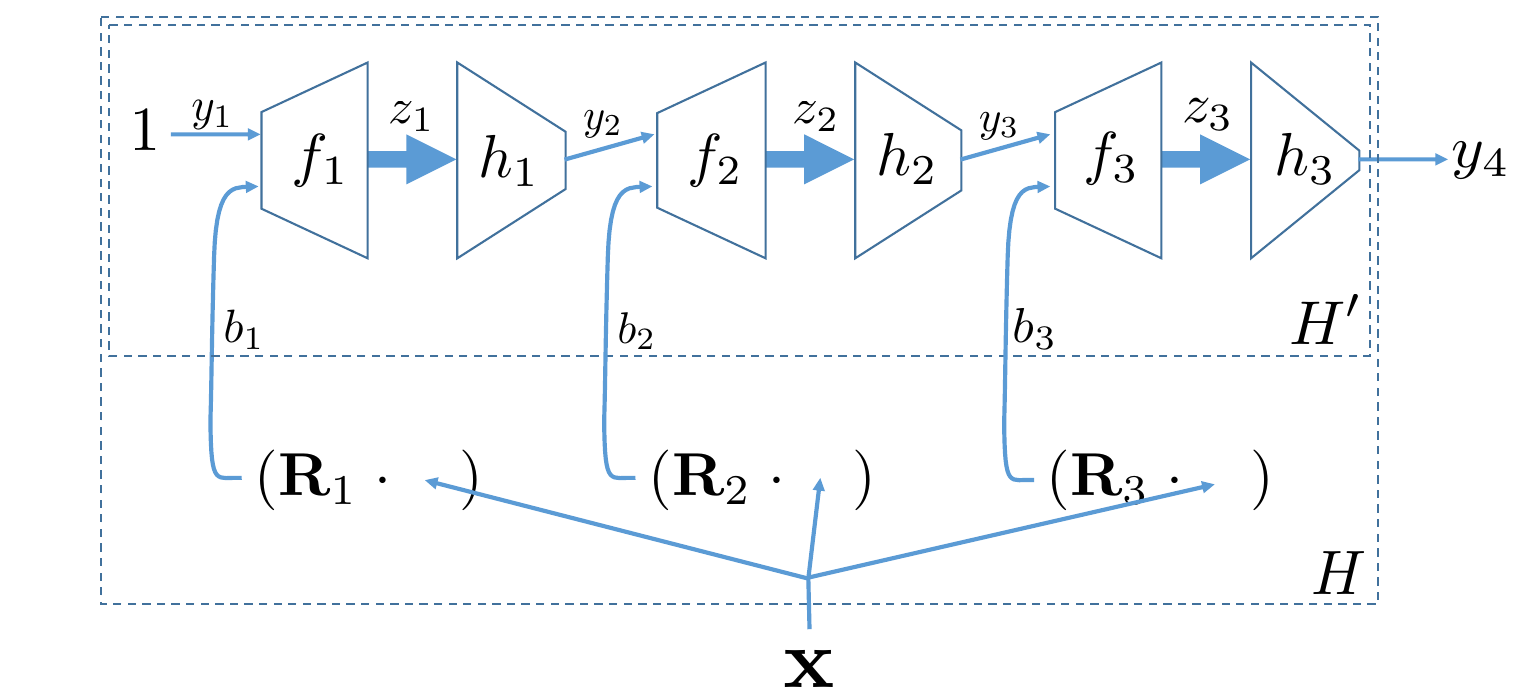}\end{center}
	\caption{An example instantiation for $\ell=3$, from \cite{zhandry2019magic}. } \label{fig:hcbit}
\end{figure}

We have the following theorem, which will finish the proof for \Cref{thm:CC_subexp_from_LWE_iO}. 
\begin{theorem} \label{thm:LF_to_PRG}
    Assuming lossy functions with sub-linear residual leakage, there exist PRGs with sub-exponentially unpredictable seeds (quantum auxiliary input).    
\end{theorem}

First, we have the claim from \cite{zhandry2019magic}. 

\begin{claim}[Claim 6.4, \cite{zhandry2019magic}]\label{claim:hc1} If $\ell\geq m+\lambda$, and if $\bv$ is drawn uniformly at random, then $(H',H'(\bv))$ is statistically close to $(H',R)$  where $R$ is uniformly random in $[2^m]$.\end{claim}

We will thus set $\ell=m+\lambda$ in our construction of $H'$. We now present our main theorem (\Cref{thm:LF_to_PRG}) of the section.

\begin{proof}[\Cref{thm:LF_to_PRG}] We show that $H$ in Construction~\ref{constr:hcbit} is a pseudorandom generator that is secure for sub-exponentially unpredictable seeds. 

Let $\lambda$ be the security parameter.
Let $(x, \rho_z) \gets D$ be a sub-exponentially unpredictable distribution where $|x| = q$ (the input length of the PRG). In other words, there is no efficient algorithm that given $\rho_z$, outputs $x$ with probability more than $2^{-\lambda^{c_1}}$ for some constant $0 < c_1 \leq 1$. 
Let $\LF$ be a $(m+\lambda+1, k)$-lossy function with sub-linear residual leakage $r$ such that $2^{-r} \gg 2^{-\lambda^{c_1}}$. Note that there always exists a constant $c$ such that $r = 2^{(m+\lambda+1)^c} \ll 2^{\lambda^{c_1}}$. 

Recall that $H(\xv)=H'(\Rm\cdot\xv)$, and that $H'(\bv)$ is statistically close to random when $\bv$ is random (by \Cref{claim:hc1}).  Therefore, it suffices to show that the following distributions are indistinguishable:
\begin{align*}
    (\fv,\hv,\Rm,\rho_z,H'(\Rm\cdot\xv)) \text{ v.s. }  (\fv,\hv,\Rm,\rho_z,H'(\bv)) \text{ for a uniformly random $\bv$}.
\end{align*}

Suppose an adversary $\adv$ has non-negligible advantage $\epsilon$ in distinguishing the two distributions.  Define $\bv^{(i)}$ so that the first $i$ bits of $\bv^{(i)}$ are equal to the first $i$ bits of $\Rm\cdot\xv$, and the remaining $\ell-i$ bits are chosen uniformly at random independently of $\xv$.  Define {\bf Hybrid $i$} to be the case where $\adv$ is given the distribution $(\fv,\hv,\Rm,\rho_z,H'(\bv^{(i)}))$.

We know that $\adv$ distinguishes {\bf Hybrid 0} from {\bf Hybrid $\ell$} with probability $\epsilon$.  Choose an $i$ uniformly at random from $[\ell]$.  Then the adversary distinguishes {\bf Hybrid $(i-1)$} from {\bf Hybrid $i$} with expected advantage at least $\epsilon/\ell$.  Next, observe that since bits $i+1$ through $\ell$ are random in either case, they can be simulated independently of the challenge.  Moreover, $H'(\bv)$ can be computed given $H_{i-1}'(\bv_{[i-1]})$, $b_i$ (be random or equal to $\Rm_i,\xv$), and the random $b_{i+1},\dots,b_\ell$.  Thus, we can construct an adversary $\adv'$ that distinguishes the following distributions:
\[(i,\fv,\hv,\Rm_{[i-1]},\rho_z,H_{i-1}'(\Rm_{[i-1]}\cdot\xv),\Rm_i,\Rm_i\cdot\xv)\text{ and }(i,\fv,\hv,\Rm_{[i-1]},\rho_z,H_{i-1}'(\Rm_{[i-1]}\cdot\xv),\Rm_i,b_i)\]
with advantage $\epsilon/\ell$, where $i$ is chosen randomly in $[\ell]$, $\Rm_{[i-1]}$ consists of the first $i-1$ rows of $\Rm$, $\Rm_i$ is the $i$th row of $\Rm$, and $b_i$ is a random bit.

$\adv'$ cannot distinguish $f_i$ generated as $\LF.S_\lossy(1^\lambda)$ from the honest $f_i$ generated from $\LF.S_\inj(1^\lambda)$, except with negligible probability.  This means, if we generate $f_i\gets\LF.S_\lossy(1^\lambda)$, we have that $\adv'$ still distinguishes the distributions
\begin{equation}\label{eqn:1}(i,\fv,\hv,\Rm_{[i-1]},\rho_z,H_{i-1}'(\Rm_{[i-1]}\cdot\xv),\Rm_i,\Rm_i\cdot\xv)\text{ and }(i,\fv,\hv,\Rm_{[i-1]},\rho_z,H_{i-1}'(\Rm_{[i-1]}\cdot\xv),\Rm_i,b_i)\end{equation}
with advantage $\epsilon'=\epsilon/\ell-2 \cdot \negl$.  Thus, given $(\fv,\hv,\Rm_{[i-1]},\rho_z,H_{i-1}'(\Rm_{[i-1]}\cdot\xv),\Rm_i)$, $\adv'$ is able to compute $\Rm_i\cdot \xv$ with probability $\frac{1}{2}+\epsilon'$.  Note that $\epsilon'$ is non-negligible. 

Now fix $\fv,\hv,\Rm_{[i-1]}$, which fixes $H_{i-1}'$.  Let $y_i=H_{i-1}'(\Rm_{[i-1]}\cdot\xv)$.  Notice that since $\fv,\hv$ are fixed, there are at most $2^r$ possible values for $y_i$.  We now make the following claim:

\begin{claim}\label{claim:hc3} Let $\Ds$ be a sub-exponentially unpredictable distribution on $\Xs\times\mathcal{H}_Z$, with guessing probability no more than $2^{-\lambda^{c_1}}$.  Suppose $T:\Xs\rightarrow\Rs$ is drawn from a family $\Ts$ of efficient functions where the size of the image of $T$ is $2^r$.  Then the following distribution is also computationally unpredictable: $(x, (T,\rho_z,T(x)))$ where $T\gets\Ts$, $(x,\rho_z)\gets\Ds$, with guessing probability no more than $2^r \cdot 2^{-\lambda^{c_1}}$ (as long as $r \ll \lambda^{c_1}$).
\end{claim}
\begin{proof} Suppose we have an efficient adversary $\advB$ that predicts $x$ with non-negligible probability $\gamma$ given $T,\rho_z,T(x)$, and suppose $T$ has  image size $2^r$.  We then construct a new adversary $\advC$ that, given $x$, samples a random $T$, samples $(x',\rho_{z'})\gets\Ds$, and sets $a=T(x')$.  It then runs $\advB(T,\rho_z,a)$ to get a string $x''$, which it outputs.  Notice that $a$ is sampled from the same distribution as $T(x)$, so with probability at least $1/2^r$, $a=T(x)$.  In this case, $x''=x$ with probability $\gamma$.  Therefore, $\advC$ outputs $x$ with probability $\gamma/2^r$, which can not be greater than $2^{-\lambda^{c_1}}$. Thus, $\gamma$ is at most $2^r \cdot 2^{-\lambda^{c_1}}$.  
\end{proof}

\noindent Using Claim~\ref{claim:hc3} with $T=H_{i-1}'(\Rm_{[i-1]}\cdot\xv)$, we see that $(\xv,(i,\fv,\hv,\Rm_{[i-1]},\rho_z,\allowbreak H_{i-1}'(\Rm_{[i-1]}\cdot\xv)))$ is computationally unpredictable.  Moreover, $\Rm_i\cdot\xv$ is a Goldreich-Levin~\cite{goldreich1989hard} hardcore bit. We rely on the following lemma, which we will prove in the next section: 
\begin{lemma}[Quantum Goldreich-Levin] \label{lem:quantum_goldreich_levin}
If there exists a quantum algorithm, that given a random $r$ and an auxiliary quantum input $\ket {\psi_x}$, it computes $\langle x, r\rangle$ with probability at least $1/2 + \epsilon$ (where the probability is taken over the choice of $x$ and random $r$); then there exists a quantum algorithm that takes $\ket {\psi_x}$ and  extracts $x$ with probability $4 \cdot \epsilon^2$. 

The same lemma holds if the quantum auxiliary input is a mixed state, by convexity. 
\end{lemma}

Applying the quantum Goldreich-Levin theorem to the computationally unpredictable distribution $(\xv,(i,\fv,\hv,\Rm_{[i-1]},\rho_z, H_{i-1}'(\Rm_{[i-1]}\cdot\xv))\;)$, we see that there exists an algorithm that extracts $x$ with probability at least $4 \cdot \epsilon'^2$.  This contradicts the computationally unpredictability of the underlying distribution, proving \Cref{thm:LF_to_PRG}.

\end{proof}

\subsection{Quantum Goldreich-Levin, with Quantum Auxiliary Input}

In this section, we prove the final step:

\begin{replemma}{lem:quantum_goldreich_levin}
If there exists a quantum algorithm, that given random $r$ and auxiliary quantum input $\ket {\psi_x}$, it computes $\langle x, r\rangle$ with probability at least $1/2 + \epsilon$ (where the probability is taken over the choice of $x$ and random $r$); then there exists a quantum algorithm that takes $\ket {\psi_x}$ and  extracts $x$ with probability $4 \cdot \epsilon^2$. 

The same lemma holds if the quantum auxiliary input is a mixed state, by convexity. 
\end{replemma}

The proof is the same as that in \cite{adcock2002quantum}, but quantum auxiliary input about $x$ is considered. 
\begin{proof}
Assume there exists a unitary $U$, given $r$ and an auxiliary quantum state $\ket {\psi_x}$, it computes $\langle x, r\rangle$ with probability more than $1/2+\epsilon$ . Since $r$ is classical information, $U$ can be modeled as: read $r$, applies $U_{r}$. For every $x, r$, we have:
\begin{align*}
    U \ket r \ket {\psi_x} \ket {\mathbf{0}^m} = &  \ket r U_{r} \ket {\psi_x} \ket {\mathbf{0}^m} \\
    = & \ket r \left( \alpha_{x, r} \ket {\langle x, r \rangle} \ket {\phi_{x, r}} + \beta_{x, r} \ket {\overline{\langle x, r \rangle} } \ket {\phi'_{x, r}} \right) = \ket r \ket {\Phi_{x, r}},
\end{align*}
where $\ket {\mathbf{0}^m}$ is the working space, $\alpha_{x, r}$ is the coefficient for computing $\langle x, r\rangle$ correctly and $\beta_{x,r}$ for an incorrect answer. 

Let $\epsilon_x$ be the probability that the quantum algorithm answers correctly on $x$ and $R$ be the space of all $r$, we have the success probability as: 
\begin{align*}
    \mathbb{E}_r\left[ |\alpha_{x,r}|^2 \right] = \frac{1}{|R|} \sum_r |\alpha_{x,r}|^2 = 1/2 + \epsilon_x.
\end{align*}

Now we fix an $x$ and $r$. Our algorithm for extracting $x$ does the following: it starts with the state above, then (1.) it applies a $Z$-gate(phase-flip gate) to get 
\begin{align*}
    & \ket r \left( \alpha_{x, r} (-1)^{\langle x, r\rangle} \ket {\langle x, r \rangle} \ket {\phi_{x, r}} + \beta_{x, r} (-1)^{\overline{\langle x, r\rangle}} \ket {\overline{\langle x, r \rangle} } \ket {\phi'_{x, r}} \right) \\
    = &  \ket r  (-1)^{\langle x, r\rangle}  \left( \alpha_{x, r}\ket {\langle x, r \rangle} \ket {\phi_{x, r}} - \beta_{x, r} \ket {\overline{\langle x, r \rangle} } \ket {\phi'_{x, r}} \right) \\
    = & \ket r  \ket {\Phi_{x,r}'}.
\end{align*}
Then (2.) it applies $U^{\dagger}$, because we have $\langle \Phi_{x, r} \,|\, \Phi'_{x,r}\rangle = (-1)^{\langle x, r\rangle} (|\alpha_{x,r}|^2 - |\beta_{x,r}|^2)$, 
\begin{align*}
    U^{\dagger} \ket r \ket {\Phi'_{x,r}} =  \ket r \left( (-1)^{\langle x, r\rangle} (|\alpha_{x,r}|^2 - |\beta_{x,r}|^2) \ket{\psi_x} \ket {\mathbf{0}^m} + \ket {{\sf err}_{x,r}}   \right), 
\end{align*}
where $ \ket {{\sf err}_{x,r}}$ is orthogonal to $\ket {\psi_x} \ket {\mathbf{0}^m}$. 

In the first two step, we actually compute everything over a uniform superposition of $r$. 
Next (3.) it applies {\sf QFT} on $r$ register, 
\begin{align*}
    &{\sf QFT}\, \frac{1}{\sqrt{|R|}} \sum_r \ket r \left( (-1)^{\langle x, r\rangle} (|\alpha_{x,r}|^2 - |\beta_{x,r}|^2) \ket{\psi_x} \ket {\mathbf{0}^m} + \ket {{\sf err}_{x,r}}   \right) \\
    =& \frac{1}{|R|}  \sum_y \sum_r \ket y  (-1)^{\langle y, r\rangle} \left( (-1)^{\langle x, r\rangle} (|\alpha_{x,r}|^2 - |\beta_{x,r}|^2) \ket{\psi_x} \ket {\mathbf{0}^m} + \ket {{\sf err}_{x,r}}   \right). 
\end{align*}

Therefore, the phase on $\ket x \ket {\psi_{x}} \ket {\mathbf{0}^m}$ is at least, 
\begin{align*}
    \frac{1}{|R|} \sum_{r} \left( |\alpha_{x,r}|^2 - |\beta_{x,r}|^2 \right) \geq 2 \cdot \epsilon_x. 
\end{align*}
It measures $r$ register and with probability at least $4 \cdot \epsilon^2_x$, it extracts $x$. 

By convexity, the quantum algorithm succeeds in extracting $x$ is at least $4 \cdot \epsilon^2$. 
\end{proof}

\section{Proofs of Coset State Properties}
\subsection{Proof for \texorpdfstring{\Cref{thm: monogamy info}}{the Information-Theoretical Monogamy-of-Entanglement Property}}
\label{proof:monogamy}

\paragraph{Proof of \Cref{thm: monogamy info}. }

In this section, we prove Theorem \ref{thm: monogamy info}, the information-theoretic monogamy property of coset states. 
The proof resembles the proof of monogamy for BB84 states in~\cite{tomamichel2013monogamy}. However, the extra algebraic structure of subspace states requires a more refined analysis. We first state the lemmas that are required for the main theorem.

Assume $A \subseteq \mathbb{F}_2^n$ is of dimension $n/2$. 
We use $\AS(A)$ to denote the set of all cosets of $A$. Since $\dim(A) = n/2$, $|\AS(A)| = 2^{n/2}$. 
Note that if $A+s \ne A+s_0$, then they are disjoint.
Because each coset $A+s$ of $A$ has a canonical form, which is $\can_A(s)$, we will identify $\AS(A)$ with the set of all canonical vectors (where cosets are identified with their canonical vectors). 

We use $R_2^n$ to denote the set of all subspaces of dimension $n/2$ in $\mathbb{F}_2^n$.

\begin{lemma}
    Fixing a subspace $A$, the coset states $\ket {A_{s,s'}}$ and $\ket {A_{s_0, s'_0}}$ are orthogonal if and only if $A+s \ne A+s_0$ or $A'+s' \ne A'+s_0'$.
\end{lemma}

\begin{proof}
    If $A+s \ne A+s_0$,  then $|A_{s, s'}\rangle$ has support over $A + s$ but $|A_{s_0, s'_0}\rangle$ has support over $A + s_0$. Because  they have disjoint support, it is easy to see they are orthogonal.  
    
    If $A'+s' \ne A'+s'_0$, we can apply {\sf QFT} and use the same argument in the Fourier domain.
\end{proof}

\begin{lemma}
    Fixing $A$, $|A_{s,s'}\rangle$ for all $s, s' \in \AS(A), \AS(A^\perp)$ form a basis.
\end{lemma}
\begin{proof}
    We already know that the states $\ket {A_{s, s'}}$ and $\ket {A_{s_0, s'_0}}$ are orthogonal if $s, s' \ne s_0, s'_0$.  Since there are total $2^{n/2} \times 2^{n/2} = 2^n$ states $\ket {A_{s, s'}}$, they form a basis. 
\end{proof}

\begin{lemma} \label{lem:EPR}
    Fixing $A$, $\frac{1}{2^{n/2}} \sum_{s, s' \in  \AS(A), \AS(A^\perp)} |A_{s, s'}, A_{s, s'}\rangle  = \frac{1}{2^{n/2}} \sum_{v \in \mathbb{F}_2^n} |v, v\rangle$. In other words, the summation is independent of $A$ and it is an EPR pair. 
\end{lemma}
\begin{proof}
    \begin{align*}
       \sum_{s, s' \in \AS(A), \AS(A^\perp)} |A_{s, s'}, A_{s, s'}\rangle &=  \frac{1}{ |A| |A^\perp|} \sum_{s, s' \in \mathbb{F}_2^n} |A_{s, s'}, A_{s, s'}\rangle \\
         &=  \frac{1}{ |A| |A^\perp|} \sum_{s, s' \in \mathbb{F}_2^n} \frac{1}{|A|} \sum_{\substack{a \in A \\ b \in A}} (-1)^{\langle a - b, s'\rangle} |a + s\rangle |b + s\rangle \\
          &=  \frac{2^n}{|A| |A^\perp|} \sum_{s \in \mathbb{F}_2^n} \frac{1}{|A|} \sum_{\substack{a \in A }}  |a + s\rangle |a + s\rangle \\
          &= \sum_{s\in  S}  | s\rangle |s\rangle %
    \end{align*}
    where the first equality comes from the fact that for any vectors $s_0 \in A + s$ and $s_0' \in A^\perp+s'$, $\ket{A_{s, s'}} \ket{A_{s, s'}} = \ket {A_{s_0, s_0'}} \ket {A_{s_0, s_0'}}$. 
\end{proof}

We want to prove the following statement: 
\begin{theorem} \label{thm: monogam IT canonical form}
    Fix $n \in \mathbb{N}$. For any Hilbert spaces $\Hs_B, \Hs_C$, any collections of POVMs
    \begin{align*}
        \left\{ \left\{ P^A_{s, s'} \right\}_{s, s' \in \AS(A), \AS(A^\perp)} \right\}_{A \in R_2^n}  \text{ and } \left\{ \left\{ Q^A_{s, s'} \right\}_{s, s' \in \AS(A), \AS(A^\perp)} \right\}_{A \in R_2^n}
    \end{align*}
    on the Hilbert spaces, and any CPTP map that maps $|A_{s, s'}\rangle \langle A_{s, s'}|$ into $\Ds(\Hs_B) \otimes \Ds(\Hs_C)$, we have that, 
    \begin{align*}
        \mathbb{E}_{A \in R_2^n} \mathbb{E}_{s, s' \in \AS(A), \AS(A^\perp)} {\sf Tr} \left[ \left(P^A_{s,s'} \otimes Q^A_{s, s'} \right) \cdot \Phi(|A_{s,s'}\rangle\langle A_{s,s'}|)  \right] \leq 1/\subexp(n)
    \end{align*}
    where $\subexp$ is a sub-exponential function. 
\end{theorem}

Note that this bound directly gives Theorem \ref{thm: monogamy info}, since both parties in Theorem \ref{thm: monogamy info} get the description of $A$, by applying $\can_A(\cdot)$, one could map any vectors in $A+s$ and $A^\perp+s'$ to $\can_A(s)$ and $\can_{A^\perp}(s')$.

To prove Theorem \ref{thm: monogamy info} (and the above Theorem \ref{thm: monogam IT canonical form}), we present the following theorem about the monogamy game. 
\begin{theorem} \label{thm: monogam2 IT canonical form}
    Fix $n \in \mathbb{N}$. For any Hilbert spaces $\Hs_B, \Hs_C$, any collections of POVMs
    \begin{align*}
        \left\{ \left\{ P^A_{s, s'} \right\}_{s, s' \in \AS(A), \AS(A^\perp)} \right\}_{A \in R_2^n}  \text{ and } \left\{ \left\{ Q^A_{s, s'} \right\}_{s, s' \in \AS(A), \AS(A^\perp)} \right\}_{A \in R_2^n}
    \end{align*}
    on the Hilbert spaces, and any state $\rho$,  we have 
    \begin{align*}
        \mathbb{E}_{A \in R_2^n} \sum_{s, s' \in \AS(A), \AS(A^\perp)} {\sf Tr} \left[ \left(|A_{s,s'}\rangle \langle A_{s,s'}| \otimes P^A_{s,s'} \otimes Q^A_{s, s'} \right) \cdot \rho \right] \leq  1/\subexp(n)
    \end{align*}
    where $\subexp$ is a sub-exponential function.
\end{theorem}

Next, we show that to prove Theorem \ref{thm: monogamy info}, we only need to show Theorem \ref{thm: monogam2 IT canonical form}. 
\begin{lemma}
Theorem \ref{thm: monogam2 IT canonical form} implies Theorem  \ref{thm: monogam IT canonical form} (and hence Theorem \ref{thm: monogamy info}). 
\end{lemma}
\begin{proof}
    For convenience, let $S = \mathbb{F}_2^n$.
    Assume there exists a strategy for the game in Theorem \ref{thm: monogam IT canonical form} which achieves advantage $\delta$. We construct a strategy (preparing $\rho$ and POVMs) for the game in Theorem \ref{thm: monogam2 IT canonical form} which achieves the same advantage. 
    
    \begin{enumerate}
        \item Prepare the state $\rho = \frac{1}{{|S|}} (I \otimes \Phi) \sum_{s, s' \in S } |s, s\rangle \langle s', s'|$,  which is equal to the following (for any subspace $A$) by Lemma \ref{lem:EPR}, 
        \begin{align*}
            (I \otimes \Phi) \sum_{s, s' \in S} |s, s\rangle \langle s', s'| &= (I \otimes \Phi)   \sum_{\substack{s_1, s_1' \in \AS(A), \AS(A^\perp) \\ s_2, s_2'  \in \AS(A), \AS(A^\perp)}} |A_{s_1, s_1'}, A_{s_1, s_1'}\rangle \langle A_{s_2, s_2'}, A_{s_2, s_2'}| \\
            &= \sum_{\substack{s_1, s_1' \in  \AS(A), \AS(A^\perp) \\ s_2, s_2'  \in \AS(A), \AS(A^\perp) }} |A_{s_1, s_1'}\rangle \langle A_{s_2, s_2'}| \otimes \Phi\left(|A_{s_1, s_1'}\rangle \langle  A_{s_2, s_2'}|\right) \\
        \end{align*}
        \item $\overline{P}^A_{s, s'} = P^A_{s, s'}$ and $\overline{Q}^A_{s, s'} = Q^A_{s, s'}$ where $P, Q$ are POVMs for the game in Theorem  \ref{thm: monogam IT canonical form} and $\overline{P}, \overline{Q}$ are the POVMs for the game in Theorem \ref{thm: monogam2 IT canonical form}.
    \end{enumerate}
    
    Thus, we have that the advantage is, 
    \begin{align*}
        \quad & \mathbb{E}_{A \in R_2^n} \sum_{s, s'  \in  \AS(A), \AS(A^\perp)} {\sf Tr} \left[ \left(|A_{s,s'}\rangle \langle A_{s,s'}| \otimes \overline{P}^A_{s,s'} \otimes \overline{Q}^A_{s, s'} \right) \cdot \rho \right] \\
        =\, & \mathbb{E}_{A \in R_2^n} \frac{1}{|S|} \sum_{s, s'  \in  \AS(A), \AS(A^\perp)}  {\sf Tr} \left[ |A_{s,s'}\rangle \langle A_{s,s'}| \otimes \left(\left( \overline{P}^A_{s,s'} \otimes \overline{Q}^A_{s, s'} \right) \cdot  \Phi\left(|A_{s, s'}\rangle \langle  A_{s, s'}|\right)  \right) \right] \\
        =\, & \mathbb{E}_{A\in R_2^n} \frac{1}{|S|} \sum_{s, s'  \in  \AS(A), \AS(A^\perp)}  {\sf Tr} \left[ |A_{s,s'}\rangle \langle A_{s,s'}| \otimes \left(\left( {P}^A_{s, s'} \otimes {Q}^A_{s, s'} \right) \cdot  \Phi\left(|A_{s, s'}\rangle \langle  A_{s,  s'}|\right)  \right) \right] \\
        =\, &  \mathbb{E}_{A \in R_2^n} \mathbb{E}_{s, s'  \in  \AS(A), \AS(A^\perp)} {\sf Tr} \left[ \left(P^A_{s,s'} \otimes Q^A_{s, s'} \right) \cdot \Phi(|A_{s,s'}\rangle\langle A_{s,s'}|)  \right] = \delta
    \end{align*}
\end{proof}

Without loss of generality, we can assume that the adversary's strategy is pure (see more discussion in Lemma 9, \cite{tomamichel2013monogamy}). In other words, all $P^A_{s, s'}$ and $Q^A_{s, s'}$ are projections. 
\begin{proof}[Proof of Theorem  \ref{thm: monogam2 IT canonical form}]
    First, we define $\Pi^A$ as
    \begin{align*}
        \Pi^A = \sum_{s, s' \in \AS(A), \AS(A^\perp)}   |A_{s, s'}\rangle \langle A_{s,s'}| \otimes P^A_{s, s'} \otimes Q^A_{s, s'}
    \end{align*}
    Note that $\Pi^A$ is a projection. 
    By definition, the advantage is 
    \begin{align*}
        \, & \frac{1}{|R_2^n|} \sum_{A \in R_2^n} \sum_{s, s' \in \AS(A), \AS(A^\perp)} {\sf Tr} \left[|A_{s, s'}\rangle \langle A_{s,s'}| \otimes P^A_{s, s'} \otimes Q^A_{s, s'} \cdot \rho \right] \\
        \leq \, & \mathbb{E}_{v_1, \cdots, v_n} \left[ \frac{1}{\binom{n}{n/2}} \sum_{A \in {\sf span}_{n/2}(v_1, \cdots, v_n)}  \sum_{s, s' \in \AS(A), \AS(A^\perp)}   {\sf Tr} \left[|A_{s, s'}\rangle \langle A_{s,s'}| \otimes P^A_{s, s'} \otimes Q^A_{s, s'} \cdot \rho \right] \right]
    \end{align*}
    where $(v_1, \cdots, v_n)$ range over all possible bases of the space, and ${\sf span}_{n/2}(v_1, \cdots, v_n)$ is the set of all subspaces spanned by exactly $n/2$ vectors in $(v_1, \cdots, v_n)$.

    In other words, we decompose the sampling procedure of $R_2^n$ into two steps: (1) sample a random basis; (2) choose $n/2$ vectors in the basis. 
    
    Then we have, for any fixed basis $v_1, \cdots, v_n$, 
     \begin{align*}
        \quad \, &  \frac{1}{\binom{n}{n/2}}  \sum_{A \in {\sf span}_{n/2}(v_1, \cdots, v_n)}   \sum_{s, s' \in \AS(A), \AS(A^\perp)}  {\sf Tr} \left[|A_{s, s'}\rangle \langle A_{s,s'}| \otimes P^A_{s, s'} \otimes Q^A_{s, s'} \cdot \rho \right]  \\
         = \, &  \frac{1}{\binom{n}{n/2}}   {\sf Tr} \left[ \sum_{A \in {\sf span}_{n/2}(v_1, \cdots, v_n)}  \sum_{s, s' \in \AS(A), \AS(A^\perp)}  |A_{s, s'}\rangle \langle A_{s,s'}| \otimes P^A_{s, s'} \otimes Q^A_{s, s'} \cdot \rho \right]   \\
          \leq \, &  \frac{1}{\binom{n}{n/2}}    \left\vert \sum_{A \in {\sf span}_{n/2}(v_1, \cdots, v_n)} \Pi^A \right\vert
    \end{align*}
    where $|\cdot|$ is the $\infty$-Schatten norm. 
    
    \begin{lemma} \label{lem:operator_norm_bound}
        For every  fixed basis $v_1, \cdots, v_n \in \mathbb{F}_2^n$, we have 
        \begin{align*}
         \frac{1}{\binom{n}{n/2}}    \left\vert \sum_{A \in {\sf span}_{n/2}(v_1, \cdots, v_n)} \Pi^A \right\vert \leq  \frac{1}{\binom{n}{n/2}}   \sum_{t=0}^{n/2} \binom{n/2}{t}^2  2^{-t} = O\left( 2^{-\sqrt{n}} \right)
        \end{align*}
    \end{lemma}
    If we can prove the above lemma, we finished our proof for Theorem \ref{thm: monogam2 IT canonical form}. 
    \begin{proof}[Proof of Lemma \ref{lem:operator_norm_bound}]
        We first show the upper bound is sub-exponentially small. 
        By the fact that $\left(\frac{ n}{k} \right)^k \leq \binom{n}{k} \leq \left(\frac{e n}{k} \right)^k$ for all $1 \leq k \leq n$, we have:
    \begin{align*}
        \frac{1}{\binom{n}{n/2}}   \sum_{t=1}^{n/2} \binom{n/2}{t}^2  2^{-t}  \leq &  \frac{1}{\binom{n}{n/2}} \sum_{t=1}^{ \sqrt{n} } \binom{n/2}{t}^2 + 2^{- \sqrt{n}} \\
         \leq & \frac{1}{2^{n/2}} \sum_{t=1}^{  \sqrt{n} } \left( \frac{e n}{2 t} \right)^{2 t}  + 2^{- \sqrt{n} } \\
         \leq & \frac{\sqrt{n} }{2^{n/2}} \cdot \left( \frac{e n}{2} \right)^{\sqrt{n}}  + 2^{- \sqrt{n} } \\
         = & \exp(- \Omega(n - \sqrt{n} \log n)) + 2^{- \sqrt{n} }
    \end{align*}
    
        Next, we prove the remaining part of the lemma. The idea is similar to that in~\cite{tomamichel2013monogamy}. 
        
        We  require the following lemma. 
        \begin{lemma}[Lemma 2 in~\cite{tomamichel2013monogamy}] \label{lem:tfkw13_helper}
        Let $A_1$, $A_2$, $\cdots$, $A_N \in P(\mathcal{H})$ (positive
        semi-definite operators on $\mathcal{H}$), and let $\{\pi_k\}_{k \in [N]}$ be a set of $N$ mutually orthogonal permutations of $[N]$. Then, 
        \begin{align*}
            \left\vert \sum_{i \in [N]} A_i \right\vert \leq  \sum_{k \in [N]} \max_{i \in [N]}\left\vert \sqrt{A_i} \sqrt{A_{\pi_k(i)}} \right\vert.
        \end{align*}
        \end{lemma}
        A set $\{\pi_k\}_{k \in [N]}$ is called a set of mutually orthogonal permutations, if for every $\pi \ne \pi'$ in the set, $\pi(i) \ne \pi'(i)$ for all $i \in [N]$. %
        
        Fixing basis $v_1, \cdots, v_n$, there are a total of $\binom{n}{n/2}$ subspaces that can be sampled by picking a subset of $\{v_1, \cdots, v_n\}$ of size $n/2$. So, in our case, $N = \binom{n}{n/2}$.

        We define a collection of permutations on ${\sf span}_{n/2}(v_1, \cdots, v_n)$ through a graph:
        
        \begin{itemize}
           \item Recall that each subspace $A$ is described as $\{u_1, \cdots, u_{n/2}\}, \{u_{n/2+1}, \cdots, u_n\}$ where the subspace is spanned by $u_1, \cdots, u_{n/2}$.
            $\{u_{n/2+1}, \cdots, u_n\}$ are the vectors in $\{v_i\}_{i \in [n]}$ that are not in the subspace.  
            
            For convenience, we denote the basis choices for $A$ as a string $\ell \in \{0,1\}^n$ with Hamming weight $\frac{n}{2}$. We then define a set for all possible basis choices:
            \begin{align*}
                C_{n,n/2} = \left\{\ell \in \{0,1\}^n: |\ell| = \frac{n}{2} \right\}.
            \end{align*}
            
            \item Let $t \in \{0, \cdots, \frac{n}{2}\}$. Let  $G_{n,t}$ be a graph with vertex set in  $C_{n,n/2}$ and an edge between any $\ell, \ell' \in C_{n, n/2}$ where the number of positions $\ell, \ell'$ are both 1 is exactly $\frac{n}{2} - t$.
            
            \item We then turn the graph $G_{n,t}$ into a directed graph by taking each edge in $G_{n,t}$ into two directed edges. Denote $d_t$ as the least in-degree (and also least out-degree) for each vertex.  
            
            \item Then we can find $d_t$ directed cycles that cover all vertices with disjoint edges. We observe that each such directed cycle corresponds to a permutation $\pi_{t, j}$ of $C_{n, n/2}$ where $j \in [d_t], t \in [n/2]$.   
            
            \item For all $j \neq  j'$, $\pi_{t,},\pi_{t, j'}$ are orthogonal since in our construction the edges are disjoint. Moreover, for any $t \neq t'$, the permutations $\pi_{t,j},\pi_{t', j}$ are orthogonal.
            
            \item The degree $G_{n,t}$ is at least $\binom{n/2}{t}^2$:  because we have $\binom{n/2}{t}$ choices of positions where we can remove $t$  number of 1's from the current basis string $\ell$ (i.e. remove $t$ vectors from the current basis) and $\binom{n/2}{t}$ choices for us to flip $t$ number of 0-positions into 1's (i.e. add in $t$ vectors from outside the current basis set).
            
        \end{itemize}

        %
        %
        %
        
        There are a total of $\sum_{t=0}^{n/2} \binom{n/2}{t}^2 = \binom{n}{n/2} = N$ permutations $\pi_{t, j}$. Therefore, by Lemma \ref{lem:tfkw13_helper} and  $\Pi^A$ are all projections, we have 
        \begin{align*}
             \frac{1}{\binom{n}{n/2}}    \left\vert \sum_{A \in {\sf span}_{n/2}(v_1, \cdots, v_n)} \Pi^A \right\vert \leq \frac{1}{\binom{n}{n/2}}  \sum_{\pi_{t, j}} \max_{A \in {\sf span}_{n/2}(v_1, \cdots, v_n)} \left\vert \Pi^A \Pi^{\pi_{t, j}(A)} \right\vert
        \end{align*}
        Because $\Pi^A$ is a projection, $\sqrt{\Pi^A} = \Pi^A$ for all $A$. 
        
        Next, we prove the following claim: for every $A, A' \in {\sf span}_{n/2}(v_1, \cdots, v_n)$, $|\Pi^A \Pi^{A'}| \leq 2^{\dim(A \cap A') - n/2}$. 
        
        Define 
        \begin{align*}
            \overline{\Pi}^A &=\sum_{s, s' \in \AS(A), \AS(A^\perp)}   |A_{s, s'}\rangle \langle A_{s,s'}| \otimes P^A_{s, s'} \otimes I \\
            \overline{\Pi}^{A'} &= \sum_{s, s' \in \AS(A), \AS(A^\perp)}   |A_{s, s'}\rangle \langle A_{s,s'}| \otimes I \otimes Q^A_{s, s'} 
        \end{align*}

        From the fact that (1) for two semi-definite operators $A, B$ such that $A \leq B$, their $\infty$-Schatten norm satisfies $|A| \leq |B|$; (2) for a semi-definite operator $A$, $|A|^2 = |A A^\dagger|$, we have: 
        \begin{align*} 
            |\Pi^A \Pi^{A'}|^2 \leq |\overline{\Pi}^A \overline{\Pi}^{A'}|^2 = \left|\overline{\Pi}^A \overline{\Pi}^{A'} \overline{\Pi}^{A'} \overline{\Pi}^{A} \right| = \left|\overline{\Pi}^A \overline{\Pi}^{A'} \overline{\Pi}^{A} \right|
        \end{align*}
        
        Then we have, 
        \begin{align*}
            \overline{\Pi}^A \overline{\Pi}^{A'} \overline{\Pi}^{A} &=\sum_{\substack{s_1, s_1' \in \AS(A), \AS(A^\perp) \\s_2, s_2' \in \AS(A'), \AS(A'^\perp) \\ s_3, s_3' \in \AS(A), \AS(A^\perp) }} |A_{s_1, s_1'}\rangle \langle A_{s_1, s_1'}| A'_{s_2, s_2'}\rangle \langle A'_{s_2, s_2'}| A_{s_3, s_3'} \rangle \langle A_{s_3, s_3'}| \\
            & \quad\quad\quad\quad\quad\quad\quad\quad\quad\quad\quad\quad\quad\quad\quad\quad  \otimes P^A_{s_1, s_1'} P^A_{s_3, s_3'} \otimes Q^{A'}_{s_2, s_2'} \\
            &= \sum_{\substack{s_1, s_1' \in \AS(A), \AS(A^\perp) \\s_2, s_2' \in \AS(A'), \AS(A'^\perp)  }}  |\langle A_{s_1, s_1'}| A'_{s_2, s_2'}\rangle|^2 \cdot |A_{s_1, s_1'}\rangle \langle A_{s_1, s_1'}| \otimes P^A_{s_1, s_1'} \otimes Q^{A'}_{s_2, s'_2}
        \end{align*}
        
        Since for all $s_1, s_1', s_2, s'_2$, $|A_{s_1, s_1'}\rangle \langle A_{s_1, s_1'}| \otimes P^A_{s_1, s_1'} \otimes Q^{A'}_{s_2, s'_2}$ are projections, its Schatten-$\infty$ norm is bounded by the largest $|\langle A_{s_1, s_1'}| A'_{s_2, s_2'}\rangle|^2$. 
        \begin{align*}
            |\langle A_{s_1, s_1'}| A'_{s_2, s_2'}\rangle| &\leq \frac{1}{2^{n/2}} \sum_{a \in S} [a \in A+s_1 \,\wedge\, a \in A'+s_2] = 2^{\dim(A \cap A')}/2^{n/2}
        \end{align*}
        This is because, for all basis vectors outside of $A\cap A'$, their coefficient is determined by $s_1, s_2$. Therefore, the only degree of freedom comes from the basis in $A\cap A'$. 
        
        Overall, $|\Pi^A \Pi^{A'}| \leq 2^{\dim(A \cap A')}/2^{n/2}$.  Thus, 
        \begin{align*}
             \frac{1}{\binom{n}{n/2}}  \sum_{\pi_{t, j}} \max_{A \in {\sf span}_{n/2}(v_1, \cdots, v_n)} \left\vert \Pi^A \Pi^{\pi_{t, j}(A)} \right\vert
             & \leq \frac{1}{\binom{n}{n/2}}  \sum_{t=0}^{n/2} \binom{n/2}{t}^2 2^{t - n/2} \\
             & = \frac{1}{\binom{n}{n/2}}  \sum_{t=0}^{n/2} \binom{n/2}{t}^2 2^{-t}
        \end{align*}
        Thus, we proved Lemma \ref{lem:operator_norm_bound}.
    \end{proof}
    
   This completes the proof of Theorem \ref{thm: monogam2 IT canonical form}, and thus of Theorem \ref{thm: monogamy info}.
\end{proof}

\subsection{Proof of \texorpdfstring{\Cref{thm: monogamy comp}}{the Computational Monogamy-of-Entanglement Property}}
\label{sec: appendix monogamy comp}

\paragraph{Proof.}
We consider the following hybrids.
\begin{itemize}
        \item {Hyb 0:}  This is the original security game $\mathsf{CompMonogamy}$. 
        
        \item {Hyb 1:} Same as Hyb 0 except $\adv_0$ gets $\iO(\shO(A)(\cdot-s))$, $\iO(A^\perp+s')$ and $\ket{A_{s,s'}}$, for a uniformly random superspace $B$ of $A$, of dimension $3/4n$. 
        
        \item {Hyb 2:} Same as Hyb 1 except $\adv_0$ gets $\iO(\shO(B)(\cdot-s))$, $\iO(A^\perp+s')$ and $\ket{A_{s,s'}}$, for a uniformly random superspace $B$ of $A$, of dimension $3/4n$.
        
        \item {Hyb 3:} Same as Hyb 2 except for the following. The challenger samples $s,s', A$, and a uniformly random superspace $B$ of $A$ as before. The challenger sets $t = s + w_B$, where $w_B \leftarrow B$. Sends $\iO(\shO(B)(\cdot-t))$, $\iO(A^\perp+s')$ and $\ket{A_{s,s'}}$ to $\adv_0$.
        
        \item {Hyb 4:} Same as Hyb 3 except $\adv_0$ gets $\iO(\shO(B)(\cdot-t))$, $\iO(\shO(A^{\perp})(\cdot-s'))$ and $\ket{A_{s,s'}}$.
        
        \item {Hyb 5:} Same as Hyb 4 except $\adv_0$ gets $\iO(\shO(B)(\cdot-t))$, $\iO(\shO(C^{\perp})(\cdot-s'))$ and $\ket{A_{s,s'}}$, for a uniformly random subspace $C \subseteq A$ of dimension $n/4$. 
        
        \item {Hyb 6:} Same as Hyb 5 except for the following. The challenger sets $t' = s' + w_{C^{\perp}}$, where $ w_{C^{\perp}} \leftarrow C^{\perp}$. $\adv_0$ gets $\iO(\shO(B)(\cdot-t))$, $\iO(\shO(C^{\perp})(\cdot-t'))$ and $\ket{A_{s,s'}}$. 
        
        \item {Hyb 7:} Same as Hyb 6 except the challenger sends $B, C, t,t'$ in the clear to $\adv_0$.
    \end{itemize}

\begin{lemma}
\label{lem: hyb 0-1}
For any QPT adversary $(\adv_0, \adv_1, \adv_2)$,
$$\left|\Pr[(\adv_0, \adv_1, \adv_2) \text{ wins in Hyb 1}] - \Pr[(\adv_0, \adv_1, \adv_2) \text{ wins in Hyb 0}] \right| = \negl(\lambda) \,.$$
\end{lemma}

\begin{proof}
Suppose for a contradiction there was a QPT adversary $\adv$ such that:
\begin{equation}
\label{eq: difference 1}
    \left|\Pr[(\adv_0, \adv_1, \adv_2) \text{ wins in Hyb 1}] - \Pr[(\adv_0, \adv_1, \adv_2) \text{ wins in Hyb 0}] \right| 
\end{equation}
is non-negligible. 

Such an adversary can be used to construct $\adv'$ which distinguishes $\iO(A+s)$ from $\iO(\shO(A)(\cdot -s)$, which is impossible by the security of the (outer) $\mathsf{iO}$, since $A+s$ and $\iO(A)(\cdot -s)$ compute the same functionality.

Fix $n$, let $A \subseteq \mathbb{F}_2^n$, $s,s' \in \mathbb{F}_2^n$ be such that the difference in \eqref{eq: difference 1} is maximized. Suppose $\Pr[\adv \text{ wins in Hyb 1}] > \Pr[\adv \text{ wins in Hyb 0}]$, the other case being similar.

$\adv'$ proceeds as follows:
\begin{itemize}
    \item Receives as a challenge a circuit $P$ which is either $\mathsf{iO}(A+s)$ or $ \mathsf{iO}(\mathsf{shO}(A)(\cdot - s))$. Creates the state $\ket{A_{s,s'}}$. Gives $P$, $\mathsf{shO}(A^{\perp}+s')$ and $\ket{A_{s,s'}}$ as input to $\adv_0$.
    \item $\adv_0$ returns a bipartite state. $\adv'$ forwards the first register to $\adv_1$ and the second to $\adv_2$. $\adv_1$ returns $(s_1, s_1')$ and $\adv_2$ returns $(s_2, s_2')$. $\adv'$ checks If $s_1, s_2 \in A+s$ and $s_1',s_2' \in A^{\perp} +s'$. If so, $\adv'$ guesses that $P = \mathsf{iO}(\mathsf{shO}(A)(\cdot - s))$, otherwise that $P = \mathsf{iO}(A+s)$.
\end{itemize}
It is straightforward to verify that $\adv'$ succeeds at distinguishing with non-negligible probability.
\end{proof}

\begin{lemma}
\label{lem: hyb 1-2}
For any QPT adversary $(\adv_0, \adv_1, \adv_2)$,
$$\left|\Pr[(\adv_0, \adv_1, \adv_2) \text{ wins in Hyb 2}] - \Pr[(\adv_0, \adv_1, \adv_2) \text{ wins in Hyb 1}] \right| = \negl(\lambda) \,.$$
\end{lemma}

\begin{proof}
Suppose for a contradiction there was a QPT adversary $\adv$ such that:
\begin{equation*}
    \left|\Pr[(\adv_0, \adv_1, \adv_2) \text{ wins in Hyb 2}] - \Pr[(\adv_0, \adv_1, \adv_2) \text{ wins in Hyb 1}] \right| \,, 
\end{equation*}
is non-neglibile. 

We argue that $\adv$ can be used to construct an adversary $\adv'$ that breaks the security of $\shO$. 

Fix $n$. Suppose $\Pr[(\adv_0, \adv_1, \adv_2) \text{ wins in Hyb 2}] > \Pr[(\adv_0, \adv_1, \adv_2) \text{ wins in Hyb 1}]$, the other case being similar. 

$\adv'$ proceeds as follows:
\begin{itemize}
    \item Sample $A \subseteq \mathbb{F}_2^n$ uniformly at random. Send $A$ to the challenger.
    \item The challenger returns a program $P$ which is either $\shO(A)$ or $\shO(B)$ for a uniformly sampled superspace $B$. $\adv'$ samples uniformly $s,s' \in \mathbb{F}_2^n$, and creates the state $\ket{A_{s,s'}}$. Gives $\iO(P(\cdot - s))$, $\iO(A^{\perp}+s')$ and $\ket{A_{s,s'}}$ as input to $\adv_0$. The latter returns a bipartite state. $\adv'$ forwards the first register to $\adv_1$ and the second register to $\adv_2$.
    \item $\adv_1$ returns a pair $(s_1, s_1')$ and $\adv_2$ returns a pair $(s_2,s_2')$. $\adv'$ checks that $s_1, s_2 \in A+s$ and  $s_1', s_2' \in A^{\perp} + s'$. If so, then $\adv'$ guesses that $P = \shO(B)$, otherwise that $P = \shO(A)$.
\end{itemize}
It is straightforward to verify that $\adv'$ succeeds at the security game for $\shO$ with non-negligible advantage.
\end{proof}

\begin{lemma}
\label{lem: hyb 2-3}
For any QPT adversary $\adv$,
$$\left|\Pr[\adv \text{ wins in Hyb 3}] - \Pr[\adv \text{ wins in Hyb 2}] \right| = \negl(\lambda) \,.$$
\end{lemma}
\begin{proof}
The proof is similar to the proof of Lemma \ref{lem: hyb 0-1}, and follows from the security of $\iO$ and the fact that  $\shO(B)(\cdot - s)$ and $\shO(B)(\cdot - t)$ compute the same functionality. 
\end{proof}

\begin{lemma}
For any QPT adversary $\adv$,
$$\left|\Pr[\adv \text{ wins in Hyb 4}] - \Pr[\adv \text{ wins in Hyb 3}] \right| = \negl(\lambda) \,.$$
\end{lemma}
\begin{proof}
The proof is analogous to that of Lemma \ref{lem: hyb 0-1}.
\end{proof}

\begin{lemma}
For any QPT adversary $\adv$,
$$\left|\Pr[\adv \text{ wins in Hyb 5}] - \Pr[\adv \text{ wins in Hyb 4}] \right| = \negl(\lambda) \,.$$
\end{lemma}
\begin{proof}
The proof is analogous to that of Lemma \ref{lem: hyb 1-2}.
\end{proof}

\begin{lemma}
For any QPT adversary $\adv$,
$$\left|\Pr[\adv \text{ wins in Hyb 6}] - \Pr[\adv \text{ wins in Hyb 5}] \right| = \negl(\lambda) \,.$$
\end{lemma}
\begin{proof}
The proof is analogous to that of Lemma \ref{lem: hyb 2-3}.
\end{proof}

\begin{lemma}
For any QPT adversary $\adv$ for Hyb 6, there exists an adversary $\adv'$ for Hyb 7 such that
$$ \Pr[\adv' \text{ wins in Hyb 7}] \geq \Pr[\adv \text{ wins in Hyb 6}] \,.$$
\end{lemma}
\begin{proof}
This is immediate.
\end{proof}

\begin{lemma}
\label{lem:comp_vector_hiding}
For any adversary $(\adv_0, \adv_1, \adv_2)$,
$$\Pr[(\adv_0, \adv_1, \adv_2) \text{ wins in Hyb 7}] = \negl(\lambda) \,.$$
\end{lemma}
\begin{proof}
Suppose there exists an adversary $(\adv_0, \adv_1, \adv_2)$ for Hyb 7 that wins with probability $p$. 

We first show that, without loss of generality, one can take $B$ to be the subspace of vectors such that the last $n/4$ entries are zero (and the rest are free), and one can take $C$ to be such that the last $3/4n$ entries are zero (and the rest are free). We construct the following adversary $(\adv_0', \adv_1', \adv_2')$ for the game where $B$ and $C$ have the special form above with trailing zeros, call these $B_*$ and $C_*$, from an adversary $(\adv_0, \adv_1, \adv_2)$ for the game of Hyb 7.
\begin{itemize}
\item $\adv_0'$ receives a state $\ket{A_{s,s'}}$, together with $t$ and $t'$, for some $C_* \subseteq A \subseteq B_*$, where $t = s + w_{B_*}$ for $w_{B_*} \leftarrow B_*$, and $t' = s' + w_{C_*^{\perp}}$, where $w_{C_*^{\perp}} \leftarrow C_*^{\perp}$.
\item $\adv_0'$ picks uniformly random subspaces $B$ and $C$ of dimension $\frac{3}{4}n$ and $\frac{n}{4}$ respectively such that $C \subseteq B$, and a uniformly random isomorphism $\mathcal{T}$ mapping $C_*$ to $C$ and $B_*$ to $B$. We think of $\mathcal{T}$ as a change-of-basis matrix. $\adv_0'$ applies to  $\ket{A_{s,s'}}$ the unitary $U_{\mathcal{T}}$ which acts as $\mathcal{T}$ on the standard basis elements. $\adv_0'$ gives $U_{\mathcal{T}}\ket{A}$ to $\adv_0$ together with $B$, $C$, $\mathcal{T}(t)$ and $(\mathcal{T}^{-1})^T(t')$. $\adv_0'$ receives a bipartite state from $\adv_0$. Forwards the first register to $\adv_1'$ and the second register to $\adv_2'$.
\item $\adv_1'$ forwards the received register to $\adv_1$, and receives a pair $(s_1, s_1')$ as output. $\adv_1'$ returns $(\mathcal{T}^{-1}(s_1),  \mathcal{T}^{T}(s_1'))$ to the challenger. $\adv_2'$ proceeds analogously.
\end{itemize}

First, notice that
\begin{align*}
    U_{\mathcal{T}} \ket{A_{s,s'}} &= U_{\mathcal{T}}  \sum_{v \in A} (-1)^{\langle v,s' \rangle}\ket{v+s} \\
    & =  \sum_{v \in A} (-1)^{\langle v,s' \rangle}\ket{\mathcal{T}(v)+\mathcal{T}(s)} \\
    & =  \sum_{w \in \mathcal{T}(\mathcal{A})} (-1)^{\langle \mathcal{T}^{-1}(w), s' \rangle}\ket{w+\mathcal{T}(s)} \\
    & = \sum_{w \in \mathcal{T}(A)} (-1)^{\langle w, (\mathcal{T}^{-1})^T(s') \rangle}\ket{w+\mathcal{T}(s)} \\
    & = \ket{\mathcal{T}(A)_{z, z'}} \,,
\end{align*}
where $z = \mathcal{T}(s)$ and $z' = (\mathcal{T}^{-1})^T (s')$.

Notice that $\mathcal{T}(A)$ is a uniformly random subspace between $C$ and $B$, and that $z$ and $z'$ are uniformly random vectors in $\mathbb{F}_2^n$. Moreover, we argue that:
\begin{itemize}
\item[(i)] $\mathcal{T}(t)$ is distributed as a uniformly random element of $z+B$.
\item[(ii)] $(\mathcal{T}^{-1})^T(t')$ is distributed as a uniformly random element of $z' + C^{\perp}$. 
\end{itemize}

For (i), notice that 
$$\mathcal{T}(t) = \mathcal{T}(s+w_{B_*}) = \mathcal{T}(s)+\mathcal{T}(w_{B_*}) = z + \mathcal{T}(w_{B_*})\,,$$ where $w_{B_*}$ is uniformly random in $B_*$. Since $\mathcal{T}$ is an isomorphism with $\mathcal{T}(B_*) = B$, then $\mathcal{T}(w_{B_*})$ is uniformly random in $B$. Thus, $\mathcal{T}(t)$ is distributed as a uniformly random element in $z+B$.

For (ii), notice that 
$$(\mathcal{T}^{-1})^T(t') = (\mathcal{T}^{-1})^T(s'+w_{C_*^{\perp}}) = (\mathcal{T}^{-1})^T(s') + (\mathcal{T}^{-1})^T(w_{C_*^{\perp}}) = z' + (\mathcal{T}^{-1})^T(w_{C_*^{\perp}}) \,,$$
where $w_{C_*^{\perp}}$ is uniformly random in $C_*^{\perp}$. We claim that $(\mathcal{T}^{-1})^T(w_{C_*^{\perp}})$ is uniformly random in $C^{\perp}$. Notice, first, that the latter belongs to $C^{\perp}$. Let $x \in C$, then 
$$\langle (\mathcal{T}^{-1})^T(w_{C_*^{\perp}}), x \rangle = \langle w_{C_*^{\perp}}, \mathcal{T}^{-1} (x) \rangle  = 0\,,$$
where the last equality follows because $w_{C_*^{\perp}} \in C_*^{\perp}$, and $\mathcal{T}^{-1}(C) = C_*$. The claim follows from the fact that $(\mathcal{T}^{-1})^T$ is a bijection.

Hence, $\adv_0$ receives inputs from the correct distribution, and thus, with probability $p$, both $\adv_1$ and $\adv_2$ return the pair $(z = \mathcal{T}(s), z' (\mathcal{T}^{-1})^T (s'))$. Thus, with probability $p$, $\adv_1'$ and $\adv_2'$ both return $(\mathcal{T}^{-1}(z),  \mathcal{T}^{T}(z')) = (s, s')$ to the challenger, as desired.
\vspace{2mm}

So, we can now assume that $B$ is the space of vectors such that the last $\frac{n}{4}$ entries are zero, and $C$ is the space of vectors such that the last $\frac34 n$ entries are zero. Notice then that the sampled subspace $A$ is uniformly random subspace subject to the last $\frac{n}{4}$ entries being zero, and the first $\frac{n}{4}$ entries being free. From an adversary $(\adv_0, \adv_1, \adv_2)$ for Hybrid 7 with such $B$ and $C$, we will construct an adversary $(\adv_0', \adv_1', \adv_2')$ for the information-theoretic monogamy where the ambient subspace is $\mathbb{F}_2^{n'}$, where $n' = \frac{n}{2}$.

\begin{itemize}
    \item $\adv_0'$ receives $\ket{A_{s,s'}}$, for uniformly random $A \subseteq \mathbb{F}_2^{n'}$ of dimension $n'/2$ and uniformly random $s, s' \in \mathbb{F}_2^{n'}$. $\adv_0'$ samples $\tilde{s}, \tilde{s}', \hat{s}, \hat{s}' \leftarrow \mathbb{F}_2^{\frac{n}{4}}$. Let  $\ket{\phi} = \frac{1}{2^{n/8}} \sum_{x \in \{0,1\}^{n/4} } (-1)^{\langle x,\tilde{s}' \rangle}\ket{x + \tilde{s}}$. $\adv_0'$ creates the state 
$$ \ket{W} = \ket{\phi} \otimes \ket{A_{s,s'}} \otimes \ket{\hat{s}}\,,$$
$\adv_0'$ gives to $\adv_0$ as input the state $\ket{W}$, together with $t = 0^{3n/4}|| \hat{s} + w_B$ for $w_B \leftarrow B$ and $t' = \hat{s}'||0^{3n/4} + w_{C^{\perp}}$, for $w_{C^{\perp}} \leftarrow C^{\perp}$. $\adv_0$ returns a bipartite state. $\adv_0'$ forwards the first register to $\adv_1'$ and the second register to $\adv_2'$. 
\item $\adv_1'$ receives $A$ from the challenger. $\adv_1'$ sends to $\adv_1$ the previously received register, together with the the subspace $A' \subseteq \mathbb{F}_2^n$ whose first $n/4$ entries are free, the last $n/4$ entries are zero, and the middle $n/2$ entries belong to $A$. $\adv_1$ returns a pair $(s_1, s_1') \in \mathbb{F}_2^{n} \times \mathbb{F}_2^{n}$. Let $r_1 = [s_1]_{\frac{n}{4}+1, \frac{3}{4}n} \in \mathbb{F}_2^{n/2}$ be the ``middle'' $n/2$ entries of $s_1$. Let $r_1 =  [s'_1]_{\frac{n}{4}+1, \frac{3}{4}n} \in \mathbb{F}_2^{n/2} $. $\adv_1'$ outputs $(r_1, r'_1)$. 
\item $\adv_2'$ receives $A$ from the challenger. $\adv_2'$ sends to $\adv_2$ the previously received register, together with the the subspace $A' \subseteq \mathbb{F}_2^n$ whose first $n/4$ entries are free, the last $n/4$ entries are zero, and the middle $n/2$ entries belong to $A$. $\adv_2$ returns a pair $(s_2, s_2') \in \mathbb{F}_2^{n} \times \mathbb{F}_2^{n}$. Let $r_2 = [s_2]_{\frac{n}{4}+1, \frac{3}{4}n} \in \mathbb{F}_2^{n/2}$ be the ``middle'' $n/2$ entries of $s_2$. Let $r'_2 =  [s'_2]_{\frac{n}{4}+1, \frac{3}{4}n} \in \mathbb{F}_2^{n/2} $. $\adv_2'$ outputs $(r_2, r_2')$.
\end{itemize}

Notice that 
\begin{align*}
\ket{W} &=  \ket{\phi} \otimes \ket{A_{s,s'}} \otimes \ket{\tilde{s}} \\
&= \sum_{x \in \{0,1\}^{n/4}, v \in A} (-1)^{\langle x,\tilde{s}' \rangle} (-1)^{\langle v,s' \rangle}\Big| (x+\tilde{s})||(v+s)||\hat{s}\Big\rangle \\
&= \sum_{x \in \{0,1\}^{n/4}, v \in A} (-1)^{\langle (x||v||0^{n/4}), (\tilde{s}'|| s'||\hat{s}') \rangle} \Big| x||v||0^{n/4} + \tilde{s}||s||\hat{s} \Big \rangle \\
&= \sum_{w \in \tilde{A}} (-1)^{\langle w, z' \rangle} \ket{w + z} = \ket{\tilde{A}_{z,z'}} \,,
\end{align*}
where $z = \tilde{s}||s||\hat{s}$, $z' = \tilde{s}'||s'||\hat{s}'$, and $\tilde{A} \subseteq \mathbb{F}_2^n$ is the subspace in which the first $n/4$ entries are free, the middle $n/2$ entries belong to subspace $A$, and the last $n/4$ entries are zero.

Notice that the subspace $\tilde{A}$, when averaging over the choice of $A$, is distributed precisely as in the game of Hybrid 7 (with the special choice of $B$ and $C$); $z, z'$ are uniformly random in $\mathbb{F}_2^{n}$; $t$ is uniformly random from $z + B$, and $t'$ is uniformly random from $z'+C^{\perp}$. It follows that, with probability $p$, the answers returned by $\adv'_1$ and $\adv'_2$ are both correct. 

From the information-theoretic security of the monogamy game, Theorem \ref{thm: monogamy info}, it follows that $p$ must be negligible. 
\end{proof}

\section{More Discussions On Anti-Piracy Security}
\subsection{Anti-Piracy Implies CPA Security}

\label{sec:unclonable_dec_cpa_ag_implies_cpa}

\begin{lemma}
    If a single-decryptor encryption scheme satisfies CPA-style $\gamma$-anti-piracy security (\Cref{def:weak_ag}) for all inverse poly $\gamma$, it also satisfies CPA security. 
\end{lemma}
\begin{proof}
    Let $\As$ be an adversary that breaks CPA security with advantage $\delta$. We construct the following adversary $\Bs$ that breaks its CPA-style $(\delta/2)$-anti-piracy security. 
    
     $\Bs$ upon receiving a public key $\pk$ and a quantum key $\rho_{\sk}$, it prepares the following programs: 
            \begin{itemize}
                \item It runs the stateful adversary $\As$ on $(1^\lambda, \pk)$, it outputs $(m_0, m_1)$.  
                \item Let $(\sigma[R_1], U_{1})$ be the stateful algorithm $\As$ in the CPA security game (after outputting $(m_0, m_1)$), except when it outputs a bit $b$, it outputs $m_b$; $(\sigma[R_2], U_{2})$ be the honest decryption algorithm using $\rho_{\sk}$;  $\aux = (m_0, m_1)$ be the output of $\As$.  
            \end{itemize}
        
    First, we observe that $\sigma[R_1]$ and $\sigma[R_2]$ are un-entangled. For $(\sigma[R_1], U_{1})$, because $\As$ wins CPA games with advantage $\delta$, here it also outputs the correct message with probability $1/2 + \delta$. For $(\sigma[R_2], U_{2})$, by the correctness of the scheme, it outputs the correct message with probability $1 - \negl(\lambda)$. Overall, $\Bs$ wins the game with probability $1/2 + \delta - \negl(\lambda) \gg 1/2 + \delta/2$. 
\end{proof}

\subsection{Strong Anti-Piracy Implies Regular Definition}
\label{sec:unclonable dec strong_implies_weak}

In this section, we show the notion of strong anti-piracy security from Definition \ref{def:strong_ag} implies that from Definition \ref{def:weak_ag}.
\begin{proof}
    Assume a single-decryptor encryption scheme satisfies the strong notation of anti-piracy. For any adversary $\As$, consider the game $\strongantipiracy$: %
    \begin{itemize}
    \item  At the beginning of the game, the challenger takes a security parameter $\lambda$ and obtains keys $(\sk, \pk) \gets \setup(1^\lambda)$.

    \item The challenger sends $\cA$   public-key $\pk$ and one copy of decryption key $\qsk$ corresponding to $\pk$.  
     
    \item $\As$ finally outputs two (entangled) quantum decryptors $\D_1 = (\sigma[R_1], U_{1})$ and $\D_2 = (\sigma[R_2],\allowbreak U_{2})$  and  $\aux = (m_0, m_1)$ ($m_0 \ne m_1$)
    
    \item The challenger outputs 1 (for $\cA$ winning) if and only if \emph{both} quantum decryptors $\D_1, \D_2$ are tested to be $\gamma$-good with respect to $\pk$ and $\aux$. 
    \end{itemize}
    
    It (the challenger outputs $1$) happens with only negligible probability.
    In other words, with overwhelming probability over the distribution of $(\sk, \pk)$ and $(m_0, m_1)$, by applying projective measurement (the projective measurement $\cE_1, \cE_2$ inside both threshold implementations) and obtaining $(d_1, d_0)$ and $(d'_1, d'_0)$ for $\D_1, \D_2$ respectively, at least one of $d_1, d'_1$  is smaller than $\frac{1}{2} + \gamma$, by \Cref{def:gamma_good_decryptor} (the definition of $\gamma$-good decryptor).

    Also note that, in \Cref{def:weak_ag}, the game only differs in the test phase, 
    \begin{itemize}
    \item  The first three steps are identical to those in the above game. 
    
    \item The challenger samples $b_1, b_2$ and $r_1, r_2$ uniformly at random and generates ciphertexts $c_1 = \enc(\pk, m_{b_1}; r_1)$ and $c_2 = \enc(\pk, m_{b_2}; r_2)$. The challenger runs $\sf D_1$ on $c_1$ and $\sf D_2$ on $c_2$ and it outputs $1$ (the game is won by the adversary) if and only if $\sf D_1$ outputs $m_{b_1}$ and $\sf D_2$ outputs $m_{b_2}$. 
    \end{itemize}
    By the definition of projective measurement (\Cref{def:project_implement}), the distribution of the second game, can be computed by its projective measurement. 
    In other words, the test phase can be computed in the following equivalent way:
    \begin{itemize}
        \item   Apply the projective measurement $\cE_1, \cE_2$ and obtain $(d_1, d_0)$ and $(d'_1, d'_0)$ for $\D_1, \D_2$ respectively. The challenger then samples two bits $b_1, b_2$ independently, where $b_1 = 1$ with probability $d_1$ and $b_2 = 1$ with probability $d'_1$.  It outputs $1$ if and only if $b_1 = b_2 = 1$. 
    \end{itemize}
    Since we know that with overwhelming probability over the distribution of $(\sk, \pk)$ and $(m_0, m_1)$, by applying projective measurement and obtaining $(d_1, d_0)$ and $(d'_1, d'_0)$ for $\D_1, \D_2$ respecitvely, at least one of $d_1, d'_1$  is smaller than $\frac{1}{2} + \gamma$, we can bound the probability of succeeding in the second game. 
    \begin{align*}
        \Pr[\text{$\As$ succeeds}] &\leq 1 \cdot \Pr\left[d_1 \geq \frac{1}{2} + \gamma \,\wedge\, d'_1 \geq \frac{1}{2} + \gamma \right]  \\
        & \quad\quad\quad\quad + \left(\frac{1}{2} + \gamma\right) \cdot \Pr\left[d_1 \leq \frac{1}{2} + \gamma \,\vee\, d'_1 \leq \frac{1}{2} + \gamma \right] \\
         &\leq \negl(\lambda) + \left(\frac{1}{2} + \gamma\right)
    \end{align*}
    Therefore, it also satisfies the weak definition (\Cref{def:weak_ag}).
\end{proof}

\subsection{Strong Anti-Piracy, with Random Challenge Plaintexts}

\label{sec:strong_anti_piracy_random}

\begin{definition}[Testing a quantum decryptor, with random challenge plaintexts] 
\label{def:gamma_good_decryptor_random}
   Let $\gamma \in [0,1]$. Let $\pk$ be a public key. We refer to the following procedure as a {$\gamma$-good test for a quantum decryptor} with respect to $\pk$ and random challenge plaintexts:
   \begin{itemize}
       \item The procedure takes as input a quantum decryptor $(\rho, U)$.
       \item Let $\mathcal{P} = (P, I - P)$ be the following mixture of projective measurements (in the sense of Definition \ref{def:mixture_of_projective}) acting on some quantum state $\rho'$:
       \begin{itemize}
       \item Sample a uniform random message $m \leftarrow \cM$. Compute $c \leftarrow \enc(\pk, m)$.
       \item Run the quantum decryptor $(\rho',U)$ on input $c$. Check whether the outcome is $m$. If so, output $1$, otherwise output $0$.
      \end{itemize}
       \item Let $\ti_{1/|\cM| + \gamma}(\cP)$ be the threshold implementation of $\cP$ with threshold value $\frac{1}{|\cM|} + \gamma$, as defined in \Cref{def:thres_implement}. Run $\ti_{1/|\cM| + \gamma}(\cP)$ on $(\rho, U)$, and output the outcome. If the output is $1$, we say that the test passed, otherwise the test failed.
   \end{itemize}
\end{definition}

Now we are ready to define the strong $\gamma$-anti-piracy game. 

\begin{definition}[(Strong) $\gamma$-Anti-Piracy Game, with Random Challenge Plaintexts]
\label{def:gamma_anti_piracy_game_random}
 A strong anti-piracy security game (for random plaintexts) for adversary $\cA$ is denoted as $\sf{StrongAntiPiracyGuess}(1^\lambda)$, which consists of the following steps: %
\begin{enumerate}
    \item \textbf{Setup Phase}: At the beginning of the game, the challenger takes a security parameter $\lambda$ and obtains keys $(\sk, \pk) \gets \setup(1^\lambda)$.

    \item \textbf{Quantum Key Generation Phase}:
    The challenger sends $\cA$  the public-key $\pk$ and one copy of decryption key $\qsk$. %

    \item \textbf{Output Phase}: Finally, $\As$ outputs a (possibly mixed and entangled) state $\sigma$ over two registers $R_1, R_2$ and two quantum circuits $(U_1, U_{2})$. They can be viewed as two quantum decryptors $\D_1 = (\sigma[R_1], U_{1})$ and $\D_2 = (\sigma[R_2],\allowbreak U_{2})$.
    
    \item \textbf{Challenge Phase}: The challenger outputs 1 (for $\cA$'s winning) if and only if \emph{both} quantum decryptors $\D_1, \D_2$ are tested to be $\gamma$-good with respect to $\pk$ and random challenge plaintexts. 
\end{enumerate}
\end{definition}

\begin{definition}[(Strong) $\gamma$-Anti-Piracy-Security] \label{def:strong_ag_random}
 A single-decryptor encryption scheme satisfies strong $\gamma$-anti-piracy security against random plaintexts, if for any QPT adversary $\cA$,  there exists a negligible function $\negl(\cdot)$ such that the following holds for all $\lambda \in \N$: 
  \begin{align}
    \Pr\left[b = 1, b \gets \sf{StrongAntiPiracyGuess}(1^\lambda) \right]\leq \negl(\lambda)
    \end{align}
\end{definition}

We claim \Cref{def:strong_ag_random} implies \Cref{def:weak_ag_random}. %
The proof is done in the same way as that in \Cref{sec:unclonable dec strong_implies_weak}. We omit the proof here. 

To prove both our constructions satisfy the strong anti-piracy security against random messages:
\begin{itemize}
    \item Construction based on strong monogamy property: the proof works in the exactly same way, except the compute-and-compare program $\CC[f, y, m_b]$ for a uniform bit $b$ should be replaced with $\CC[f, y, m]$ for a uniformly random message $m$.

    \item Construction based on extractable witness encryption: the proof works in the exact same way. 
\end{itemize}

\subsection{Comparing \texorpdfstring{\Cref{def:weak_ag}}{Regular Anti-Piracy Security} with \texorpdfstring{\Cref{def:weak_ag_random}}{Anti-Piracy Security with Random Ciphertexts}}
\label{sec:unclonable_dec_unified}

In \Cref{sec:unclonable_dec}, we define two anti-piracy security, namely \Cref{def:weak_ag} for chosen plaintexts and \Cref{def:weak_ag_random} for random plaintexts.  In this section, we discuss their relationship.

One would hope that anti-piracy security against chosen plaintexts (\Cref{def:weak_ag}) implies anti-piracy security against random plaintexts  (\Cref{def:weak_ag_random}), which is an analogue of security against chosen plaintext attack implies security against random plaintext attack (decrypting encryptions of random messages). However, we realize that it is unlikely to be the case for anti-piracy security. Although it is not a formal proof, this intuition explains where thinks might fail.

Consider an adversary that breaks \Cref{def:weak_ag_random}. Assume it outputs the following decryptor state: 
\begin{align*}
    \left(\sqrt{\gamma} \ket {\sf good} + \sqrt{1 - \gamma} \ket {\sf bad}\right)^{\otimes 2}, 
\end{align*}
where $\ket{\sf good}$ is a perfect decryptor state and $\ket{\sf bad}$ is a garbage state that is orthogonal to $\ket {\sf good}$. It is easy to see that it breaks \Cref{def:weak_ag_random} with advantage $\gamma^2$. However, each side can only win the CPA security game with advantage at most $1/2 + \gamma$ independently. Therefore, its advantage for \Cref{def:weak_ag} is $(1/2 + \gamma)^2$, which is smaller than the trivial advantage $1/2$.

\section{Proof \texorpdfstring{of \Cref{lem:hidden_trigger}}{Using Hidden Trigger Techniques}}
\label{sec: PRF hidden trigger}

We are going to show that an adversary can not distinguish a pair of  uniformly inputs from a pair of hidden trigger inputs by a sequence of hybrids. For simplicity, we first show the following lemma about the indistinguishability of a single random input or a single hidden trigger input.  We will then show how the proof for \Cref{lem:hidden_trigger_1d} translates to a proof for \Cref{lem:hidden_trigger} easily. 

Note that one can not get \Cref{lem:hidden_trigger} by simply applying \Cref{lem:hidden_trigger_1d} twice, as
one can not sample a random hidden trigger input by only given the public information in the security game ($\gentrigger$ requires knowing $K_2, K_3$), which is essentially required. 

\begin{lemma}\label{lem:hidden_trigger_1d}
    Assuming post-quantum $\iO$ and one-way functions, for every efficient QPT algorithm $\As$, it can not distinguish the following two cases with non-negligible advantage:
    \begin{itemize}
        \item  A challenger samples $K_1 \gets \setup(1^\lambda)$ and prepares a quantum key $\rho_K = (\{\ket{A_{i, s_i, s'_i}} \}_{i \in [\ell_0]}, \iO(P))$. Here $P$ hardcodes $K_1, K_2, K_3$. 

        \item It samples a random input $u \gets [N]$. Let $y = F_1(K_1, u)$. Parse the input as $u = u_0||u_1||u_2$.
        
        \item Let $u' \gets \gentrigger(u_0, y, K_2, K_3, \{A_i, s_i, s'_i\}_{i \in [\ell_0]})$.

        \item It flips a coin $b$ and outputs $(\rho_K, u)$ or $(\rho_K, u')$ depending on the coin.
    \end{itemize}
\end{lemma}

Note that we will mark the changes between the current hybrid and the previous hybrid {\color{red} in red}. 

\paragraph{Proof of \Cref{lem:hidden_trigger_1d}}.

\paragraph{Hybrid 0.} This is the original game where the input is sampled either uniformly at random or sampled as a hidden triggers input.  

\begin{enumerate}
    \item It samples random subspaces $A_i$ of dimension $\lambda/2$ and vectors $s_i, s'_i$ for $i = 1, 2, \cdots, \ell_0$. 
    It then prepares programs $R^0_i = \iO(A_i+s_i)$ and $R^1_i = \iO(A^\perp_i + s'_i)$ (padded to the length upper bound $\ell_2 - \ell_0$). It prepares the quantum state $\ket \psi = \bigotimes_{i} \ket {A_{i, s_i, s'_i}}$. 
    
    \item It then samples keys $K_1, K_2, K_3$ for $F_1, F_2, F_3$.
    
    \item It samples $u = u_0||u_1||u_2$ uniformly at random. Let $y = F_1(K_1, u)$. 
    
    \item It samples $u' \gets \gentrigger(u_0, y, K_2, K_3, \{A_i, s_i, s'_i\}_{i \in [\ell_0]})$.
    
    \item Generate the program $P$ as in Figure \ref{fig:program_pprime}. The adversary is given $(\ket \psi, \iO(P))$ and then $u$ or $u'$ depending on a random coin $b$.
    
\end{enumerate}

\begin{figure}[hpt]
\centering
\begin{gamespec}
\textbf{Hardcoded:} Keys $K_1, K_2, K_3$, $R^0_i, R^1_i$ for all $i \in [\ell_0]$. 

On input $x = x_0 || x_1 || x_2$ and vectors $v_1, \cdots, v_{\ell_0}$:  
\begin{enumerate}
\item If $F_3(K_3, x_1) \oplus x_2 = x_0' || Q'$  and $x_0 = x'_0$ and $x_1 = F_2(K_2, x'_0 || Q')$: 

    \quad It treats $Q'$ as a circuit and outputs $Q'(v_1, \cdots, v_{\ell_0})$. 
\item Otherwise, it checks if the following holds: for all $i \in [\ell_0]$, $R^{x_{0, i}}(v_i) = 1$. 

    \quad If they all hold, outputs $F_1(K_1, x)$. Otherwise, outputs $\bot$. 
\end{enumerate}
\end{gamespec}
\caption{Program $P$ (same as Figure \ref{fig:program_p})}
\label{fig:program_pprime}
\end{figure}

\newpage
\paragraph{Hybrid 1}
 In this hybrid, the key $K_1$ in the program $P$ is punctured at $u, u'$. The indistinguishability of Hybrid 0 and Hybrid 1 comes from the security of indistinguishability obfuscation.  

\begin{enumerate}
    \item It samples random subspaces $A_i$ of dimension $\lambda/2$ and vectors $s_i, s'_i$ for $i = 1, 2, \cdots, \ell_0$. 
    It then prepares programs $R^0_i = \iO(A_i+s_i)$ and $R^1_i = \iO(A^\perp_i + s'_i)$ (padded to length $\ell_2 - \ell_0$). It prepares the quantum state $\ket \psi = \bigotimes_{i} \ket {A_{i, s_i, s'_i}}$. 
    
    \item It then samples keys $K_1, K_2, K_3$ for $F_1, F_2, F_3$.
    
    \item It samples $u = u_0||u_1||u_2$ uniformly at random. Let $y = F_1(K_1, u)$.
    
    \item It samples $u' \gets \gentrigger(u_0, y, K_2, K_3, \{A_i, s_i, s'_i\}_{i \in [\ell_0]})$. 
    {\color{red} Let $Q$ be the obfuscation program during the execution of $\gentrigger$}.

    \item Generate the program $P$ as in Figure \ref{fig:program_P_1}. The adversary is given $(\ket \psi, \iO(P))$ and then $u$ or $u'$ depending on a random coin.
    
\end{enumerate}

\begin{figure}[hpt]
\centering
\begin{gamespec}
\textbf{Hardcoded:} {\color{red}Constants $u, u'$}; Keys ${\color{red}K_1 \setminus \{u, u'\}}, K_2, K_3$, $R^0_i, R^1_i$ for all $i \in [\ell_0]$. 

On input $x = x_0 || x_1 || x_2$ and vectors $v_1, \cdots, v_{\ell_0}$:  
\begin{enumerate}

\item {\color{red}If $x = u$ or $u'$, it outputs $Q(v_1, \cdots, v_{\ell_0})$}. %

\item If $F_3(K_3, x_1) \oplus x_2 = x_0' || Q'$  and $x_0 = x'_0$ and $x_1 = F_2(K_2, x'_0 || Q')$: 

    \quad It treats $Q'$ as a circuit and outputs $Q'(v_1, \cdots, v_{\ell_0})$. 
\item Otherwise, it checks if the following holds: for all $i \in [\ell_0]$, $R^{x_{0, i}}(v_i) = 1$.

    \quad If they all hold, outputs $F_1(K_1, x)$. Otherwise, outputs $\bot$. 
\end{enumerate}
\end{gamespec}
\caption{Program $P$}
\label{fig:program_P_1}
\end{figure}

Note that starting from this hybrid, whenever we mention $K_1$ \emph{inside a program $P$}, we mean to use the punctured key $K_1 \setminus \{u,u'\}$. Similar notations of punctured keys $K_2, K_3$ inside other programs will appear in the upcoming hybrids.

\newpage
 \paragraph{Hybrid 2.} In this hybrid,  the value of $F_1(K_1, u)$ is replaced with a uniformly random output. The indistinguishability of Hybrid 1 and Hybrid 2 comes from the pseudorandomness at punctured points of a puncturable PRF.  
\begin{enumerate}
    \item It samples random subspaces $A_i$ of dimension $\lambda/2$ and vectors $s_i, s'_i$ for $i = 1, 2, \cdots, \ell_0$. 
    It then prepares programs $R^0_i = \iO(A_i+s_i)$ and $R^1_i = \iO(A^\perp_i + s'_i)$ (padded to length $\ell_2 - \ell_0$). It prepares the quantum state $\ket \psi = \bigotimes_{i} \ket {A_{i, s_i, s'_i}}$. 
    
    \item It then samples keys $K_1, K_2, K_3$ for $F_1, F_2, F_3$.
    
    \item It samples $u = u_0||u_1||u_2$ uniformly at random. Let {\color{red} $y \gets [M]$}.
    
    \item It samples $u' \gets \gentrigger(u_0, y, K_2, K_3, \{A_i, s_i, s'_i\}_{i \in [\ell_0]})$. 
    Let $Q$ be the obfuscation program during the execution of $\gentrigger$.

    \item Generate the program $P$ as in Figure \ref{fig:program_P_1}. The adversary is given $(\ket \psi, \iO(P))$ and then $u$ or $u'$ depending on a random coin.
\end{enumerate}

\newpage
\paragraph{Hybrid 3.} In this hybrid, the check on the second line  will be skipped if $x_1$ is equal to $u_1$ or $u'_1$. By Lemma 2 of \cite{sahai2014use}, adding this check does not affect its functionality, except with negligible probability. 

The lemma says, to skip the check on the second line, $x_1$ will be equal to one of $\{u_1, u'_1\}$. To see why it does not change the functionality of the program, by \Cref{lem:sw14_lem1} and for all but negligible fraction of all keys $K_2$, if $x_1 = u'_1$, there is only one way to make the check satisfied and the input is $u_0, u'_2$. This input $u' = u_0||u'_1||u'_2$ is already handled in the first line. Therefore, the functionality does not change. 

After this change, $F_3(K_3, \cdot)$ will never be executed on those inputs. We can then puncture the key $K_3$ on them. The indistinguishability comes from the security of \iO.

\begin{enumerate}
    \item It samples random subspaces $A_i$ of dimension $\lambda/2$ and vectors $s_i, s'_i$ for $i = 1, 2, \cdots, \ell_0$. 
    It then prepares programs $R^0_i = \iO(A_i+s_i)$ and $R^1_i = \iO(A^\perp_i + s'_i)$ (padded to length $\ell_2 - \ell_0$). It prepares the quantum state $\ket \psi = \bigotimes_{i} \ket {A_{i, s_i, s'_i}}$. 
    
    \item It then samples keys $K_1, K_2, K_3$ for $F_1, F_2, F_3$.
    
    \item It samples $u = u_0||u_1||u_2$ uniformly at random. Let  $y \gets [M]$.
    
    \item It samples $u' \gets \gentrigger(u_0, y, K_2, K_3, \{A_i, s_i, s'_i\}_{i \in [\ell_0]})$.
    Let $Q$ be the obfuscation program during the execution of $\gentrigger$.

    \item Generate the program $P$ as in Figure \ref{fig:program_P_2}. The adversary is given $(\ket \psi, \iO(P))$ and then $u$ or $u'$ depending on a random coin.
\end{enumerate}

\begin{figure}[hpt]
\centering
\begin{gamespec}
\textbf{Hardcoded:} {Constants $u, u'$}; Keys ${K_1 \setminus \{u, u'\}}, K_2, {\color{red}K_3\setminus \{u_1, u'_1\}}$, $R^0_i, R^1_i$ for all $i \in [\ell_0]$. 

On input $x = x_0 || x_1 || x_2$ and vectors $v_1, \cdots, v_{\ell_0}$:  
\begin{enumerate}

\item If $x = u$ or $u'$, it outputs $Q(v_1 \cdots, v_{\ell_0})$. %

\item {\color{red}If $x_1 = u_1$ or $u'_1$, skip this check.} If $F_3(K_3, x_1) \oplus x_2 = x_0' || Q'$  and $x_0 = x'_0$ and $x_1 = F_2(K_2, x'_0 || Q')$: 

    \quad It treats $Q'$ as a circuit and outputs $Q'(v_1, \cdots, v_{\ell_0})$. 
\item Otherwise, it checks if the following holds: for all $i \in [\ell_0]$, $R^{x_{0, i}}(v_i) = 1$. 

    \quad If they all hold, outputs $F_1(K_1, x)$. Otherwise, outputs $\bot$. 
\end{enumerate}
\end{gamespec}
\caption{Program $P$}
\label{fig:program_P_2}
\end{figure}

\newpage
\paragraph{Hybrid 4.} In this hybrid, before checking $x_1 = F_2(K_2, x'_0 || Q')$, it checks if $x'_0||Q' \ne u_0||Q$. Because if $x'_0||Q' = u_0||Q$ and the last check $x_1 = F_2(K_2, x_0'||Q')$ is also satisfied, we know that 
\begin{align*}
x_1 = F_2(K_2, x'_0||Q') = F_2(K_2, u_0||Q) = u_1' \,\,\text{ (by the definition of $\gentrigger$).}     
\end{align*}
Therefore the step 2 will be skipped (by the first check). 
Thus, we can puncture $K_2$ at $u_0||Q$
The indistinguishability also comes from the security of \iO. %

\begin{enumerate}
    \item It samples random subspaces $A_i$ of dimension $\lambda/2$ and vectors $s_i, s'_i$ for $i = 1, 2, \cdots, \ell_0$. 
    It then prepares programs $R^0_i = \iO(A_i+s_i)$ and $R^1_i = \iO(A^\perp_i + s'_i)$ (padded to length $\ell_2 - \ell_0$). It prepares the quantum state $\ket \psi = \bigotimes_{i} \ket {A_{i, s_i, s'_i}}$. 
    
    \item It then samples keys $K_1, K_2, K_3$ for $F_1, F_2, F_3$.
    
    \item It samples $u = u_0||u_1||u_2$ uniformly at random. Let  $y \gets [M]$.
    
    \item It samples $u' \gets \gentrigger(u_0, y, K_2, K_3, \{A_i, s_i, s'_i\}_{i \in [\ell_0]})$.
    Let $Q$ be the obfuscation program during the execution of $\gentrigger$.

    \item Generate the program $P$ as in Figure \ref{fig:program_P_3}. The adversary is given $(\ket \psi, \iO(P))$ and then $u$ or $u'$ depending on a random coin.
\end{enumerate}

\begin{figure}[hpt]
\centering
\begin{gamespec}
\textbf{Hardcoded:} {Constants $u, u'$}; Keys ${K_1 \setminus \{u, u'\}}, {\color{red} K_2 \setminus \{u_0||Q\}}$, ${K_3\setminus \{u_1, u'_1\}}$, $R^0_i, R^1_i$ for all $i \in [\ell_0]$. 

On input $x = x_0 || x_1 || x_2$ and vectors $v_1, \cdots, v_{\ell_0}$:  
\begin{enumerate}

\item If $x = u$ or $u'$, it outputs $Q(v_1 \cdots)$.

\item {If $x_1 = u_1$ or $u'_1$, skip this check.} If $F_3(K_3, x_1) \oplus x_2 = x_0' || Q'$  and $x_0 = x'_0$ {and \color{red}$x'_0||Q' \ne u_0||Q$}, then also check $x_1 = F_2(K_2, x'_0 || Q')$:

    \quad It treats $Q'$ as a circuit and outputs $Q'(v_1, \cdots, v_{\ell_0})$. 
\item Otherwise, it checks if the following holds: for all $i \in [\ell_0]$, $R^{x_{0, i}}(v_i) = 1$.

    \quad If they all hold, outputs $F_1(K_1, x)$. Otherwise, outputs $\bot$. 
\end{enumerate}
\end{gamespec}
\caption{Program $P$}
\label{fig:program_P_3}
\end{figure}

\newpage
\paragraph{Hybrid 5.} In this hybrid, since the key $K_2$ has been punctured at $u_0||Q$, we can replace the evaluation of $F_2(K_2, \cdot)$ at the input with a uniformly random value. The indistinguishability comes from the pseudorandomness of the underlying puncturable PRF $F_2$. 

We expand the $\gentrigger$ procedure. 

\begin{enumerate}
    \item It samples random subspaces $A_i$ of dimension $\lambda/2$ and vectors $s_i, s'_i$ for $i = 1, 2, \cdots, \ell_0$. 
    It then prepares programs $R^0_i = \iO(A_i+s_i)$ and $R^1_i = \iO(A^\perp_i + s'_i)$ (padded to length $\ell_2 - \ell_0$). It prepares the quantum state $\ket \psi = \bigotimes_{i} \ket {A_{i, s_i, s'_i}}$. 
    
    \item It then samples keys $K_1, K_2, K_3$ for $F_1, F_2, F_3$.
    
    \item It samples $u = u_0||u_1||u_2$ uniformly at random. Let  $y \gets [M]$.
    
    \item It samples $u' \gets \gentrigger(u_0, y, K_2, K_3, \{A_i, s_i, s'_i\}_{i \in [\ell_0]})$ as follows: 
        \begin{enumerate}
        \item Let $Q$ be the obfuscation of the program (padded to length $\ell_2 - \ell_0$) that takes inputs $v_1, \cdots, v_{\ell_0}$ and outputs $y$ if and only if for every input $v_i$, if $u_{0, i} = 0$, then $v_i$ is in $A_i + s_i$ and otherwise it is in $A^\perp_i + s'_i$. 
        \item {\color{red} $u'_1 \gets [2^{\ell_1}]$} (since $F_2(K_2, u_0||Q)$ has been replaced with a uniformly random value).
        \item $u'_2 \gets F_3(K_3, u'_1) \oplus (u_0 ||Q)$.
        \item It outputs $u' = u_0||u'_1||u'_2$. 
    \end{enumerate}

    \item Generate the program $P$ as in Figure \ref{fig:program_P_3}. The adversary is given $(\ket \psi, \iO(P))$ and then $u$ or $u'$ depending on a random coin.
\end{enumerate}

\newpage
\paragraph{Hybrid 6.} In this hybrid, since the key $K_3$ has been punctured at $u'_1$, we can replace the evaluation of $F_3(K_3, \cdot)$ at $u'_1$ with a uniformly random value. The indistinguishability comes from the pseudorandomness of the underlying puncturable PRF $F_3$. 

\begin{enumerate}
    \item It samples random subspaces $A_i$ of dimension $\lambda/2$ and vectors $s_i, s'_i$ for $i = 1, 2, \cdots, \ell_0$. 
    It then prepares programs $R^0_i = \iO(A_i+s_i)$ and $R^1_i = \iO(A^\perp_i + s'_i)$ (padded to length $\ell_2 - \ell_0$). It prepares the quantum state $\ket \psi = \bigotimes_{i} \ket {A_{i, s_i, s'_i}}$. 
    
    \item It then samples keys $K_1, K_2, K_3$ for $F_1, F_2, F_3$.
    
    \item It samples $u = u_0||u_1||u_2$ uniformly at random. Let  $y \gets [M]$.
    
    \item It samples $u'$ as follows: 
        \begin{enumerate}
        \item $u'_1 \gets [2^{\ell_1}]$;
        \item {\color{red} $u'_2 \gets [2^{\ell_2}]$}.
        \item It outputs $u' = u_0||u'_1||u'_2$. 
    \end{enumerate}

    \item Generate the program $P$ as in Figure \ref{fig:program_P_3}. The adversary is given $(\ket \psi, \iO(P))$ and then $u$ or $u'$ depending on a random coin.
\end{enumerate}

In this hybrids, $u, u'$ are sampled independently, uniformly at random and they are symmetric in the program. The distributions for $b=0$ and $b=1$ are identical and even unbounded adversary can not distinguish these two cases. 
Therefore we finish the proof for \Cref{lem:hidden_trigger_1d}.

\begin{remark}
    The program $P$ depends on $Q_u$. Although $Q_u$ is indexed by $u$, it only depends on $u_0$. Thus, the distributions for $b=0$ and $b=1$ are identical
\end{remark}
\qed

\paragraph{Finishing the proof for \Cref{lem:hidden_trigger}.} The only difference between \Cref{lem:hidden_trigger} and \Cref{lem:hidden_trigger_1d} is the number of inputs sampled: either a single input $u$ (or $u'$) or a pair of independent inputs $u, w$ (or $u', w'$). 

All hybrids for \Cref{lem:hidden_trigger} are the same for the corresponding hybrids for \Cref{lem:hidden_trigger_1d}, except two inputs are sampled.  Thus every time $K_1, K_2$ or $K_3$ are punctured according to $u$ or $u'$ in the proof of \Cref{lem:hidden_trigger_1d}, $K_1, K_2$ or $K_3$ are punctured \emph{twice} according to both $u, u'$ and $w, w'$ in the proof of \Cref{lem:hidden_trigger}.

We are now giving the proof. If indistinguishability of some hybrid is not explained, it follows from the same reason as that in the corresponding hybrid in the proof of \Cref{lem:hidden_trigger_1d}.

\newpage
\paragraph{Hybrid 0.} The original game where both $u, w$ are sampled uniformly at random or $u', w'$ are random hidden trigger inputs. 

\begin{enumerate}
    \item It samples random subspaces $A_i$ of dimension $\lambda/2$ and vectors $s_i, s'_i$ for $i = 1, 2, \cdots, \ell_0$. 
    It then prepares programs $R^0_i = \iO(A_i+s_i)$ and $R^1_i = \iO(A^\perp_i + s'_i)$ (padded to the length upper bound $\ell_2 - \ell_0$). It prepares the quantum state $\ket \psi = \bigotimes_{i} \ket {A_{i, s_i, s'_i}}$. 
    
    \item It then samples keys $K_1, K_2, K_3$ for $F_1, F_2, F_3$.
    
    \item It samples $u = u_0||u_1||u_2$ uniformly at random. Let $y_u = F_1(K_1, u)$.
    
    \item It samples $u' \gets \gentrigger(u_0, y_u, K_2, K_3, \{A_i, s_i, s'_i\}_{i \in [\ell_0]})$.

    \item It samples $w = w_0||w_1||w_2$ uniformly at random. Let $y_w = F_1(K_1, w)$. 
    
    \item It samples $w' \gets \gentrigger(w_0, y_w, K_2, K_3, \{A_i, s_i, s'_i\}_{i \in [\ell_0]})$.
    
    \item Generate the program $P$ as in Figure \ref{fig:program_p_2d}. The adversary is given $(\ket \psi, \iO(P))$ and then $(u, w)$ or $(u', w')$ depending on a random coin $b$.
    
\end{enumerate}

\begin{figure}[hpt]
\centering
\begin{gamespec}
\textbf{Hardcoded:} Keys $K_1, K_2, K_3$, $R^0_i, R^1_i$ for all $i \in [\ell_0]$. 

On input $x = x_0 || x_1 || x_2$ and vectors $v_1, \cdots, v_{\ell_0}$:  
\begin{enumerate}
\item If $F_3(K_3, x_1) \oplus x_2 = x_0' || Q'$  and $x_0 = x'_0$ and $x_1 = F_2(K_2, x'_0 || Q')$: 

    \quad It treats $Q'$ as a circuit and outputs $Q'(v_1, \cdots, v_{\ell_0})$. 
\item Otherwise, it checks if the following holds: for all $i \in [\ell_0]$, $R^{x_{0, i}}(v_i) = 1$. 

    \quad If they all hold, outputs $F_1(K_1, x)$. Otherwise, outputs $\bot$. 
\end{enumerate}
\end{gamespec}
\caption{Program $P$}
\label{fig:program_p_2d}
\end{figure}

\newpage
\paragraph{Hybrid 1.} In this hybrid, the key $K_1$ in the program $P$ is punctured at $u, u', w, w'$.

\begin{enumerate}
    \item It samples random subspaces $A_i$ of dimension $\lambda/2$ and vectors $s_i, s'_i$ for $i = 1, 2, \cdots, \ell_0$. 
    It then prepares programs $R^0_i = \iO(A_i+s_i)$ and $R^1_i = \iO(A^\perp_i + s'_i)$ (padded to the length upper bound $\ell_2 - \ell_0$). It prepares the quantum state $\ket \psi = \bigotimes_{i} \ket {A_{i, s_i, s'_i}}$. 
    
    \item It then samples keys $K_1, K_2, K_3$ for $F_1, F_2, F_3$.
    
    \item It samples $u = u_0||u_1||u_2$ uniformly at random. Let $y_u = F_1(K_1, u)$.
    
    \item It samples $u' \gets \gentrigger(u_0, y_u, K_2, K_3, \{A_i, s_i, s'_i\}_{i \in [\ell_0]})$. 
    {\color{red} Let $Q_u$ be the obfuscation program during the execution of $\gentrigger$}.

    \item It samples $w = w_0||w_1||w_2$ uniformly at random. Let $y_w = F_1(K_1, w)$. 
    
    \item It samples $w' \gets \gentrigger(w_0, y_w, K_2, K_3, \{A_i, s_i, s'_i\}_{i \in [\ell_0]})$. 
    {\color{red} Let $Q_w$ be the obfuscation program during the execution of $\gentrigger$}.

    \item Generate the program $P$ as in Figure \ref{fig:program_p_2d_2}. The adversary is given $(\ket \psi, \iO(P))$ and then $(u, w)$ or $(u', w')$ depending on a random coin $b$.
    
\end{enumerate}

\begin{figure}[hpt]
\centering
\begin{gamespec}
\textbf{Hardcoded:} Keys ${\color{red} K_1 \setminus \{u, u', w, w'\}}, K_2, K_3$, $R^0_i, R^1_i$ for all $i \in [\ell_0]$. 

On input $x = x_0 || x_1 || x_2$ and vectors $v_1, \cdots, v_{\ell_0}$:  
\begin{enumerate}
\item {\color{red}If $x = u$ or $u'$, it outputs $Q_u(v_1, \cdots, v_{\ell_0})$. If $x = w$ or $w'$, it outputs $Q_w(v_1, \cdots, v_{\ell_0})$}. 

\item If $F_3(K_3, x_1) \oplus x_2 = x_0' || Q'$  and $x_0 = x'_0$ and $x_1 = F_2(K_2, x'_0 || Q')$: 

    \quad It treats $Q'$ as a circuit and outputs $Q'(v_1, \cdots, v_{\ell_0})$. 
\item Otherwise, it checks if the following holds: for all $i \in [\ell_0]$, $R^{x_{0, i}}(v_i) = 1$. 

    \quad If they all hold, outputs $F_1(K_1, x)$. Otherwise, outputs $\bot$. 
\end{enumerate}
\end{gamespec}
\caption{Program $P$}
\label{fig:program_p_2d_2}
\end{figure}

\newpage
\paragraph{Hybrid 2.} In this hybrid, $y_u$ and $y_w$ are sampled uniformly at random. Note that as long as $u \ne w$ (with overwhelming probability), we can apply the  pseudorandomness at punctured points of a puncturable PRF. 

\begin{enumerate}
    \item It samples random subspaces $A_i$ of dimension $\lambda/2$ and vectors $s_i, s'_i$ for $i = 1, 2, \cdots, \ell_0$. 
    It then prepares programs $R^0_i = \iO(A_i+s_i)$ and $R^1_i = \iO(A^\perp_i + s'_i)$ (padded to the length upper bound $\ell_2 - \ell_0$). It prepares the quantum state $\ket \psi = \bigotimes_{i} \ket {A_{i, s_i, s'_i}}$. 
    
    \item It then samples keys $K_1, K_2, K_3$ for $F_1, F_2, F_3$.
    
    \item It samples $u = u_0||u_1||u_2$ uniformly at random. {\color{red} Let $y_u \gets [M]$}.
    
    \item It samples $u' \gets \gentrigger(u_0, y_u, K_2, K_3, \{A_i, s_i, s'_i\}_{i \in [\ell_0]})$. 
    Let $Q_u$ be the obfuscation program during the execution of $\gentrigger$.

    \item It samples $w = w_0||w_1||w_2$ uniformly at random. {\color{red} Let $y_w \gets [M]$}. 
    
    \item It samples $w' \gets \gentrigger(w_0, y_w, K_2, K_3, \{A_i, s_i, s'_i\}_{i \in [\ell_0]})$. 
    Let $Q_w$ be the obfuscation program during the execution of $\gentrigger$.

    \item Generate the program $P$ as in Figure \ref{fig:program_p_2d_2}. The adversary is given $(\ket \psi, \iO(P))$ and then $(u, w)$ or $(u', w')$ depending on a random coin $b$.
    
\end{enumerate}

\newpage
\paragraph{Hybrid 3.} In this hybrid, $P$ is changed by checking if $x_1$ is equal to $u_1, u_1', w_1$ or $w_1'$. Moreover, $K_3$ is punctured at these points. 

\begin{enumerate}
    \item It samples random subspaces $A_i$ of dimension $\lambda/2$ and vectors $s_i, s'_i$ for $i = 1, 2, \cdots, \ell_0$. 
    It then prepares programs $R^0_i = \iO(A_i+s_i)$ and $R^1_i = \iO(A^\perp_i + s'_i)$ (padded to the length upper bound $\ell_2 - \ell_0$). It prepares the quantum state $\ket \psi = \bigotimes_{i} \ket {A_{i, s_i, s'_i}}$. 
    
    \item It then samples keys $K_1, K_2, K_3$ for $F_1, F_2, F_3$.
    
    \item It samples $u = u_0||u_1||u_2$ uniformly at random.  Let $y_u \gets [M]$.
    
    \item It samples $u' \gets \gentrigger(u_0, y_u, K_2, K_3, \{A_i, s_i, s'_i\}_{i \in [\ell_0]})$. 
    Let $Q_u$ be the obfuscation program during the execution of $\gentrigger$.

    \item It samples $w = w_0||w_1||w_2$ uniformly at random.  Let $y_w \gets [M]$. 
    
    \item It samples $w' \gets \gentrigger(w_0, y_w, K_2, K_3, \{A_i, s_i, s'_i\}_{i \in [\ell_0]})$. 
    Let $Q_w$ be the obfuscation program during the execution of $\gentrigger$.

    \item Generate the program $P$ as in Figure \ref{fig:program_p_2d_4}. The adversary is given $(\ket \psi, \iO(P))$ and then $(u, w)$ or $(u', w')$ depending on a random coin $b$.
\end{enumerate}

\begin{figure}[hpt]
\centering
\begin{gamespec}
\textbf{Hardcoded:} Keys $K_1 \setminus \{u, u', w, w'\}, K_2, {\color{red} K_3 \setminus \{u_1, u_1', w_1, w_1'\}}$, $R^0_i, R^1_i$ for all $i \in [\ell_0]$. 

On input $x = x_0 || x_1 || x_2$ and vectors $v_1, \cdots, v_{\ell_0}$:  
\begin{enumerate}
\item If $x = u$ or $u'$, it outputs $Q_u(v_1, \cdots, v_{\ell_0})$. If $x = w$ or $w'$, it outputs $Q_w(v_1, \cdots, v_{\ell_0})$. 

\item {\color{red} If $x_1$ is equal to $u_1, u_1', w_1$ or $w_1'$, then skip this check.} If $F_3(K_3, x_1) \oplus x_2 = x_0' || Q'$  and $x_0 = x'_0$ and $x_1 = F_2(K_2, x'_0 || Q')$: 

    \quad It treats $Q'$ as a circuit and outputs $Q'(v_1, \cdots, v_{\ell_0})$. 
\item Otherwise, it checks if the following holds: for all $i \in [\ell_0]$, $R^{x_{0, i}}(v_i) = 1$. 

    \quad If they all hold, outputs $F_1(K_1, x)$. Otherwise, outputs $\bot$. 
\end{enumerate}
\end{gamespec}
\caption{Program $P$}
\label{fig:program_p_2d_4}
\end{figure}

\newpage
\paragraph{Hybrid 4.} In this hybrid, before checking $x_1 = F_2(K_2, x'_0 || Q')$, it checks if $x'_0||Q' \ne u_0||Q_u$ and $x'_0||Q' \ne w_0||Q_w$.
Because if $x'_0||Q' = u_0||Q_u$ and the last check $x_1 = F_2(K_2, x_0'||Q')$ is also satisfied, we know that 
\begin{align*}
x_1 = F_2(K_2, x'_0||Q') = F_2(K_2, u_0||Q_u) = u_1' \,\,\text{ (by the definition of $\gentrigger$).}     
\end{align*}
Therefore the step 2 will be skipped (by the first check). Similarly, if $x'_0||Q' = w_0||Q_w$ and the last check $x_1 = F_2(K_2, x_0'||Q')$ is also satisfied, we know that 
\begin{align*}
x_1 = F_2(K_2, x'_0||Q') = F_2(K_2, w_0||Q_w) = w_1' \,\,\text{ (by the definition of $\gentrigger$).}
\end{align*}
Finally, we puncture $K_2$ at $u_0 ||Q_u$ and $w_0||Q_w$.

\begin{enumerate}
    \item It samples random subspaces $A_i$ of dimension $\lambda/2$ and vectors $s_i, s'_i$ for $i = 1, 2, \cdots, \ell_0$. 
    It then prepares programs $R^0_i = \iO(A_i+s_i)$ and $R^1_i = \iO(A^\perp_i + s'_i)$ (padded to the length upper bound $\ell_2 - \ell_0$). It prepares the quantum state $\ket \psi = \bigotimes_{i} \ket {A_{i, s_i, s'_i}}$. 
    
    \item It then samples keys $K_1, K_2, K_3$ for $F_1, F_2, F_3$.
    
    \item It samples $u = u_0||u_1||u_2$ uniformly at random.  Let $y_u \gets [M]$.
    
    \item It samples $u' \gets \gentrigger(u_0, y_u, K_2, K_3, \{A_i, s_i, s'_i\}_{i \in [\ell_0]})$. 
    Let $Q_u$ be the obfuscation program during the execution of $\gentrigger$.

    \item It samples $w = w_0||w_1||w_2$ uniformly at random.  Let $y_w \gets [M]$. 
    
    \item It samples $w' \gets \gentrigger(w_0, y_w, K_2, K_3, \{A_i, s_i, s'_i\}_{i \in [\ell_0]})$. 
    Let $Q_w$ be the obfuscation program during the execution of $\gentrigger$.

    \item Generate the program $P$ as in Figure \ref{fig:program_p_2d_5}. The adversary is given $(\ket \psi, \iO(P))$ and then $(u, w)$ or $(u', w')$ depending on a random coin $b$.
\end{enumerate}

\begin{figure}[hpt]
\centering
\begin{gamespec}
\textbf{Hardcoded:} Keys $K_1 \setminus \{u, u', w, w'\}, {\color{red} K_2 \setminus\{u_0||Q_u, w_0||Q_w\}}, {K_3 \setminus \{u_1, u_1', w_1, w_1'\}}$, $R^0_i, R^1_i$ for all $i \in [\ell_0]$. 

On input $x = x_0 || x_1 || x_2$ and vectors $v_1, \cdots, v_{\ell_0}$:  
\begin{enumerate}
\item If $x = u$ or $u'$, it outputs $Q_u(v_1, \cdots, v_{\ell_0})$. If $x = w$ or $w'$, it outputs $Q_w(v_1, \cdots, v_{\ell_0})$. 

\item If $x_1$ is equal to $u_1, u_1', w_1$ or $w_1'$, then skip this check. If $F_3(K_3, x_1) \oplus x_2 = x_0' || Q'$  and $x_0 = x'_0$ and {\color{red} $x_0'||Q' \ne u_0 ||Q_u$ and  $x_0'||Q' \ne w_0 ||Q_w$ } and $x_1 = F_2(K_2, x'_0 || Q')$: 

    \quad It treats $Q'$ as a circuit and outputs $Q'(v_1, \cdots, v_{\ell_0})$. 
\item Otherwise, it checks if the following holds: for all $i \in [\ell_0]$, $R^{x_{0, i}}(v_i) = 1$. 

    \quad If they all hold, outputs $F_1(K_1, x)$. Otherwise, outputs $\bot$. 
\end{enumerate}
\end{gamespec}
\caption{Program $P$}
\label{fig:program_p_2d_5}
\end{figure}

\newpage
\paragraph{Hybrid 5.}  In this hybrid, since the key $K_2$ has been punctured at $u_0||Q_u$ and $w_0 || Q_w$, we can replace the evaluation of $F_2(K_2, \cdot)$ at these two inputs with uniformly random values, as long as $u_0 \ne w_0$ (with overwhelming probability). The indistinguishability comes from the pseudorandomness of the underlying puncturable PRF $F_2$. 

\begin{enumerate}
    \item It samples random subspaces $A_i$ of dimension $\lambda/2$ and vectors $s_i, s'_i$ for $i = 1, 2, \cdots, \ell_0$. 
    It then prepares programs $R^0_i = \iO(A_i+s_i)$ and $R^1_i = \iO(A^\perp_i + s'_i)$ (padded to the length upper bound $\ell_2 - \ell_0$). It prepares the quantum state $\ket \psi = \bigotimes_{i} \ket {A_{i, s_i, s'_i}}$. 
    
    \item It then samples keys $K_1, K_2, K_3$ for $F_1, F_2, F_3$.
    
    \item It samples $u = u_0||u_1||u_2$ uniformly at random.  Let $y_u \gets [M]$.
    
    \item It samples $u' \gets \gentrigger(u_0, y_u, K_2, K_3, \{A_i, s_i, s'_i\}_{i \in [\ell_0]})$ as follows:
    
        \begin{enumerate}
        \item Let $Q_u$ be the obfuscation of the program (padded to length $\ell_2 - \ell_0$) that takes inputs $v_1, \cdots, v_{\ell_0}$ and outputs $y_u$ if and only if for every input $v_i$, if $u_{0, i} = 0$, then $v_i$ is in $A_i + s_i$ and otherwise it is in $A^\perp_i + s'_i$. 
        \item {\color{red} $u'_1 \gets [2^{\ell_1}]$} (since $F_2(K_2, u_0||Q_u)$ has been replaced with a uniformly random value).
        \item $u'_2 \gets F_3(K_3, u'_1) \oplus (u_0 ||Q_u)$.
        \item It outputs $u' = u_0||u'_1||u'_2$.
    \end{enumerate}

    \item It samples $w = w_0||w_1||w_2$ uniformly at random.  Let $y_w \gets [M]$. 
    
    \item It samples $w' \gets \gentrigger(w_0, y_w, K_2, K_3, \{A_i, s_i, s'_i\}_{i \in [\ell_0]})$ as follows: 
    
    \begin{enumerate}
        \item Let $Q_w$ be the obfuscation of the program (padded to length $\ell_2 - \ell_0$) that takes inputs $v_1, \cdots, v_{\ell_0}$ and outputs $y_w$ if and only if for every input $v_i$, if $w_{0, i} = 0$, then $v_i$ is in $A_i + s_i$ and otherwise it is in $A^\perp_i + s'_i$. 
        \item {\color{red} $w'_1 \gets [2^{\ell_1}]$} (since $F_2(K_2, w_0||Q_w)$ has been replaced with a uniformly random value).
        \item $w'_2 \gets F_3(K_3, w'_1) \oplus (w_0 ||Q_w)$.
        \item It outputs $w' = w_0||w'_1||w'_2$. 
    \end{enumerate}

    \item Generate the program $P$ as in Figure \ref{fig:program_p_2d_5}. The adversary is given $(\ket \psi, \iO(P))$ and then $(u, w)$ or $(u', w')$ depending on a random coin $b$.
\end{enumerate}

\newpage
\paragraph{Hybrid 6.}  In this hybrid, since the key $K_3$ has been punctured at $u'_1, w'_1$, we can replace the evaluation of $F_3(K_3, \cdot)$ at $u'_1, w'_1$ with uniformly random values (as long as $u'_1 \ne w'_1$, which happens with overwhelming probability). The indistinguishability comes from the pseudorandomness of the underlying puncturable PRF $F_3$.

\begin{enumerate}
    \item It samples random subspaces $A_i$ of dimension $\lambda/2$ and vectors $s_i, s'_i$ for $i = 1, 2, \cdots, \ell_0$. 
    It then prepares programs $R^0_i = \iO(A_i+s_i)$ and $R^1_i = \iO(A^\perp_i + s'_i)$ (padded to the length upper bound $\ell_2 - \ell_0$). It prepares the quantum state $\ket \psi = \bigotimes_{i} \ket {A_{i, s_i, s'_i}}$. 
    
    \item It then samples keys $K_1, K_2, K_3$ for $F_1, F_2, F_3$.
    
    \item It samples $u = u_0||u_1||u_2$ uniformly at random.  Let $y_u \gets [M]$.
    
    \item It samples $u'$ as follows:
    
        \begin{enumerate}
        \item {$u'_1 \gets [2^{\ell_1}]$}.
        \item {\color{red} $u'_2 \gets [2^{\ell_2}]$}.
        \item It outputs $u' = u_0||u'_1||u'_2$.
    \end{enumerate}

    \item It samples $w = w_0||w_1||w_2$ uniformly at random.  Let $y_w \gets [M]$.
    
    \item It samples $w'$ as follows:
    
    \begin{enumerate}
        \item $w'_1 \gets [2^{\ell_1}]$.
        \item {\color{red} $w'_2 \gets [2^{\ell_2}]$}.
        \item It outputs $w' = w_0||w'_1||w'_2$.
    \end{enumerate}

    \item Generate the program $P$ as in Figure \ref{fig:program_p_2d_5}. The adversary is given $(\ket \psi, \iO(P))$ and then $(u, w)$ or $(u', w')$ depending on a random coin $b$.
\end{enumerate}

In this hybrids, $u, u'$, $w, w'$ are sampled independently, uniformly at random and they are symmetric in the program. The distributions for $b=0$ and $b=1$ are identical and even unbounded adversary can not distinguish these two cases. 
Therefore we finish the proof for \Cref{lem:hidden_trigger}.  

\begin{remark}
    The program $P$ depends on $Q_u$ and $Q_w$. Although $Q_u$ and $Q_w$ are indexed by $u$ and $w$, they only depend on $u_0, w_0$ respectively. Thus, the distributions for $b=0$ and $b=1$ are identical
\end{remark}

\qed

\section{Proof of \texorpdfstring{\Cref{thm:PRF_indistinguishable_antipiracy}}{Indistinguishability Anti-Piracy Security for PRF}}
\label{sec:indistinguishable_anti_piracy_PRF}

The proof for \Cref{thm:PRF_indistinguishable_antipiracy} is similar to proof for \Cref{thm:PRF_antipiracy}, but has some main differences in the final reduction.
We highlight the changes made: the {\color{red}red-colored} parts are the differences between the latter hybrid and  the former hybrid; the {\color{blue} blue-colored}  parts are differences between original hybrids in proof for \Cref{thm:PRF_antipiracy} and these new hybrids.

     \paragraph{Hybrid 0.} 
     Hybrid 0 is the original anti-piracy indistinguihsability security game.
     
    \begin{enumerate}
        \item  A challenger samples $K_1 \gets \setup(1^\lambda)$ and prepares a quantum key $\rho_K = (\{\ket{A_{i, s_i, s'_i}} \}_{i \in [\ell_0]}, \iO(P))$. Here $P$ hardcodes $K_1, K_2, K_3$.

        \item $\As$ upon receiving $\rho_K$, it runs and prepares a pair of (potentially entangled) quantum states $\sigma[R_1], \sigma[R_2]$ as well as quantum circuits $U_1, U_2$. 
        \item The challenger also prepares two inputs $u, w$ as follows:
            \begin{itemize}
                \item It samples $u$ uniformly at random. 
                Let $y_u  = F_1(K_1, u)$.

                \item It samples $w$ uniformly at random.
                Let $y_w = F_1(K_1, w)$.

            \end{itemize}
        \item {\color{blue}It samples  $y_u', y_w'$  uniformly at random}. 
        
        \item {\color{blue}The challenger then samples uniform coins $b_0, b_1 \gets \{0,1\}$. If $b_0 = 0$, give $(u, y_u)$ to quantum program $(U_1, \sigma[R_1])$; else give $(u, y_u')$ to quantum program $(U_1, \sigma[R_1])$.
        Similarly, if $b_1 = 0$, give $(w, y_w)$ to quantum program $(U_2, \sigma[R_2])$; else give $(w, y_w')$}.
        
        \item The outcome of the game is 1 if and only if both quantum programs successfully {\color{blue}produce $b_0' = b_0$ and $b_1' = b_1$} respectively. 
    \end{enumerate}

    \paragraph{Hybrid 1} 
    The changes between Hybrid 0 and 1 are exactly the two cases the adversary needs to distinguish between in the game of \Cref{lem:hidden_trigger}.
      Assume there exists an algorithm that  distinguishes Hybrid 0 and 1 with \emph{non-negligible} probability $\epsilon(\lambda)$, then these exists
      an algorithm that breaks the game in \Cref{lem:hidden_trigger} with probability $\epsilon(\lambda) - \negl(\lambda)$.

     The reduction algorithm receives $\rho_k$ and $u, w$ or $u', w'$ from the challenger in \Cref{lem:hidden_trigger}; it computes $y_u, y_w$ using $\iO(P)$ on the received inputs respectively and gives them to the quantum decryptor states $\sigma[R_1], \sigma[R_2]$. If they both {\color{blue} output the guess }correctly, then the reduction outputs 0 for $u, w$, otherwise it outputs 1 for  $u', w'$.

    \begin{enumerate}
        \item  A challenger samples $K_1 \gets \setup(1^\lambda)$ and prepares a quantum key $\rho_K = (\{\ket{A_{i, s_i, s'_i}} \}_{i \in [\ell_0]}, \iO(P))$. Here $P$ hardcodes $K_1, K_2, K_3$.

        \item $\As$ upon receiving $\rho_K$, it runs and prepares a pair of (potentially entangled) quantum states $\sigma[R_1], \sigma[R_2]$ as well as quantum circuits $U_1, U_2$. 
        \item The challenger also prepares two inputs $u', w'$ as follows:
            \begin{itemize}
                \item It samples $u = u_0||u_1||u_2$ uniformly at random. 
                Let $y_u  = F_1(K_1, u)$. 
                
                {\color{red}
                Let $u' \gets \gentrigger(u_0, y_u, K_2, K_3, \{A_i, s_i, s'_i\}_{i \in [\ell_0]})$. }

                \item It samples $w = w_0||w_1||w_2$ uniformly at random.
                Let $y_w = F_1(K_1, w)$. 
                
                {\color{red}
                Let $w' \gets \gentrigger(w_0, y_w, K_2, K_3, \{A_i, s_i, s'_i\}_{i \in [\ell_0]})$. }

            \end{itemize}
            
         \item {\color{blue}It samples  $y_u', y_w'$  uniformly at random}. 
        
        \item {\color{blue}The challenger then samples uniform coins $b_0, b_1 \gets \{0,1\}$. If $b_0 = 0$, give $(u', y_u)$ to quantum program $(U_1, \sigma[R_1])$; else give $(u', y_u')$ to quantum program $(U_1, \sigma[R_1])$.
        Similarly, if $b_1 = 0$, give $(w', y_w)$ to quantum program $(U_2, \sigma[R_2])$; else give $(w', y_w')$}.
        
        \item The outcome of the game is 1 if and only if both quantum programs successfully {\color{blue}produce $b_0' = b_0$ and $b_1' = b_1$} respectively. 
    \end{enumerate}

    \paragraph{Hybrid 2.}  In this hybrid, if $u_0 \ne w_0$ (which happens with overwhelming probability),  $F_1(K_1, u)$ and $F_1(K_1, w)$ can be replaced with truly random strings. Since both inputs have enough min-entropy $\ell_1 + \ell_2 \geq m + 2\lambda  + 4$ (as $u_1||u_2$ and $w_1||w_2$ are completely uniform and not given to the adversary) and $F_1$ is an extracting puncturable PRF, both outcomes $y_u, y_w$ are statistically close to independently random outcomes.
    Thus, Hybrid 1 and Hybrid 2 are statistically close. 
    
    \begin{enumerate}
        \item  A challenger samples $K_1 \gets \setup(1^\lambda)$ and prepares a quantum key $\rho_K = (\{\ket{A_{i, s_i, s'_i}} \}_{i \in [\ell_0]}, \iO(P))$. Here $P$ hardcodes $K_1, K_2, K_3$.

        \item $\As$ upon receiving $\rho_K$, it runs and prepares a pair of (potentially entangled) quantum states $\sigma[R_1], \sigma[R_2]$ as well as quantum circuits $U_1, U_2$. 
        \item The challenger also prepares two inputs $u', w'$ as follows:
            \begin{itemize}
                \item It samples {\color{red} $u_0$ uniformly at random}. It then samples {\color{red} a uniformly random $y_u$}. 
                
                Let $u' \gets \gentrigger(u_0, y_u, K_2, K_3, \{A_i, s_i, s'_i\}_{i \in [\ell_0]})$.

                \item It samples {\color{red} $w_0$ uniformly at random}. It then samples {\color{red} a uniformly random $y_w$}. 
                
                Let $w' \gets \gentrigger(w_0, y_w, K_2, K_3, \{A_i, s_i, s'_i\}_{i \in [\ell_0]})$. 

            \end{itemize}
            
         \item {\color{blue}It samples  $y_u', y_w'$  uniformly at random}. 
        
        \item {\color{blue}The challenger then samples uniform coins $b_0, b_1 \gets \{0,1\}$. If $b_0 = 0$, give $(u, y_u)$ to quantum program $(U_1, \sigma[R_1])$; else give $(u, y_u')$ to quantum program $(U_1, \sigma[R_1])$.
        Similarly, if $b_1 = 0$, give $(w, y_w)$ to quantum program $(U_2, \sigma[R_2])$; else give $(w, y_w')$}.
        
        \item The outcome of the game is 1 if and only if both quantum programs successfully {\color{blue}produce $b_0' = b_0$ and $b_1' = b_1$} respectively. 
    \end{enumerate}

  \paragraph{Hybrid 3.} %

        \begin{enumerate}
            \item {\color{red} A challenger first samples $\{A_i, s_i, s'_i\}_{i \in [\ell_0]}$ and prepares the quantum states $\{\ket{A_{i, s_i, s'_i}} \}_{i \in [\ell_0]}$}. It treat the the quantum states $\{\ket{A_{i, s_i, s'_i}} \}_{i \in [\ell_0]}$ as the quantum decryption key $\rho_{\sk}$ for our single-decryptor encryption scheme and the secret key $\sk$ is $\{A_i, s_i, s'_i\}_{i \in [\ell_0]}$.  Similarly, let $\pk = \{R^0_i, R^1_i\}_{i \in [\ell_0]}$ where $R^0_i = \iO(A_i+s_i)$ and $R^1_i = \iO(A^\perp_i + s'_i)$. 
        
            \item  It samples $y_u, y_w$  uniformly at random. {\color{blue}It also samples  $y_u', y_w'$  uniformly at random}.
            
            \item {\color{blue} Then it flips two random coins $b_0, b_1 \gets \{0,1\}$. If $b_0 = 1$, let $(u_0, Q_0) \gets \enc(\pk, y_u)$; else let $(u_0, Q_0) \gets \enc(\pk, y_u')$.
            Similarly, if $b_1 = 1$,  let $(w_0, Q_1) \gets \enc(\pk, y_w)$; else, let $(w_0, Q_1) \gets \enc(\pk, y_w')$}.
            $\enc(\pk, \cdot)$ is the encryption algorithm of the underlying single-decryptor encryption scheme using $\pk$.

            \item The challenger constructs the program $P$ which hardcodes ${K_1}, K_2, K_3$.  It then prepares $\rho_K$, which is  $(\{\ket{A_{i, s_i, s'_i}} \}_{i \in [\ell_0]}, \iO(P))$.
            
            \item $\As$ upon receiving $\rho_K$, it runs and prepares a pair of (potentially entangled) quantum states $\sigma[R_1], \sigma[R_2]$ as well as quantum circuits $U_1, U_2$. 
            \item The challenger also prepares two inputs $u', w'$ as follows (as $\gentrigger$ does):
                \begin{itemize}
                    \item 
                    Let $u_1' \gets F_2(K_2, u_0 || Q_0)$ and $u_2'\gets F_3(K_3, u_1') \oplus (u_0||Q_0)$. Let $u' = u_0 || u_1' || u_2'$. 
                    
                    \item 
                    Let $w_1' \gets F_2(K_2, w_0 || Q_1)$ and $w_2'\gets F_3(K_3, w_1') \oplus (w_0||Q_1)$. Let $w' = w_0 || w_1' || w_2'$. 
                \end{itemize}
        
        \item {\color{blue}The challenger again samples uniform coins $\delta_0, \delta_1 \gets \{0,1\}$. If $\delta_0 = 0$, give $(u', y_u)$ to quantum program $(U_1, \sigma[R_1])$; else give $(u', y_u')$ to quantum program $(U_1, \sigma[R_1])$.
        Similarly, if $\delta_1 = 0$, give $(w', y_w)$ to quantum program $(U_2, \sigma[R_2])$; else give $(w', y_w')$}.

        \item The outcome of the game is 1 if both quantum programs successfully {\color{blue}produce the answers below respectively:
        \begin{itemize}
            \item If $(U_1, \sigma[R_1])$ outputs 0 and $b_0 = \delta_0$, or if it outputs 1 and $b_0 \neq \delta_0$, then $(U_1, \sigma[R_1])$ succeeds. Otherwise it fails.
            \item If $(U_2, \sigma[R_2])$ outputs 0 and $b_1 = \delta_1$, or if it outputs 1 and $b_1 \neq \delta_1$, then $(U_2, \sigma[R_2])$ succeeds. Otherwise it fails.
        \end{itemize} }

        \end{enumerate}

       Note that the only differences of Hyb 2 and Hyb 3 are the orders of executions and {\color{blue} that the challenger prepares $u', w'$ from one of $(y_u, y_u')$ and one of $(y_w, y_w')$ respectively, instead of preparing them from $y_u, y_w$ first and choosing random $y_u', y_w'$ later}. Since we will check if the random coins match in Step 8 of Hybrid 3, the game is essentially the same to an adversary as the game in Hybrid 2.

        \vspace{1em}
    
        Given an algorithm $\As$ that wins the indistinguishability anti-piracy game for PRF in Hybrid 3 with non-negligible probability $\gamma(\lambda)$, we can build another algorithm $\Bs$ that breaks the  {\color{blue} (regular) CPA-style $\gamma$-anti-piracy security  (see \Cref{def:weak_ag})} of the underlying single-decryptor encryption scheme. 
        \begin{itemize}
            \item $\Bs$ plays as the challenger in the game of Hybrid 3. 
            \item $\Bs$ will get  $\rho_{\sk} = \{\ket{A_{i, s_i, s'_i}} \}_{i \in [\ell_0]}$ and $\pk = \{\iO(A_i + s_i), \iO(A^\perp_i + s'_i)\}_{i \in [\ell_0]}$ in the anti-piracy game. 
            
            \item $\Bs$ prepares $K_1, K_2, K_3$ and the program $P$. Let $\rho_K = (\{\ket{A_{i, s_i, s'_i}} \}_{i \in [\ell_0]}, \iO(P))$.
            
            \item $\Bs$ gives $\rho_{K}$ to $\As$ and $\As$ prepares a pair of (potentially entangled) quantum states $\sigma[R_1], \sigma[R_2]$ as well as quantum circuits $U_1, U_2$.  
            
             \item  {\color{blue}$\Bs$ also samples uniform random $y_u, y_u', y_w, y_w'$ and sends $(y_u, y_u')$ and $(y_w, y_w')$ as the challenge plaintexts for the two quantum programs, to the challenger of single-decryptor encryption anti-piracy game}.

             \item $\Bs$ then creates quantum programs $\P_1, \P_2$ which will do the following steps.

            \item $\P_1$ receives challenge ciphertext $\ct_0 = (u_0, Q_0) $ (which will be encryption of either $y_u$ or $y_u'$), and $\P_2$ receives challenge ciphertext $\ct_1 = (w_0, Q_1)$ (which will be encryption of either $y_w$ or $y_w'$). They each independently prepares $u', w'$ as follows :
            \begin{itemize}
                    \item 
                    Let $u_1' \gets F_2(K_2, u_0 || Q_0)$ and $u_2'\gets F_3(K_3, u_1') \oplus (u_0||Q_0)$. Let $u' = u_0 || u_1' || u_2'$. 
                    
                    \item 
                    Let $w_1' \gets F_2(K_2, w_0 || Q_1)$ and $w_2'\gets F_3(K_3, w_1') \oplus (w_0||Q_1)$. Let $w' = w_0 || w_1' || w_2'$. 
                \end{itemize}
                
           \item {\color{blue} $\P_1$ gives either $(u', y_u)$ or $(u', y_u')$ depending on  a random coin $\delta_0 \gets \{0,1\}$, to $(\sigma[R_1], U_{1})$; 
           $\P_2$ gives either $(w', y_w)$ or $(w', y_w')$ depending on random coin $\delta_1 \gets \{0,1\}$, to $(\sigma[R_2], U_{2})$}.

            \item Then $\P_1$ and $\P_2$ respectively run $(\sigma[R_1], U_{1})$ and $(\sigma[R_2], U_{2})$ on their challenge received to output their answers $a_1$ and $a_2$. 

      {\color{blue}  \item Finally, depending on answers received and the coins $\delta_0, \delta_1$, $\P_1$ and $\P_2$  does the following:
         
          For $\P_1$: if $(\sigma[R_1], U_{1})$ outputs $0$(which means the program thinks it receives an input and its PRF evaluation):
                \begin{itemize}
                    \item if $\delta_0 = 0$: $\P_1$ outputs 0 (for encryption of $y_u$) to the challenger.
                    \item else, $\delta_0 = 1$: $\P_1$ outputs 1 (for  encryption of $y_u'$) to the challenger.
                \end{itemize}
        If  $(\sigma[R_1], U_{1})$ outputs $1$ (which means the program thinks it receives an input and a random value):
                \begin{itemize}
                    \item if $\delta_0 = 0$: $\P_1$ outputs 1 (for encryption of $y_u'$) to the challenger.
                    \item else, $\delta_0 = 1$: $\P_1$ outputs 0 (for  encryption of $y_u$) to the challenger.
                \end{itemize} 
           
           Similarly on the $\P_2$ and $(\sigma[R_2], U_2)$ side}.
            
        \end{itemize}
        
        We observe that the advantage of $\Bs$ in the CPA-style $\gamma$-anti-piracy game of single-decryptor encryption is the same as advantage of $\As$ in the indistinguishability anti-piracy game for PRF. \qed

\end{document}